\documentclass[11pt,a4paper]{article}

\usepackage[margin=1in]{geometry}

\usepackage{theorem,latexsym,graphicx,amssymb}
\usepackage{amsmath,enumerate}
\usepackage{wrapfig}
\usepackage{subfigure}
\usepackage{float}
\usepackage{graphicx}
\usepackage{epsfig}
\usepackage{xspace}
\usepackage{paralist}
\usepackage{times}
\usepackage[compact]{titlesec}

\usepackage{algorithmic}
\usepackage[linesnumbered,ruled,vlined,commentsnumbered]{algorithm2e}

\usepackage{enumerate}
\usepackage{cases}
\usepackage{caption}
\usepackage{url}
\usepackage{tikz}
\usetikzlibrary{arrows}
\usetikzlibrary{topaths,calc}
\usepackage{booktabs}
\usepackage{mathtools}

\usepackage{comment}

\usepackage{color}


\usepackage{multicol}

\newenvironment{proof}{{\bf Proof:  }}{\hfill\rule{2mm}{2mm}}

\numberwithin{figure}{section}
\numberwithin{equation}{section}

\newtheorem{theorem}{Theorem}[section]
\newtheorem{definition}[theorem]{Definition}

\newtheorem{corollary}[theorem]{Corollary}
\newtheorem{lemma}[theorem]{Lemma}

\newtheorem{proposition}{\hskip\parindent Proposition}[section]

\newcommand{\R}{\mathbb{R}}

\newcounter{note}[section]
\renewcommand{\thenote}{\thesection.\arabic{note}}

\newcommand{\initOneLiners}{%
    \setlength{\itemsep}{0pt}
    \setlength{\parsep }{0pt}
    \setlength{\topsep }{0pt}
}

\newcommand*\samethanks[1][\value{footnote}]{\footnotemark[#1]}

\newcommand{\shortv}[1]{}

\newcommand{\Diag}{{\sf Diag}}
\newcommand{\one}{\ensuremath{\mathbf{1}}} 
\newcommand{\mcal}{\mathcal}
\newcommand{\err}{\mathsf{err}}

\newcommand{\NPhard}{{\ensuremath{\mathbf{NP}}-hard}\xspace}

\newcommand{\ignore}[1]{}
\newcommand{\hide}[1]{}

\newcommand{\hubert}[1]{\refstepcounter{note}$\ll${\sf Hubert's
        Comment~\thenote:} {\sf \textcolor{blue}{#1}}$\gg$\marginpar{\tiny\bf HC~\thenote}}

\usepackage{enumerate}
\usepackage{diagbox}
\usepackage[most]{tcolorbox}

\newtheorem{lm}{Lemma}
\theoremstyle{definition}

\newtheorem{qq}[lm]{Problem}

\title{Opinion Dynamics with Varying Susceptibility to Persuasion via Non-Convex Local Search}

\author{
Rediet Abebe\thanks{UC Berkeley} \and
T-H. Hubert Chan\thanks{The University of Hong Kong} \and
Jon Kleinberg\thanks{Cornell University} \and
Zhibin Liang\samethanks[2] \and
David Parkes\thanks{Harvard University} \and
Mauro Sozio\thanks{LTCI, T\'el\'ecom ParisTech University} \and
Charalampos E. Tsourakakis\thanks{Boston University and ISI Foundation}
}
\date{}

\begin{document}

\begin{titlepage}

\maketitle

\begin{abstract}
A long line of work in social psychology has studied variations in people's susceptibility to persuasion -- the extent to which they are willing to modify their opinions on a topic.  This body of literature suggests an interesting perspective on theoretical models of opinion formation by interacting parties in a network: in addition to considering interventions that directly modify people's intrinsic opinions, it is also natural to consider interventions that modify people's susceptibility to persuasion.  

In this work, motivated by this fact we propose a new framework for social influence.  Specifically, we adopt a popular model for social opinion dynamics,  where each agent has some fixed innate opinion, and a resistance that measures the importance it places on its innate opinion; agents influence one another's opinions through an iterative process. Under non-trivial conditions, this iterative process converges to some equilibrium opinion vector.  For the unbudgeted variant of the problem, the goal is to select the resistance of each agent (from some given range) such that the sum of the equilibrium opinions is minimized.

We prove that the objective function is in general non-convex. Hence, formulating the problem as a convex program as in an early version of this work (Abebe et al., KDD'18) might have potential correctness issues.  We instead analyze the structure of the objective function, and show that any local optimum is also a global optimum, which is somehow surprising as the objective function might not be convex. Furthermore, we combine the iterative process and the local search paradigm to design very efficient algorithms that can solve the unbudgeted variant of the problem optimally on large-scale graphs containing millions of nodes. Finally, we propose and evaluate experimentally a family of heuristics for the budgeted variation of the problem. 
\end{abstract}

\thispagestyle{empty}
\end{titlepage}

\section{Introduction}
\label{sec:intro}

A rich line of empirical work in development and social psychology has studied people's susceptibility to persuasion. This property measures the extent to which individuals are willing to modify their opinions in reaction to the opinions expressed by those around them,
and it is distinct from the opinions they express. Research in the area has ranged from adolescent susceptibility to peer pressure related to risky and antisocial behavior \cite{allen2006,dielman1992susceptibility,evans1992measuring,Steinberg2007,
WakefieldFNG2003} to the role of susceptibility to persuasion in politics \cite{Fransen2015,murphy2004persuasion,Neiheisel2015}. Individuals' susceptibility to persuasion 
can be affected by specific strategies and framings aimed at increasing susceptibility \cite{Cialdini1993,Cialdini2001,Crowley1994,Kaptein2009,mcguire1964inducing,
tormala2004source, tormala2002doesn}. For 
instance, if it is known that an individual is receptive to \emph{persuasion by authority}, one can adopt a strategy that utilizes arguments from official sources and authority figures 
to increase that individuals' susceptibility to persuasion with respect to a particular topic.  

Modifying network opinions has far-reaching implications including product marketing, public health campaigns, the success of political candidates, and public opinions on issues of global 
interest. In recent years, there has also been work in Human Computer Interaction focusing on 
{\em persuasive technologies}, which are designed with the goal of changing a person's attitude or behavior \cite{Fogg2002,IJsselsteijn2006,Kaptein2009}. This work has shown that not only do people differ in their susceptibility to persuasion, but that persuasive technologies can also be adapted 
to each individual to change their susceptibility to persuasion.  Despite the long line of empirical work emphasizing the importance of
individuals' susceptibility to persuasion, to our knowledge
theoretical studies of opinion formation models have not focused on
interventions at the level of susceptibility.  Social influence studies have considered interventions that  directly act on the opinions themselves, both in discrete models (e.g., 
\cite{Domingos2001,Kempe2003,AbebeAK18,Berger2007,Hu2014,JoshiRZ16})
and more recently in   continuous models \cite{GionisTT13,Musco2018}.

In this work, we adopt an opinion formation model inspired by the work of DeGroot \cite{DeGroot1974} and Friedkin and Johnsen \cite{Friedkin1999}, and we initiate a study of the impact of interventions at the level of susceptibility. In this model, each agent $i$ is endowed with an {\em innate opinion} $s_i \in [0,1]$ and a parameter representing susceptibility to persuasion, which we will call 
the {\em resistance parameter}, $\alpha_i \in (0,1]$. 
The innate opinion $s_i$ reflects the intrinsic position of agent $i$
on a certain topic, while  $\alpha_i$ reflects the agent's willingness, 
or lack thereof, to conform with the opinions of neighbors in the social 
network.We term $\alpha_i$ the agent's ``resistance'' because a high value of $\alpha_i$ corresponds to a lower tendency 
to conform with neighboring opinions. According to the opinion dynamics model, the final opinion of each
agent $i$ is a function of the social network, the set of innate
opinions, and the resistance parameters, determined by computing
the {\em equilibrium} state of a dynamic process of opinion updating. We study the following natural question:

\begin{tcolorbox}
\begin{qq}
\label{prob1} 
Given an opinion dynamics model, and a set of agents, each of whom has an innate opinion that reflects the agent's intrinsic position on a topic, and a range for the resistance parameter measuring the agent's propensity for changing their opinion, how should we set the agents' resistance parameter in order to minimize the total sum of opinions at equilibrium?
\end{qq}
\end{tcolorbox}

Observe that the problem is trivial if the resistance of each
agent can be picked from the closed interval~$[0,1]$. For minimizing the equilibrium opinions, it suffices to make the agent with the minimum innate opinion the most resistant (setting its resistance to 1)
and everyone else totally compliant (setting its resistance to 0). Similarly, for the maximization problem it suffices to make the agent with the maximum innate opinion the most resistant, and the rest of the nodes totally compliant.  The problem is non-trivial if
the resistance $\alpha_i$ of each agent~$i$ can take value from some interval $[l_i, u_i]$, where $0 < l_i < u_i < 1$.   We discuss the model and Problem~\ref{prob1} in greater detail in Section~\ref{sec:model}.

\noindent \textbf{Our Contributions.}  In this work, we make the following key contributions. 

\begin{itemize}
\item \emph{Opinion Dynamics with Varying Susceptibility to Persuasion.}  We introduce a novel framework for social influence that focuses on interventions at the level of susceptibility.  

\item \emph{Analysis of the unbudgeted problem structure.}  We prove that the objective function is in general neither convex nor concave.  We analyze the mathematical structure of the problem in Section~\ref{sec:structure}.
Perhaps the most important technical insight in this paper is that we show (in Lemma~\ref{le:alpha_descend}) that
if the current vector solution is not optimal, then there exists a coordinate that can be flipped
such that the objective will be strictly improved.  This shows that an optimal vector can be found by a simple local search algorithm.

\item \emph{Local search with irrevocable updates.}  In general, local search could still take exponential time to find an optimal solution, for instance, the simplex algorithm for linear programming.
For minimizing the sum of equilibrium opinions, we show (in Lemma~\ref{lemma:once}) that starting from the upper bound resistance vector,
then the local search algorithm will flip each coordinate at most once, which implies that an optimal vector can be found in polynomial time.

\item \emph{Efficient Local Search on Large-Scale Graphs.}  Typically, in local search, the objective function needs to be evaluated
at the current solution in each step.  However, since the objective function involves matrix inverse,
its evaluation will be too expensive when the dimension of the matrix is in the order of millions.
Instead, we use the iterative process of the opinion dynamics model itself to approximate the equilibrium vector.
We have developed several update strategies for local search.  For \emph{conservative} or \emph{opportunistic} updates,
one always makes sure that the error of the estimated equilibrium vector is small enough before any coordinate of the resistance vector is flipped.  For \emph{optimistic} update, one might flip a coordinate of the resistance vector even before the estimated equilibrium vector is accurate enough.  However, this might introduce mistakes which need to be corrected later.  Nevertheless,
experiments show that mistakes are rarely made by the optimistic update strategy. In any case,
for all three update strategies, an optimal vector will be returned when the local search terminates.

Our approaches are scalable and can run on networks with millions of nodes. We report the experimental results in Section~\ref{sec:experiment}. In particular, using multiple number of threads, the optimistic update strategy can solve the problem optimally on networks with up to around 65 million nodes.

\item  \emph{Scalable Heuristics for the Budgeted problem.} We provide a family of efficient heuristics for the budgeted version of our problem, and a detailed experimental evaluation on large-scale real-world networks.
\end{itemize}

\noindent \textbf{Comparison with Previous Versions.}  A preliminary version~\cite{AbebeKPT18} of this work presented 
the problem, but it was overlooked that the objective function might not be convex or concave.
A subsequent work~\cite{DBLP:conf/www/ChanLS19} rectified this issue, and showed that local search
can be performed efficiently to reach the optimal solution, even if the objective function is non-convex.
The current presentation combines results from the aforementioned two works~\cite{AbebeKPT18,DBLP:conf/www/ChanLS19}.  We have also included a more detailed Section~\ref{sec:heuristic} on heuristic algorithms for the budgeted version of the problem,
and the related experiments in Section~\ref{sec:experiment_budgeted}.
 
\section{Related work} 
\label{sec:related}

To our knowledge, we are the first to consider an optimization framework based on opinion dynamics with varying susceptibility to persuasion. In the following we review briefly some work that lies close to ours. 

\noindent \emph{Susceptibility to Persuasion.} Asch's conformity experiments
are perhaps the most famous study on the impact of agents'
susceptibility to change their opinions \cite{Asch1955}. This study shows
how agents have different propensities for conforming with others.
These propensities are modeled in our context by the set of parameters
$\alpha$. Since the work of Asch, there have been various theories on
peoples' susceptibility to persuasion and how these can be affected. A
notable example is Cialdini's Six Principles of Persuasion, which
highlight reciprocity, commitment and consistency, social proof,
authority, liking, and scarcity, as key principles which can be
utilized to alter peoples' susceptibility to persuasion
\cite{Cialdini1993,Cialdini2001}. This framework, and others, have been
discussed in the context of altering susceptibility to persuasion in a
variety of contexts.  Crowley and Hoyer \cite{Crowley1994}, and McGuire
\cite{mcguire1964inducing} discuss the `optimal arousal theory', i.e., how novel
stimuli can be  utilized for persuasion when discussing arguments.

\noindent \emph{Opinion Dynamics Models.}  Opinion dynamics model social learning processes. DeGroot introduced a continuous opinion dynamics model in his seminal work on consensus formation \cite{DeGroot1974}. A set of $n$ individuals in society start with initial opinions on a subject. Individual opinions are updated using the average of the neighborhood of a fixed social network. Friedkin and Johnsen \cite{Friedkin1999} extended the DeGroot model to include both disagreement and consensus by mixing each individual's \emph{innate belief} with some weight into the averaging process. This has inspired a lot of follow up work, including \cite{AcemogluO11,BindelKO15,DasGM2014,GhaderiS13,GionisTT13}.

\noindent \emph{Optimization and Opinion Dynamics.}  Bindel et al. use the Friedkin-Johnsen model as a framework for understanding  the price of
anarchy in society when individuals selfishly update their opinions in
order to minimize the stress they experience \cite{BindelKO15}. They also
consider network design questions: given a budget of $k$ edges, and a
node $u$, how should we add those $k$ edges to $u$ to optimize an
objective related to the stress?   Gionis, Terzi, and Tsaparas \cite{GionisTT13} use the
same model to identify a set of target nodes whose innate opinions can
be modified to optimize the sum of expressed opinions.
Musco, Musco, and Tsourakakis adopt the same model to understand
which graph topologies minimize the sum of disagreement and
polarization \cite{Musco2018} .

\noindent \emph{Inferring opinions and conformity parameters.} While the
expressed opinion of an agent is readily observable in a social network,
both the agent's innate opinion and conformity parameter are hidden, and
this leads to the question of inferring them. Such
inference problems have been  studied by Das et al. \cite{DasGKP2014,Das2013}.
Specifically, Das et al. give a near-optimal sampling algorithm for
estimating the true average innate opinion of the social network and
justify the algorithm both analytically and experimentally \cite{Das2013}.
Das et al.  view the problem of susceptibility parameter estimation as
a problem in constrained optimization and give efficient algorithms,
which they validate on real-world data \cite{DasGKP2014}.

\noindent \emph{Non-Convex Optimization.}  In general, optimizing a non-convex function under non-convex constraints is NP-hard.  However, in many cases, one can exploit the structure of the objective function or constraints to devise polynomial-time algorithms; see the survey by Jain and Kar~\cite{DBLP:journals/ftml/JainK17} on non-convex optimization algorithms encountered in machine learning.
Indeed, variants of the gradient descent have been investigated
to escape saddle points by Jin et al.~\cite{Jin0NKJ17},
who also gave examples of problems where all local optima are also global optima;
some examples are 
tensor decomposition~\cite{GeHJY15}, dictionary learning~\cite{SunQW17}, phase retrieval~\cite{SunQW16}, matrix sensing~\cite{BhojanapalliNS16,ParkKCS17} and matrix completion~\cite{GeLM16}.  However,
all these problems involve some quadratic loss functions,
whose structures are totally different from our objective functions which involve matrix inverse.

Hartman~\cite{hartman1959} considered the special case that the objective function is the difference of
two convex functions. Strekalovsky devised a local search method to optimize such objective functions.  Even though the objective functions in our problem are somewhere convex and somewhere concave (see Figure~\ref{fig:example}),
it is not immediately clear if they can be expressed as differences of convex functions.

\section{Model} 
\label{sec:model}

Let $G = (V,E)$ be a simple, undirected graph, where $V = [n]$ is the set of agents and $E$ is the set of edges. Each agent $i \in V$ is associated with an \emph{innate opinion} $s_i \in [0,1]$, where higher values correspond to more favorable opinions towards a given topic and a parameter measuring an agent's susceptibility to persuasion $\alpha_i \in (0,1]$, where higher values signify agents who are less susceptible to changing their opinion. We call $\alpha_i$ the \emph{resistance parameter}. The opinion dynamics evolve in discrete time according to the following model, inspired by the work of \cite{DeGroot1974,Friedkin1999}: 

\begin{equation}
\label{eq:dynamics}
z_i(t+1) = \alpha_i s_i + (1-\alpha_i) \frac{ \sum\limits_{j \in N(i)}  z_j(t) }{\text{deg}(i)}.
\end{equation}

\noindent   Here, $N(i) = \{ j : (i, j) \in E \}$ is the set of neighbors of $i$, and $\text{deg(i)}=|N(i)|$ is the degree of node $i$.  Equivalently, by defining $A = \Diag(\alpha)$ to be the diagonal  matrix with $A_{ii} = \alpha_i$, $I$ is the identity matrix, and $P \in [0,1]^{V \times V}$ to be the row stochastic matrix (i.e., each entry of $P$ is non-negative
and every row sums to 1) that captures agents interactions, we can rewrite Equation~\ref{eq:dynamics} as  $z^{(t+1)} := A s + (I - A) P z^{(t)}$.   Equating $z^{(t)}$ with $z^{(t+1)}$, one can see that the equilibrium opinion vector is given by $z = [I - (I-A)P]^{-1} As$, which exists under non-trivial conditions such as every $\alpha_i > 0$.
In the rest of this paper, we always call $P$ the interaction matrix.

We quantify Problem~\ref{prob1} as follows. The objective is to choose a resistance vector $\alpha$ to minimize the sum of equilibrium opinions $\langle \one, z \rangle = \one^\top z$. Observe that one can also consider maximizing the sum of equilibrium opinions; however, since the techniques are essentially the same, we will focus on
the minimization variant of the problem.

\begin{definition}[Opinion Susceptibility Problem]
\label{defn:osp}
Given a set $V$ of agents
with innate opinions $s \in [0,1]^V$
and interaction matrix $P \in [0,1]^{V \times V}$,
suppose for each $i \in V$, its resistance is restricted
to some interval $\mcal{I}_i := [l_i, u_i] \subseteq [0,1]$
where we assume that $0 < l_i < u_i < 1$.

The objective is to choose $\alpha \in \mcal{I}_V := \times_{i \in V} \mcal{I}_i \subseteq [0,1]^V$
such that the following objective function is minimized:
$$f(\alpha) := \one^\top [I-(I-A)P]^{-1}As,$$
where $A=\Diag(\alpha)$ is the diagonal matrix with
$A_{ii}=\alpha_i$.  Observe that the assumption $\alpha > 0$ ensures that the above inverse exists.
\end{definition}

\noindent \emph{Unbudgeted vs Budgeted Variants.} In Definition~\ref{defn:osp},
we are allowed to modify the resistance of any agent, and this is known as the
unbudgeted variant. We also consider the budgeted variant:  given some initial resistance vector and a budget~$k$, the resistance of at most $k$ agents can be changed. In this paper, we focus on efficient algorithms that optimally solve the
unbudgeted variant. In Section~\ref{sec:hardness}, we prove that the budgeted variant is NP-hard, and we propose efficient heuristics that scale to large networks. Designing  algorithms with solid approximation guarantees for the budgeted variant is an interesting open problem.

\noindent \textbf{Technical Assumption.}  To simplify our proofs, we assume that the interaction matrix $P$ corresponds to an irreducible random walk.  Irreducibility is satisfied if $P$ arises from a connected graph.
\section{NP-hard Budgeted Opinion Susceptibility Problem}
\label{sec:hardness}

We now consider the setting where there is a constraint on the size of the target-set. That is, we want to identify a set $T \subseteq V$ of size $k$ 
such that changing the resistance parameters of agents in $T$ optimally maximizes (resp. minimizes) the sum of equilibrium opinions. 
We use $\alpha^{(0)}$ to denote the given initial resistance vector.
For $T \subseteq V$,
we define $F(T) := \max\{ f(\alpha): \forall i \notin T, \alpha_i = \alpha^{(0)}_i \}$;
observe that $F$ is defined with respect to the
initial resistance vector $\alpha^{(0)}$. The budgeted opinion optimization problem is to maximize $F(T)$ subject to the budget constraint $|T| \leq k$.

\noindent \begin{theorem}
The budgeted opinion optimization problem is \NPhard. 
\end{theorem}

\begin{proof}
We give a reduction from the vertex cover problem for regular graphs.
Given a $d$-regular graph $G=(V, E)$ and an integer $K$, the vertex
cover problem asks whether there exists a set $S$ of nodes with size at
most $K$ such that $S$ is a vertex cover, i.e., every edge in $E$ is
incident to at least one node in $S$.
For simplicity, we assume that $\sqrt{d}$ is an integer.

We give a reduction from the above vertex
cover problem to
the decision version of the opinion optimization problem.
In addition to a given graph $G'$, the
innate opinion vector $s$, the
initial resistance vector $\alpha^{(0)}$
and the budget $k$,
an instance of the decision
version of opinion maximization 
also has some threshold $\theta$.
The instance is ``yes'' \emph{iff}
there exists some node set $T$ in $G'$ with
size at most $k$ such that $F(T) \geq \theta$. To illustrate our ideas, we first give a reduction in which each agent's resistance parameter is in the range $[0,1]$.
Then, we show how to restrict the resistance to the range $[\epsilon, 1]$
for some small enough $\epsilon > 0$.

\noindent \textbf{Reduction Construction.}
Suppose we are given an instance of the vertex cover problem for
regular graphs. We construct an instance of the decision version of
the opinion optimization problem. Define $G' = (V \cup V', E
\cup E')$, where $V$ and $E$ come from original vertex cover problem.
For $i \in V$, $s_i = 1$ and $\alpha^{(0)} = 0$;
for $i' \in V'$, $s_{i'} = 0$, and
we give more details on their initial resistance parameters.
The additional nodes $V'$ and edges $E'$ are added as follows.
Let $\sigma = 2 n^2 (\sqrt{d}+1)$ (specified later).

For each $i \in V$, we add $(\sigma + 1) \sqrt{d}$ additional nodes $V_i$
into $V'$:

\begin{enumerate}[(i)]

\item $\sqrt{d}$ \emph{flexible} nodes. Each such node has degree 1 and
is connected only to node~$i$; its initial resistance parameter in $\alpha^{(0)}$
is 0.

\item $\sigma \sqrt{d}$ \emph{stubborn} nodes.  These nodes
form $\sqrt{d}$ cliques, each of which has size $\sigma$.
In each clique, exactly one node is connected to $i$.
All the stubborn nodes have initial resistance parameters in $\alpha^{(0)}$
being 1.

\end{enumerate}

 Observe that in $G'$, the degree of each node in $V$ is $d+2\sqrt{d}$.
Finally,
we set the budget
$k = K$ and the threshold $\theta = (\sqrt{d} + 1)k + (\sqrt{d}+1) (n -k) \frac{d }{d +
\sqrt{d}}$.
To complete the reduction proof,
we show that there exists a vertex cover of size $k$ in $G$
\emph{iff} there exists some $T \subseteq V \cup V'$ of size $k$
such that $F(T) \geq \theta$.


\ignore{
Observe that the degree of $i$ is $d+2\sqrt{d}$ and the degree of each node in a clique is $\sigma-1$ except for the one connecting to $i$ of which the degree is $\sigma$.
The set $E'$ is the edges of the form
$(i, i')$ joining nodes in $V$ to their corresponding duplicate nodes as described above.
We set $s_i = 1$ and $\alpha_i = 0$ for all $i \in V$. We also set
$s_{i'} = 0$ for all $i' \in V'$.
For the set of duplicate nodes with degree one, we
set the resistance value of them to be $0$
and the others to be $1$. We call the former the \emph{flexible}
duplicate nodes, and the latter the \emph{stubborn} duplicate nodes. 
Observe that under this setting, the equilibrium opinion of each node in $V\cup V'$ can be at most 1.
We set
$k = K$ and $\theta = (\sqrt{d} + 1)k + (\sqrt{d}+1) (n -k) \frac{d }{d +
\sqrt{d}}$. This gives us a decision version of each of the opinion
maximization problem.
}

\noindent \textbf{Forward Direction.}  Suppose in $G$, there is some vertex cover $T \subset V$
with size $k$.  We show that in $G'$, $F(T) \geq \theta$;
we set $\alpha_i = 1$ for each $i \in T$,
while the resistance parameters of all other nodes remain the same as in $\alpha^{(0)}$.
We next analyze the equilibrium opinion of each node.
Observe that all stubborn nodes in $V'$ have equilibrium opinion 0.

For $i \in T$, node~$i$ has equilibrium opinion 1;
moreover, all its $\sqrt{d}$ flexible neighbors in $V'$ will also
have equilibrium opinion 1.

For $j \in V \setminus T$,
we compute its equilibrium opinion $z_j$.
Since $T$ is a vertex cover,
all $d$ neighbors of $j$ in $V$ are in $T$
and have equilibrium opinion 1.
All $\sqrt{d}$ flexible neighbors of $j$ in $V'$
have the same equilibrium opinion $z_j$,
while the $\sqrt{d}$ stubborn neighbors have opinion 0.
Therefore, $z_j$ satisfies the equation
$z_j = \frac{d}{d + 2 \sqrt{d}} \cdot 1 + 
\frac{\sqrt{d}}{d + 2 \sqrt{d}} \cdot z_j + 
\frac{\sqrt{d}}{d + 2 \sqrt{d}} \cdot 0$,
which gives $z_j = \frac{d}{d + \sqrt{d}}$.

Therefore,
we have
$F(T) \geq (\sqrt{d} +1 ) k + (\sqrt{d}+1) (n - k) \frac{d}{d + \sqrt{d}} = \theta$.

\noindent \textbf{Backward Direction.}
Suppose there is some $T \subseteq V \cup V'$ such that $|T| = k$ and
$F(T) \geq \theta$.  The goal is to show that
there is a vertex cover with size $k$ in $G$.
Observe that the innate opinions of nodes in $V$ are 1; hence,
if we are allowed to change the resistance of a node  $i \in V$,
we should set $\alpha_i = 1$ to maximize the total equilibrium opinion.

We consider the following two cases.

\begin{enumerate}

\item[1.] Case $T \subseteq V$.
i.e. all vertices in $T$ are from $V$.
We prove that $T$ is a vertex cover in $G$ by contradiction.

Assume that there exists an edge $\{i, j\} \in E$ such that both $i, j \notin T$. 
We derive an upper bound $z$ for the equilibrium opinion of $i$ and $j$.
Observe that for node~$i$,
at most $(d-1)$ of its neighbors are in $T$.
Hence, we
have $z \leq 
\frac{d-1}{d + 2 \sqrt{d}} \cdot 1 + 
\frac{\sqrt{d}+1}{d + 2 \sqrt{d}} \cdot z + 
\frac{\sqrt{d}}{d + 2 \sqrt{d}} \cdot 0$,
which gives $z_j \leq \frac{d-1}{d + \sqrt{d}-1}$;
the important fact is that 
$\frac{d-1}{d + \sqrt{d}-1} < \frac{d}{d + \sqrt{d}}$.

Observe that for any node in $V \setminus T$,
its equilibrium opinion is maximized when all its neighbors
in $V$ are in $T$.

Hence,
we have $F(T) \leq \theta'$,
where $\theta' :=
(\sqrt{d}+1)k + (\sqrt{d}+1) (n - k - 2) \frac{d}{d
+ \sqrt{d}} + (\sqrt{d}+1)2 \frac{d-1}{d + \sqrt{d}-1} < \theta$,
achieving the desired contradiction.

\item[2.] Case $T \setminus V \not= \emptyset$.
In this case,
we choose $T$ of size $k$ such that $F(T)$ is maximized;
if there is more than one such $T$,
we arbitrarily pick one such that $|T \cap V|$ is maximized.
For contradiction's sake,
we assume that $T \setminus V \neq \emptyset$
and $T \cap V$ is not a vertex cover of $G$.
(We actually just need the weaker condition
that $V \setminus T$ is non-empty.)
We further consider the following cases.

\begin{enumerate}[(i)]

\item There is some flexible node $u$ in $T \setminus V$.
Suppose the degree-1 node $u$ is connected to $i \in V$.
If $i \notin T$,
then one can consider $T' := T - u + i$;
if $i \in T$, then pick any $j \in V \setminus T$
and consider $T' := T - u + j$.

In either case, it follows that $F(T') \geq F(T)$ and
$|T' \cap V| > |T \cap V|$, achieving
the desired contradiction.

\item There is some stubborn node $u$ in $T \setminus V$.
Suppose $u$ is in the clique associated with $i \in T$.
Observe that at most $k$ nodes in the clique are in $T$.
Hence, it follows that the equilibrium opinion
of any stubborn node is at most $\frac{k+1}{\sigma}$.

Hence, for any $j \in V \setminus T$,
its equilibrium opinion satisfies
$z_j \leq
\frac{d}{d + 2 \sqrt{d}} \cdot 1 + 
\frac{\sqrt{d}}{d + 2 \sqrt{d}} \cdot z_j + 
\frac{\sqrt{d}}{d + 2 \sqrt{d}} \cdot \frac{k+1}{\sigma}$.
Since $\frac{k+1}{\sigma} \leq \frac{1}{2}$,
we have $z_j \leq \frac{d + \frac{1}{2} \sqrt{d}}{d + \sqrt{d}}$.

Next, if $i \notin T$, then set $j := i$;
otherwise, just pick any $j \in V \setminus T$;
consider $T' := T - u + j$.

Note that the equilibrium opinion
of stubborn nodes in the clique of $u$ can drop by at most
$\frac{k(k+1)}{\sigma} < \frac{n^2}{\sigma}$.

However, the equilibrium opinion of other nodes cannot decrease
and that of $j$
increases by at least $\frac{\frac{1}{2} \sqrt{d}}{d + \sqrt{d}}
= \frac{n^2}{\sigma}$, by the choice of $\sigma$.

Hence, we have the contradiction $F(T') > F(T)$.
\end{enumerate}
\end{enumerate}

This completes the reduction proof
for which the resistance parameter is chosen in the range $[0,1]$.
We next show to 
how to modify the proof for the case in which
the resistance value is chosen
in the interval $[\epsilon, 1]$
for some sufficiently small $\epsilon > 0$.

The key point is that when we view $F(T)$ as a function of the 
resistance parameters in the network $G'$ constructed
in the reduction, it is a continuous function.
Define $\gamma := \theta - \theta' > 0$, where
$\theta$ and $\theta'$ are defined as above.

One can choose $\epsilon > 0$ small enough such that the following holds.
In the above proof, if we replace any 0 resistance value with $\epsilon$,
then we have
(i) if $T$ is a vertex cover of size $k$ in $G$,
then $F(T) \geq \theta - \frac{\gamma}{3}$;
(ii) if $G$ does not have a vertex cover of size $k$,
then for any $T \subseteq V \cup V'$ of size $k$,
$F(T) \leq \theta' + \frac{\gamma}{3}$.
This completes the proof.

\end{proof}

\section{Structural Properties of Objective Function}
\label{sec:structure}

In this section, we investigate the properties of the objective
function $f$ in Definition~\ref{defn:osp};
we assume that the interaction matrix~$P$ and the innate opinion vector $s$ are fixed,
and $f$ is a function on the resistance vector~$\alpha$.
 
\noindent \textbf{Non-convex Objective.}  Contrary to the claim in a preliminary version of this work (see \cite{AbebeKPT18}), the objective $f$ in Definition~\ref{defn:osp} is in general not a convex function of $\alpha$.  In fact, the following example shows that it might be neither convex nor concave.  Consider three vertices $V = \{1,2,3\}$,
where the innate vector $s$ and the interaction matrix $P$ are given by:
$
s=
\begin{bmatrix}
1    \\
0.5  \\
0
\end{bmatrix}
$
and
$
P=
\begin{bmatrix}
    0     & 0.5  & 0.5 \\
    0.5   & 0    & 0.5 \\
    0.5   & 0.5  & 0
\end{bmatrix}
$.

Suppose we fix $\alpha_2 = \alpha_3 = 0.1$ and consider
the objective as a function of $\alpha_1$ as
$g(\alpha_1) = \one^\top [I-(I-A)P]^{-1}As$,
where $A = \Diag(\alpha_1, \alpha_2, \alpha_3)$. Then, the plot of~$g$ in
Figure~\ref{fig:example} (a) shows that it is not convex.
Moreover, suppose this time we fix $\alpha_1 = \alpha_2 = 0.1$ and consider
the objective as a function of $\alpha_3$ as
$h(\alpha_3) = \one^\top [I-(I-A)P]^{-1}As$.
Then, the plot of~$h$ in
Figure \ref{fig:example} (b) shows that it is not concave.
\begin{figure}[h]
\centering
\subfigure[]{
\includegraphics[width=1.55in]{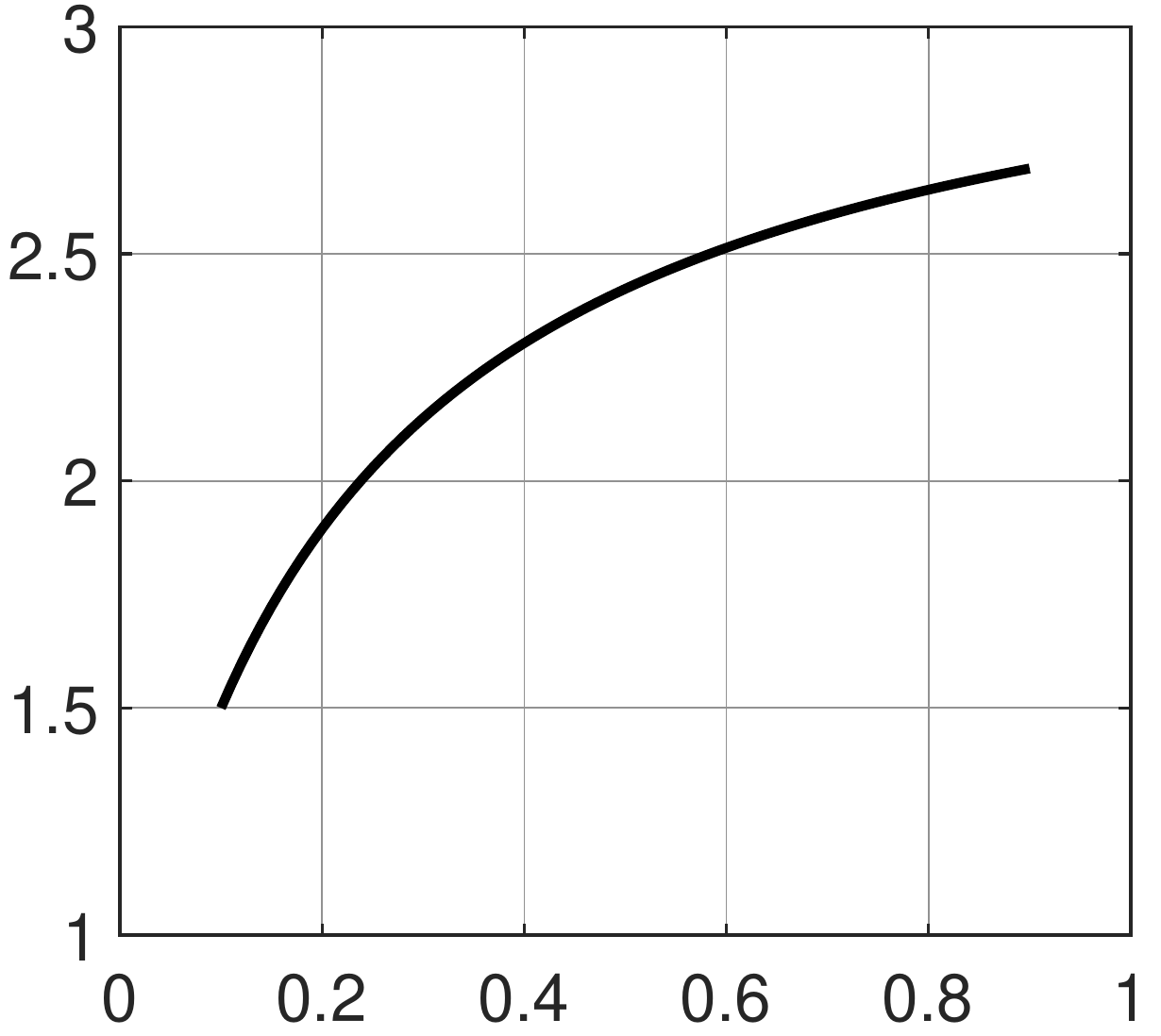}
}
\subfigure[]{
\includegraphics[width=1.55in]{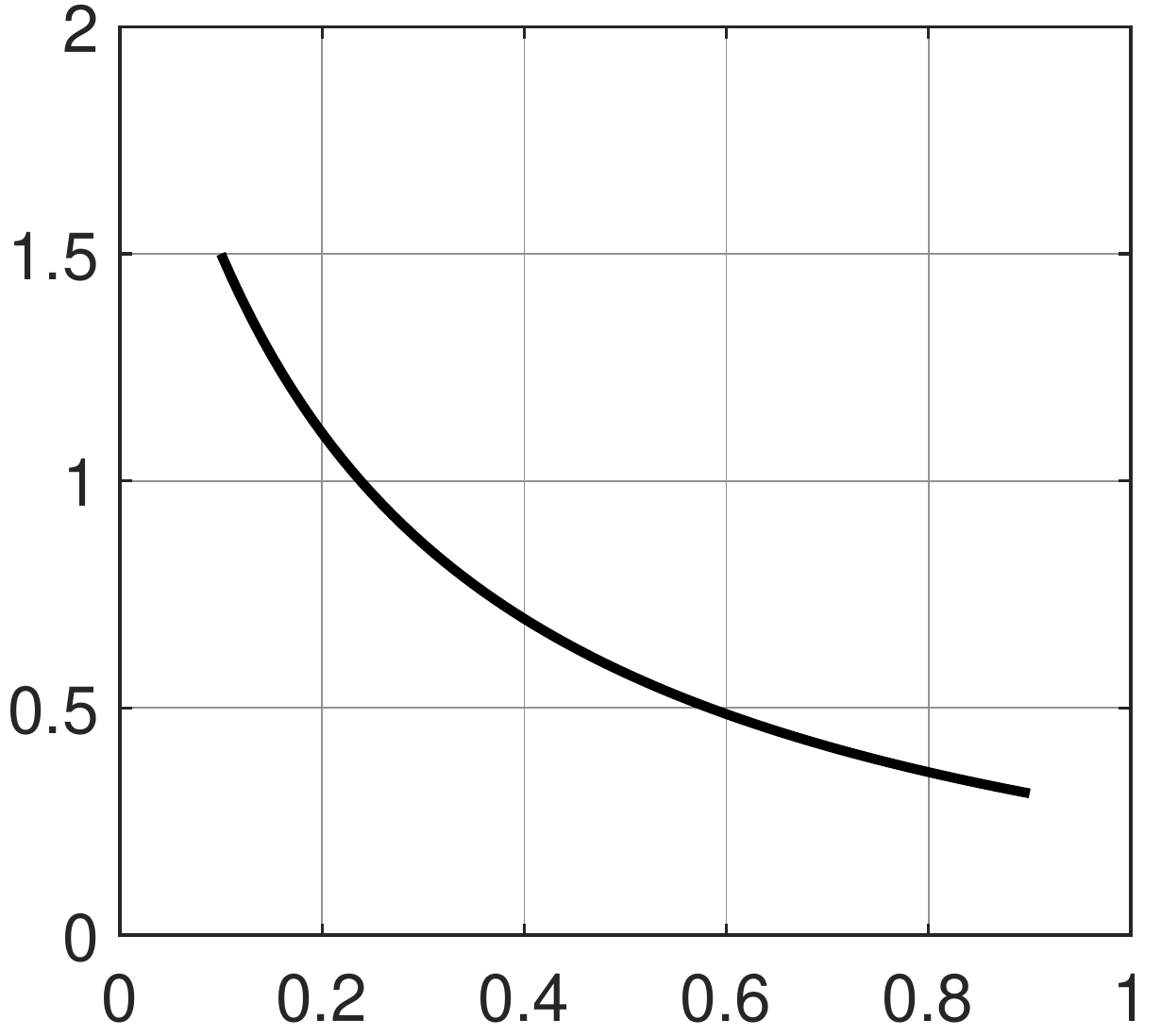}
}
\caption{Cross-Sections of Objective Function}
\label{fig:example}
\end{figure}

Fortunately, we can still exploit some properties of the function.  As we shall see,
even when the function is not convex, every local optimum (which will be defined formally) is a global optimum.  This enables us to use variants of the local search method to solve the problem optimally.

\subsection{Marginal Monotonicity}
\label{sec:monotone}

As in~\cite{AbebeKPT18}, we show that
when one chooses
the resistance $\alpha_i$ for each agent~$i \in V$,
it suffices to consider the extreme points $\{l_i, u_i\}$.
Our approach explicitly analyzes the partial
derivative $\frac{\partial f(\alpha)}{\partial \alpha_i}$
which plays a major role in the local search algorithm that we later develop.

\noindent \textbf{Intuition:
Guidance by Current Equilibrium Vector.}
Observe that given the
innate opinion vector
$s$ and irreducible interaction matrix~$P$,
for some resistance vector $\alpha \in (0,1)^V$,
the equilibrium opinion vector is given by
$z(\alpha) = [I-(I- A)P]^{-1}  A s$, where $A = \Diag(\alpha)$.
For some $i \in V$, if the innate opinion $s_i$ 
is larger than its equilibrium $z_i(\alpha)$, this suggests that
by being more stubborn, agent~$i$ should be able to increase its
equilibrium opinion.  In other words, one would expect $\frac{\partial z_i(\alpha)}{\partial \alpha_i}$
and $s_i - z_i(\alpha)$ to have the same sign.
However, what is surprising is that in Lemma~\ref{le:same_sign},
we shall see that even for any $j \in V$, 
$\frac{\partial z_j(\alpha)}{\partial \alpha_i}$ and 
$s_i - z_i(\alpha)$ have the same sign.

\textbf{Notation.}  
For any $\alpha \in \R^V$ and  $K \subseteq V$, let
$\alpha_{-K} \in \R^V$ denote the vector
such that $\alpha_{-K}(i) = \alpha(i)$ if $i \notin K$,
and $\alpha_{-K}(i) = 0$ if $i \in K$, i.e.,
the coordinates $K$ of $\alpha$ are replaced with 0.
Similarly, given $\alpha \in \R^V$,
we denote $A_{-K} =\Diag(\alpha_{-K})$.

In Definition~\ref{defn:osp}, observe that
the inverse $[I-(I-A)P]^{-1}$ is involved in the 
objective function $f(\alpha)$, where $A = \Diag(\alpha)$.
Since we wish to analyze the effect on $f(\alpha)$ of changing
only a subset of coordinates in $\alpha$,
the next lemma will be used for simplifying matrix arithmetic involving
the computation of inverses.
Its proof is deferred to Section~\ref{sec:proofs}.

\begin{lemma}[Inverse Arithmetic]
    Given ${K}\subsetneq V$ and $\alpha \in (0,1)^V$, let $A := \Diag(\alpha)$
		and recall that $P$ is the irreducible interaction matrix.
		Then, the inverse $M=[I-(I-A_{-K})P]^{-1}$ exists,
    and every entry of $M$ is positive.
		 Moreover, for each $k \in V$, 
			define $a_k = 0$ if $k \in {K}$, otherwise $a_k = \alpha_k$.
		Then, we have:
		\begin{enumerate}
        \item $(PM)_{kk}=\frac{M_{kk}-1}{1-a_k} > 0$;
        \item $(PM)_{kj}=\frac{M_{kj}}{1-a_k} > 0$, for each $j\not=k$.
    \end{enumerate}
\ignore{
    \begin{enumerate}
        \item $M_{kk}=(1-a)(PM)_{kk}+1 \geq 1$,
        \item $M_{kj}=(1-a)(PM)_{kj}$ for each $j\not=k$.
    \end{enumerate}
}
    \label{le:PM_M}
\end{lemma}

\begin{lemma}[Sign of Partial Derivative]
        In the Opinion Susceptibility Problem in Definition~\ref{defn:osp},
				given the
                innate opinion vector
                $s$ and irreducible interaction matrix~$P$,
        recall that $z(\alpha) := [I-(I-A)P]^{-1}As$, where $A = \Diag(\alpha)$.
		Then, for any $\alpha \in (0,1)^V$ and any $i, k \in V$, 
		the two values 
		$\frac{\partial z_k(\alpha)}{\partial \alpha_i}$ and
		        $s_i-z_i(\alpha)$
				have exactly the same sign in $\{-,0,+\}$.
				
				In particular,
				this implies that $\frac{\partial f(\alpha)}{\partial \alpha_i} = \sum_{k \in V} \frac{\partial z_k(\alpha)}{\partial \alpha_i}$ also has the same sign as $s_i-z_i(\alpha)$.
				\label{le:same_sign}
\end{lemma}

\begin{proof}
By the definition of the inverse of a matrix $B$,
we have 
$
BB^{-1}=I.
$
The partial derivative with respect to a variable $t$ is:
$
\frac{\partial B}{\partial t}B^{-1}
+B\frac{\partial B^{-1}}{\partial t}=0.
$
Hence, we have
$
\frac{\partial B^{-1}}{\partial t}
=-B^{-1}\frac{\partial B}{\partial t}B^{-1}.
$
Applying the above result
with $B = I-(I-A)P$ and $t=\alpha_i$
and denoting
$M = [I-(I-A)P]^{-1}$,
we get
$
\frac{\partial M}{\partial \alpha_i}
=-M\mathbf{e}_i\mathbf{e}_i^\top PM.
$
Considering $z(\alpha) = MAs$, we have
$
\frac{\partial z(\alpha)}{\partial \alpha_i}
=\frac{\partial M}{\partial \alpha_i}As
+M\mathbf{e}_i\mathbf{e}_i^\top s.
$
Replacing $\frac{\partial M}{\partial \alpha_i}$,
we obtain for any $i,k \in V$:
$$
\frac{\partial z_k(\alpha)}{\partial \alpha_i}=- \mathbf{e}_k^\top  M \mathbf{e}_i\mathbf{e}_i^\top PM As
+ \mathbf{e}_k^\top M \mathbf{e}_i\mathbf{e}_i^\top s
= M_{ki} \cdot [s_i-\mathbf{e}_i^\top Pz(\alpha)].
$$

\ignore{
\begin{align}
\begin{split}
\frac{\partial z_k(\alpha)}{\partial \alpha_i}
=&- \mathbf{e}_k^\top  M \mathbf{e}_i\mathbf{e}_i^\top PM As\\
&+ \mathbf{e}_k^\top M \mathbf{e}_i\mathbf{e}_i^\tops\\
=&\mathbf{e}_k^\top M\mathbf{e}_i[s_i-\mathbf{e}_i^\top Pz(\alpha)]
\end{split}
\nonumber
\end{align}
}
By Lemma \ref{le:PM_M} with ${K}=\emptyset$,
we know that every entry of $M$ is positive.
Thus, 
the sign of $\frac{\partial z_k(\alpha)}{\partial \alpha_i}$
is the same as that of the scalar $s_i-\mathbf{e}_i^\top Pz(\alpha)$.

Recalling $M=[I-(I-A)P]^{-1}$,
we have 
$
[I-(I-A)P]M=I
\Rightarrow
(I-A)PM=M-I
\Rightarrow
PM=(I-A)^{-1}(M-I)
$
where $(I-A)^{-1}$ exists since $\alpha_j<1$ for each $j\in V$.

Next, since $z(\alpha) = MAs$,
we have:
$$
P z(\alpha)=PMAs=(I-A)^{-1}(M-I)As
=(I-A)^{-1}[z(\alpha)-As].
$$

Finally, replacing $Pz(\alpha)$, we have
\begin{align}
\begin{split}
s_i-\mathbf{e}_i^\top Pz(\alpha)
=&s_i-\mathbf{e}_i^\top(I-A)^{-1}[z(\alpha)-As]\\
=&s_i-\frac{1}{1-\alpha_i}[z_i(\alpha)-\alpha_is_i]\\
=&\frac{1}{1-\alpha_i}[s_i-z_i(\alpha)].
\end{split}
\nonumber
\end{align}

Since $1-\alpha_i>0$,
it follows that
$\frac{\partial z_k(\alpha)}{\partial \alpha_i}$
and $s_i-z_i(\alpha)$ have exactly the same sign in $\{-,0,+\}$,
as required.
\end{proof}

\ignore{
\begin{corollary}
    \label{cor:all_decrease}
    In Definition~\ref{defn:osp},
    let $z(\alpha) := [I-(I-A)P]^{-1}As$.
    For some $\alpha$, let $\mathcal{L}$ denote the set of coordinates $k$'s
    such that $s_k-z_k(\alpha)>0$.
    Then for any $i\in \mathcal{L}$,
    if $\alpha_i$ decreases,
    $s_k-z_k(\alpha)>0$ still holds for each $k\in \mathcal{L}$.
\end{corollary}
}

The next lemma shows that the sign of the partial
derivatives with respect to coordinate~$i$
is actually independent of the current value $\alpha_i$.
Its proof is deferred to Section~\ref{sec:proofs}.

\begin{lemma}[Sign of Partial Derivative Independent of Coordinate Value]
Referring to Lemma~\ref{le:same_sign}
.
For any $\alpha \in (0,1)^V$ and any $i \in V$,
denote $M=[I-(I-A_{-\{i\}})P]^{-1}$.
Then, $\frac{\partial f(\alpha)}{\partial \alpha_i}$ has the same sign in $\{-, 0, +\}$
as $s_i -\sum_{j\not=i}M_{ij}\alpha_j s_j$,
which is independent of $\alpha_i$.
    \label{le:mar_mon}
\end{lemma}

\begin{corollary}[Extreme Points are Sufficient]
\label{cor:ext}
In Definition \ref{defn:osp}, for any $i \in V$,
fixing the resistance values of all other agents except $i$,
the objective $f(\alpha)$ is a monotone function in $\alpha_i$.
This implies that to minimize $f$,		
it suffices to consider the extreme points $\alpha_i \in \{l_i, u_i\}$, for each $i \in V$.
\end{corollary}

\subsection{Local vs Global Optimum}
\label{sec:local}

As shown in Corollary~\ref{cor:ext},
it suffices to choose the resistance vector $\alpha$
from the extreme points in Definition~\ref{defn:osp}.
Lemma~\ref{le:same_sign} readily gives a method to decide,
given a current choice of $\alpha$,
whether the objective $f$ can be decreased by changing
the resistance of some agent.
In Lemma~\ref{le:alpha_descend}, we show that if $\alpha$ is not a global
minimum, then such an agent must exist.
As we shall see, this implies that a local search method can find
a global minimum.

Given $\alpha$ and $\alpha' \in \R^V$,
denote $\Delta(\alpha, \alpha') := \{ i \in V: \alpha_i \neq \alpha_i'\}$
as the set of coordinates at which the vectors differ.

\begin{definition}[Local Minimizer]
\label{defn:local}
Given an objective function $f : \mcal{I}_V \rightarrow \R$,
a vector $\alpha \in \mcal{I}_V$ is a local minimizer of $f$,
if for all $\alpha' \in \mcal{I}_V$ such that $\Delta(\alpha, \alpha') = 1$,
$f(\alpha) \leq f(\alpha')$.
\end{definition}

\noindent \textbf{Notation.}  When we wish to
consider the effect of changing the resistance
of only 2 agents $i \neq k \in V$,
we write $f(\alpha) = f(\alpha_i, \alpha_k)$,
assuming that $\alpha_{-\{i,k\}}$ is fixed.


Lemmas~\ref{le:M2R} and \ref{le:col_max}
give some technical results
involving changing the resistance
of two agents.  Their proofs are deferred to Section~\ref{sec:proofs}.

\begin{lemma}
For any $i,k\in V$ such that $i\not= k$, let $M=[I-(I-A_{-\{i\}})P]^{-1}$
and $R=[I-(I-A_{-\{i,k\}})P]^{-1}$. Then for any $j \in V$,
we have
\begin{enumerate}
\item $M_{jk}=\frac{R_{jk}}{1+\alpha_kR_{kk}-\alpha_k}$,
\item $M_{jh}=R_{jh}-\frac{\alpha_{k}R_{jk}R_{kh}}{1+\alpha_kR_{kk}-\alpha_k}$, for $h\not=k$.
\end{enumerate}

In particular, the quantity in Lemma~\ref{le:mar_mon} can be rewritten as follows:
$$
s_i-\sum_{j\not=i}M_{ij}\alpha_js_j
=s_i-\sum_{j\not=i,k}R_{ij}\alpha_js_j
-\frac{\alpha_kR_{ik}}{1+\alpha_kR_{kk}-\alpha_k}
(s_k-\sum_{j\not=i,k}R_{kj}\alpha_js_j).
$$

\label{le:M2R}
\end{lemma}
\ignore{
\begin{proof}
Using the Sherman-Morrison formula, we have
\begin{align}
\begin{split}
M&=[I-(I-A_{-\{i,k\}})P+\alpha_k\mathbf{e}_k\mathbf{e}_k^\top P]^{-1}
\\&
=R-\frac{\alpha_k}{1+\alpha_k\mathbf{e}_k^\top PR\mathbf{e}_k}
R\mathbf{e}_k\mathbf{e}_k^\top PR
\\&
=R-\frac{\alpha_k}{1+\alpha_k(PR)_{kk}}
R\mathbf{e}_k\mathbf{e}_k^\top PR.
\end{split}
\nonumber
\end{align}
We can compute that
$(R\mathbf{e}_k\mathbf{e}_k^\top PR)_{jh}=R_{jk}(PR)_{kh}$ for $j,h\in V$.
Then we have
$$
M_{jh}=R_{jh}-\frac{\alpha_kR_{jk}(PR)_{kh}}{1+\alpha_k(PR)_{kk}}.
$$
By Lemma \ref{le:PM_M}, we obtain
$$
M_{jh}=R_{jh}-\frac{\alpha_{k}R_{jk}R_{kh}}{1+\alpha_kR_{kk}-\alpha_k}
\text{  for $j,h\in V$ and $h\not=k$,}
$$
and
$$M_{jk}
=R_{jk}-\frac{\alpha_{k}R_{jk}(R_{kk}-1)}{1+\alpha_kR_{kk}-\alpha_k}
=\frac{R_{jk}}{1+\alpha_kR_{kk}-\alpha_k}
\text{   for $j\in V$.}
$$
as required.
\end{proof}
}

\begin{lemma}[Diagonal Entry]
Suppose $\alpha \in (0,1)^V$,
recall that $A_{-\{i,k\}} := \Diag(\alpha_{-\{i,k\}})$,
and $P$ corresponds to an irreducible interaction matrix.
For any $i,k\in V$ such that $i\not=k$, let $R=[I-(I-A_{-\{i,k\}})P]^{-1}$,
then $R_{ii}=\max_{j\in V}R_{ji}$.
Moreover, $R_{ii}=R_{ki}$ if and only if
$P_{kk} + P_{ki}=1$.
\label{le:col_max}
\end{lemma}

The following lemma gives the key insight for why local search works.
Intuitively, it shows that there does not exist any discrete ``saddle point''.
Even though its proof is technical, we still include it here because of its importance.

\begin{lemma}[Switching Lemma]
Recall that $f$ is defined in Definition \ref{defn:osp}
with an irreducible interaction matrix $P$,
and assume $|V| \geq 3$.
Suppose $\alpha,\beta\in (0,1)^{V}$
such that
$\Delta(\alpha,\beta)=\{i,k\}$
for some $i \neq k$.
Moreover, suppose further that
$$\min\{f(\alpha_i,\alpha_k), f(\beta_i,\beta_k)\}
< \min \{f(\alpha_i, \beta_k), f(\beta_i, \alpha_k)\}.$$
Then, we have
$$\max\{f(\alpha_i,\alpha_k), f(\beta_i,\beta_k)\}
> \min \{f(\alpha_i, \beta_k), f(\beta_i, \alpha_k)\}.$$
%
%
%
%
%

\label{le:alpha_beta_2d}
\end{lemma}

\begin{proof}
We prove the lemma
by contradiction.
Suppose
$$\max\{f(\alpha_i,\alpha_k), f(\beta_i,\beta_k)\}
\leq \min \{f(\alpha_i, \beta_k), f(\beta_i, \alpha_k)\}.$$
Without loss of generality,
suppose further that $f(\alpha_i, \alpha_k) \geq f(\beta_i, \beta_k)$.
Then, we have
\begin{equation}
f(\alpha_i,\alpha_k)\leq
\begin{cases}
f(\beta_i,\alpha_k)\\
f(\alpha_i,\beta_k)
\end{cases}
\text{and   }
f(\beta_i,\beta_k) <
\begin{cases}
f(\beta_i,\alpha_k)\\
f(\alpha_i,\beta_k).
\end{cases}
\nonumber
\end{equation}

We remark that it is important to distinguish between the strict and non-strict inequality.
We use the notation $f'_i$ to denote the partial derivative with respect to coordinate~$i$.

From $f(\alpha_i, \alpha_k) \leq f(\beta_i, \alpha_k)$ and
the fact that $f$ is marginally monotone (Lemma~\ref{le:mar_mon}) and $f'_i(x, \alpha_k)$
has the same sign in $\{-,0,+\}$ for $x \in (0,1)$,
we have
\begin{equation}
f'_i(\_, \alpha_k) \cdot
(\alpha_i-\beta_i)\leq 0.
\label{eq:1st}
\end{equation}

On the other hand, from the strict inequality $f(\beta_i, \beta_k) < f(\beta_i, \alpha_k)$,
we know the partial derivative $f'_k(\beta_i, y)$ must have the same non-zero sign in $\{-, +\}$,
again from Lemma~\ref{le:mar_mon}.  Therefore, we have:
\begin{equation}
f'_k(\beta_i, \_) \cdot
(\alpha_k-\beta_k)  > 0.
\label{eq:2nd}
\end{equation}

Similarly, $f(\alpha_i, \alpha_k) \leq f(\alpha_i,\beta_k)$
and $f(\beta_i, \beta_k) < f(\alpha_i,\beta_k)$ give the following:
\begin{equation}
f'_k(\alpha_i, \_) \cdot (\alpha_k-\beta_k) \leq 0,
\label{eq:3rd}
\end{equation}
\begin{equation}
f'_i(\_, \beta_k) \cdot (\alpha_i - \beta_i) > 0.
\label{eq:4th}
\end{equation}

\ignore{

Assume that $f(\alpha_i,\beta_k)>f(\beta_i,\beta_k)$.
Then, combined with the three given inequalities, we have
\begin{align}
\begin{split}
\frac{\partial f(\alpha_i,\alpha_k)}{\partial\alpha_i}
(\alpha_i-\beta_i)\leq 0,
\\
\frac{\partial f(\beta_i,\beta_k)}{\partial\beta_k}
(\alpha_k-\beta_k)>0,
\\
\frac{\partial f(\alpha_i,\alpha_k)}{\partial\alpha_k}
(\alpha_k-\beta_k)\leq 0,
\\
\frac{\partial f(\beta_i,\beta_k)}{\partial\beta_i}
(\alpha_i-\beta_i)>0.
\end{split}
\nonumber
\end{align}

}

Next, using Lemma \ref{le:mar_mon} and $R$ as defined in Lemma \ref{le:M2R},
the above inequalities (\ref{eq:1st}) to (\ref{eq:4th}) become:
\begin{align}
\begin{split}
&[s_i-\sum_{j\not=i,k}R_{ij}\alpha_js_j
-\frac{\alpha_kR_{ik}}{1+\alpha_kR_{kk}-\alpha_k}
(s_k-\sum_{j\not=i,k}R_{kj}\alpha_js_j)]
(\alpha_i-\beta_i)\leq 0,
\\
&[s_k-\sum_{j\not=i,k}R_{kj}\beta_js_j
-\frac{\beta_iR_{ki}}{1+\beta_iR_{ii}-\beta_i}
(s_i-\sum_{j\not=i,k}R_{ij}\beta_js_j)]
(\alpha_k-\beta_k)>0,
\\
&[s_k-\sum_{j\not=i,k}R_{kj}\alpha_js_j
-\frac{\alpha_iR_{ki}}{1+\alpha_iR_{ii}-\alpha_i}
(s_i-\sum_{j\not=i,k}R_{ij}\alpha_js_j)]
(\alpha_k-\beta_k)\leq 0,
\\
&[s_i-\sum_{j\not=i,k}R_{ij}\beta_js_j
-\frac{\beta_kR_{ik}}{1+\beta_kR_{kk}-\beta_k}
(s_k-\sum_{j\not=i,k}R_{kj}\beta_js_j)]
(\alpha_i-\beta_i)>0.
\end{split}
\nonumber
\end{align}

Recall that $\alpha_j=\beta_j$ for $j\not=i,k$.
Hence, we denote:
$$c_i :=s_i-\sum_{j\not=i,k}R_{ij}\alpha_js_j
=s_i-\sum_{j\not=i,k}R_{ij}\beta_js_j,$$
$$c_k :=s_k-\sum_{j\not=i,k}R_{kj}\alpha_js_j
=s_k-\sum_{j\not=i,k}R_{kj}\beta_js_j,$$
$$g_i(x) :=\frac{xR_{ki}}{1+xR_{ii}-x}
\text{ and }
g_k(x) :=\frac{xR_{ik}}{1+xR_{kk}-x}
.$$
Then, we have 
\begin{align}
[c_i-g_k(\alpha_k)c_k]
(\alpha_i-\beta_i)\leq&0,
\label{eq:gk_l_0}
\\
[c_k-g_i (\beta_i)c_i]
(\alpha_k-\beta_k)>&0,
\label{eq:gi_g_0}
\\
[c_k-g_i(\alpha_i)c_i]
(\alpha_k-\beta_k)\leq&0,
\label{eq:gi_l_0}
\\
[c_i-g_k(\beta_k)c_k]
(\alpha_i-\beta_i)>&0.
\label{eq:gk_g_0}
\end{align}

Observe that $c_i \neq 0$, otherwise
(\ref{eq:gi_g_0}) and (\ref{eq:gi_l_0}) contradict each other.
Similarly, $c_k \neq 0$,
otherwise (\ref{eq:gk_l_0}) and (\ref{eq:gk_g_0}) contradict each other.
We next argue that $c_i c_k > 0$.

From (\ref{eq:gk_l_0}) and (\ref{eq:gk_g_0})
we have
\begin{equation}
[c_i-g_k(\alpha_k)c_k][c_i-g_k(\beta_k)c_k]\leq 0.
\label{eq:both_gk_l_0}
\end{equation}

If $c_i c_k < 0$, then the above expression will be positive,
because $g_k(\cdot) \geq 0$ (we shall see that later).   Hence, we conclude that $c_i$ and $c_k$ have the same sign.

From (\ref{eq:gi_g_0}) and (\ref{eq:gi_l_0})
we have
\begin{equation}
[c_k-g_i (\beta_i)c_i][c_k-g_i(\alpha_i)c_i]\leq 0.
\label{eq:both_gi_l_0}
\end{equation}

Rearranging (\ref{eq:both_gk_l_0}) and (\ref{eq:both_gi_l_0}),
we have:
$$[\frac{c_i}{c_k}-g_k(\alpha_k)][\frac{c_i}{c_k}-g_k(\beta_k)]\leq 0,$$
and
$$[\frac{c_k}{c_i}-g_i (\beta_i)][\frac{c_k}{c_i}-g_i(\alpha_i)]\leq 0.$$

\ignore{

Now we are going to argue that $c_i$ and $c_k$ have the same sign
and $c_i,c_k\not=0$.
Observe that $g_i(x),g_k(x)>0$ for $x\in(0,1)$.
If $c_i$ and $c_k$ have the opposite signs,
then (\ref{eq:both_gk_l_0}) and (\ref{eq:both_gi_l_0}) should be greater than zero, which is a contradiction.
If $c_i=0$,
then (\ref{eq:gi_g_0}) and (\ref{eq:gi_l_0}) contradict each other.
If $c_k=0$,
then (\ref{eq:gk_l_0}) and (\ref{eq:gk_g_0}) contradict each other.
Thus, we conclude that $c_i$ and $c_k$ have the same sign
and $c_i,c_k\not=0$.
}


Note that every entry of $R$ is positive by Lemma \ref{le:PM_M} and we can easily prove $g_i(\cdot)$ and $g_k(\cdot)$ are both
strictly increasing functions in $[0,1]$.
Since $\alpha,\beta\in (0,1)^{V}$,
the above two inequalities imply that
$$0=g_k(0) < g_k(\min\{\alpha_k,\beta_k\}) \leq \frac{c_i}{c_k},$$
$$\frac{c_i}{c_k}\leq g_k(\max\{\alpha_k,\beta_k\})< g_k(1)=\frac{R_{ik}}{R_{kk}}\leq1,$$
and
$$0=g_i(0)< g_i(\min\{\alpha_i,\beta_i\})\leq\frac{c_k}{c_i},$$
$$\frac{c_k}{c_i}\leq g_i(\max\{\alpha_i,\beta_i\})< g_i(1)=\frac{R_{ki}}{R_{ii}}\leq1,$$
where $\frac{R_{ik}}{R_{kk}}\leq1$ and $\frac{R_{ki}}{R_{ii}}\leq1$ are from Lemma \ref{le:col_max}.

Notice that we get $0<\frac{c_i}{c_k} < 1$ and $0<\frac{c_k}{c_i} < 1$,
which is a contradiction.
Hence, the proof is completed.
\ignore{
This implies that $\frac{c_k}{c_i} = 1$, which means that
both $\frac{R_{ik}}{R_{kk}}$ and $\frac{R_{ki}}{R_{ii}}$ are equal to 1.
By Lemma~\ref{le:col_max},
this means that $P_{kk} + P_{ki} = P_{ik} + P_{ii} = 1$.
This implies that in the random walk corresponding to $P$,
the states $i$ and $k$ are isolated from other states,
contradicting the irreducibility of the random walk.}
%
%
%
%
\end{proof}

\begin{lemma}[Descending Coordinate]
Let $f$ be the function as defined in Definition~\ref{defn:osp}.
Suppose $\alpha,\beta \in (0,1)^V$
such that $f(\alpha)> f(\beta)$.
Then, there exists some $i \in \Delta(\alpha, \beta)$ and $\gamma \in (0,1)^V$
such that $\Delta(\alpha, \gamma) = \{i\}$,
$\Delta(\gamma, \beta) = \Delta(\alpha, \beta) \setminus \{i\}$,
and $f(\alpha) > f(\gamma)$.

In other words, by switching one coordinate (corresponding to $i$) of $\alpha$ to that of $\beta$,
the objective function $f$ decreases strictly.
%
%
\label{le:alpha_descend}
\end{lemma}

\begin{proof}
We prove the lemma by induction on $|\Delta(\alpha, \beta)|$.
The base case $|\Delta(\alpha, \beta)| = 1$ is trivial.

We consider the inductive step with
$|\Delta(\alpha, \beta)| = q$, for some $q \geq 2$.
Given a list $S$ of coordinates from $\Delta(\alpha, \beta)$,
we use $\alpha_{[S]}$ to denote the resulting
vector obtained from switching coordinates $S$ of $\alpha$
to those of $\beta$.

For contradiction's sake, we assume that for all $j \in \Delta(\alpha, \beta)$,
$f(\alpha_{[j]}) \geq f(\alpha)$; moreover,
we pick $i \in \Delta(\alpha, \beta)$ such that $f(\alpha_{[i]})$ is minimized.

Observe that $f(\alpha_{[i]}) \geq f(\alpha) > f(\beta)$
and $|\Delta(\alpha_{[i]}, \beta)| = q - 1$.  Therefore,
by the induction hypothesis,
there exists some $k \in \Delta(\alpha_{[i]}, \beta)$
such that $f(\alpha_{[i]}) > f(\alpha_{[i,k]})$.

Next, starting from $\alpha$, we shall fix all coordinates in $V$ except $i$ and $k$,
and we write the objective $f(x,y)$ as a function on only these two coordinates.

Observe that we have already assumed that
\begin{equation}
f(\alpha_i, \alpha_k) \leq \min\{f(\beta_i, \alpha_k), f(\alpha_i, \beta_k)\}.
\label{eq:ind_ass}
\end{equation}

Moreover, from above, we have $f(\alpha_{[i]}) > f(\alpha_{[i,k]})$,
which translates to $f(\beta_i, \alpha_k) > f(\beta_i, \beta_k)$.
Observe that we must have $f(\alpha_i, \beta_k) \leq f(\beta_i, \beta_k)$;
otherwise, we have $f(\beta_i, \beta_k) < \min\{f(\beta_i, \alpha_k), f(\alpha_i, \beta_k)\}$,
which, together with (\ref{eq:ind_ass}), will contradict Lemma~\ref{le:alpha_beta_2d}.

Therefore, we have $f(\alpha_{[k]}) = f(\alpha_i, \beta_k) \leq f(\beta_i, \beta_k)
< f(\beta_i, \alpha_k) = f(\alpha_{[i]})$,
which contradicts the choice of $i \in \Delta(\alpha, \beta)$ to minimize $f(\alpha_{[i]})$.
This completes the inductive step and also the proof of the lemma.
\end{proof}

\begin{corollary}
\label{cor:local_opt}
For the function $f$ in Definition~\ref{defn:osp}
every local minimizer is a global minimizer.
\end{corollary}

\begin{proof}
Suppose that $\alpha$ is a local minimizer, but
there is some~$\beta$
with $f(\alpha) > f(\beta)$.
Then, Lemma~\ref{le:alpha_descend}
implies there is some $\gamma$ with $|\Delta(\alpha, \gamma)| = 1$
such that $f(\alpha) > f(\gamma)$, contradicting
that $\alpha$ is a local minimizer.
%
%
\end{proof}

\section{Efficient Local Search}
\label{sec:efficient}

In Section~\ref{sec:structure},
we conclude in Corollary~\ref{cor:ext} that
it suffices to consider the extreme points
of the search space of resistance vectors.
Moreover, Corollary~\ref{cor:local_opt} states
that every local minimizer is a global minimizer.
Since we know how to compute
the sign of the partial derivative with respect
to each coordinate using Lemma~\ref{le:same_sign},
we can design a simple local search algorithm 
to find a global minimizer.

However, it is possible that $O(2^n)$ extreme points
are encountered before a global minimizer is reached.
Fortunately, in this section, we will explore
further properties of the objective function,
and design a local search algorithm that encounters
at most $O(n)$ extreme points before finding a global minimizer.

\subsection{Irrevocable Updates}
\label{sec:irrevocable}

\noindent \textbf{Local Search Strategy.} We shall start with the
upper bound resistance vector, i.e., for each $i \in V$, $\alpha_i = u_i$.
This induces the corresponding equilibrium opinion
vector $z(\alpha)$.
According to Lemma~\ref{le:same_sign},
if there is some agent~$i$ such that
$\alpha_i = u_i$ and $s_i - z_i(\alpha) > 0$,
then we should flip $\alpha_i$ to the lower bound $l_i$.
The following lemma shows that each $\alpha_i$ will be flipped at most once.
Essentially, we show that we will never encounter the situation
that there is some agent~$k$ such that $\alpha_k = l_k$ and $s_k - z_k(\alpha) < 0$,
in which case we would have to switch $\alpha_k$ back to $u_k$.

\begin{lemma}[Each Coordinate Flipped at Most Once]
\label{lemma:once}
Starting from the upper bound resistance vector,
the above local search strategy flips each $\alpha_i$ at most once.
\end{lemma}

\begin{proof}
We first show that for each agent~$k \in V$,
the quantity $s_k - z_k(\alpha)$ cannot decrease
when $\alpha$ is modified according to the local search strategy.
According to the strategy, $\alpha$ is modified because
there is some agent~$i$ such that
$\alpha_i = u_i$ and $s_i - z_i(\alpha) > 0$.
By Lemma~\ref{le:same_sign}, $\frac{\partial z_k}{\partial \alpha_i} > 0$
for each $k\in V$.
Hence, after $\alpha_i$ is switched from $u_i$ to $l_i$,
$z_k(\alpha)$ decreases, and the quantity $s_k - z_k(\alpha)$ increases.

Observe that if a coordinate $\alpha_k$ is ever flipped from $u_k$ to $l_k$,
this means that at that moment, we must have $s_k - z_k(\alpha) > 0$, which, as we have just shown, will
stay positive after $\alpha$ is subsequently updated.
\end{proof}

\subsection{Approximating the Equilibrium Vector}

Observe that in our local search algorithm,
we need to compute the equilibrium opinion vector $z(\alpha) = [I - (I - A)P]^{-1} A s$
for the current resistance vector $\alpha$, where $A = \Diag(\alpha)$.
However, computing matrix inverse is an expensive operation.  Instead, we
approximate $z(\alpha)$ using
the recurrence $z^{(0)} \in [0,1]^V$
and $z^{(t+1)} := A s + (I - A) P z^{(t)}$.
The following lemma gives an upper bound
on the additive error for each coordinate.

\begin{lemma}[Approximation Error]
\label{lemma:error}
Suppose for some $\epsilon > 0$,
for all $i \in V$, $\alpha_i \geq \epsilon$.
Then, for every $t \geq 0$,
$\| z(\alpha) - z^{(t)}\|_\infty \leq \frac{(1 - \epsilon)^{t}}{\epsilon}$.

\ignore{
For the special case $z^{(0)} = A s$,
we have for each $i \in V$,
$0 \leq z_i(\alpha) - z^{(t)}_i \leq 
\frac{(1 - \epsilon)^{t+1}}{\epsilon}$.}
\end{lemma}

\begin{proof}
Using the Neumann series 
$[I - (I - A)P]^{-1} = \sum_{j=0}^\infty [(I-A)P]^j$,
we have
$z(\alpha) - z^{(t)} = \sum_{j=t}^\infty [(I-A)P]^j A s   -  [(I-A)P]^t z^{(0)} $.

We next prove, by induction, that
for any $x \in [0,1]^V$,
 $\| [(I-A)P]^j x \|_\infty \leq (1 - \epsilon)^j$,
for all $j \geq 0$.  The base case $j=0$ is trivial because every coordinate of $x$ is between 0 and 1.
For the inductive step, assume that for some $j \geq 0$, every coordinate of $y = [(I-A)P]^j x$ has magnitude
at most $(1 - \epsilon)^j$.  Since $P$ is a row stochastic matrix, it follows
that $\|P y \|_\infty \leq (1 - \epsilon)^j$;
finally, since $\alpha_i \geq \epsilon$ for all $i \in V$,
we have $\| (I-A) P y \|_\infty \leq (1 - \epsilon)^{j+1}$, completing the induction proof.

Finally,
observing that both $\sum_{j=t}^\infty [(I-A)P]^j A s$
and $[(I-A)P]^t z^{(0)}$ have non-negative coordinates, we have
\begin{align}
\begin{split}
\|z(\alpha) - z^{(t)}\|_\infty
&\leq \max\{\|\sum_{j=t}^\infty [(I-A)P]^j A s\|_\infty,
\|[(I-A)P]^t z^{(0)}\|_\infty\}\\
&\leq \sum_{j=t}^\infty ( 1- \epsilon)^j
= \frac{(1 - \epsilon)^{t}}{\epsilon}
\nonumber
\end{split}
\end{align}
as required.
\end{proof}

\subsection{Local Search Algorithm}

Based on Lemmas~\ref{lemma:once} and~\ref{lemma:error},
we give a local search framework in Algorithm~\ref{alg:local}.
Observe that in line~\ref{ln:perturb},
we perturb the innate opinions~$s$ slightly to ensure
that for each resistance vector~$\alpha$ encountered,
no coordinate of $s$ and $z(\alpha)$ would coincide.

The \textbf{while} loop in line~\ref{ln:while} combines
local search to update $\alpha$ and estimation of 
the equilibrium vector $z(\alpha)$.  Here are two general update strategies,
which are both captured by the non-deterministic step in line~\ref{ln:non-determ}:

\begin{itemize}
\item \textbf{Conservative Update.} The opinion vector $z$ is iteratively
updated in line~\ref{ln:iter} until all coordinates of $z$ and $s$ are sufficiently far apart.
Then, for every coordinate~$\alpha_i$ such that $\alpha_i = u_i$ and $z_i < s_i$,
we flip $\alpha_i$ to the lower bound $l_i$.

After we update $\alpha$, we reset $t$ to 0, and continue to iteratively update~$z$.
Whenever we update $\alpha$ and set $t$ to 0, we say that a new \emph{phase} begins;
we use the convention that the initial phase is known as phase~$0$.

\item \textbf{Opportunistic Update.}  Instead of waiting for the approximation error
of every coordinate of $z$ to be small enough, we can update some coordinates $\alpha_i$,
if $\alpha_i = u_i$ and $z_i \leq s_i - \err(t)$ is small enough.  However,
there is some tradeoff between waiting for the errors of all coordinates to be small enough
and updating coordinates of $\alpha$ that are ready sooner.
In Section~\ref{sec:experiment}, we will evaluate empirically different update strategies.
\end{itemize}

\begin{algorithm}[!ht]
\caption{Local Search Framework} \label{alg:local}
\begin{small}
\KwIn{Innate opinions $s \in [0,1]^V$; interaction matrix $P$;
for each agent~$i \in V$, upper $u_i$ and lower $l_i$ bounds for resistance.
}
\KwOut{Optimal resistance vector $\alpha \in \times_{i \in V} \{l_i, u_i\}$.
}

\noindent \emph{(Technical step.)} Randomly perturb each coordinate of $s$ slightly. \label{ln:perturb}

Initially, for each agent $i$, set $\alpha_i \gets u_i$ to its upper bound;
denote $\epsilon_\alpha := \min_{i \in V } \alpha_i$.

Pick arbitrary $z \in [0,1]^V$, and set $t \gets 0$;
denote $\err(t) := \frac{(1-\epsilon_\alpha)^t}{\epsilon_\alpha}$.

\While{$\exists i \in V: |s_i - z_i| \leq \err(t) \vee (z_i < s_i \wedge \alpha_i = u_i)$}{\label{ln:while}

$z \gets A s + (I-A)Pz$, where $A = \Diag(\alpha)$. \label{ln:iter}

$t \gets t+ 1$.

\emph{(Non-deterministic step.)} Pick arbitrary $L \subseteq V$ (that can be empty)
such that for each $i \in L$, $z_i \leq s_i - \err(t)$ and $\alpha_i = u_i$.
\label{ln:non-determ}

\If{$L \neq \emptyset$}{

	\For{each $i \in L$}{
		Set $\alpha_i \gets l_i$ to its lower bound (and update $\epsilon_\alpha$).
	}
	
	$t \gets 0$.
}

}

\Return{Resistance vector $\alpha$.}
\end{small}
\end{algorithm}

\noindent \textbf{Optimistic Update.}  In both conservative and opportunistic updates, a coordinate $\alpha_i$ is flipped only when we know for sure that the current estimate $z_i$ has small enough error with respect to the equilibrium $z_i(\alpha)$; hence, no mistake in flipping any $\alpha_i$ is ever made.  However, our insight is that as the algorithm proceeds, the general trend is for every $z_i$ to decrease.

The first intuition is that if we start with some $z^{(0)}$ such that every coordinate of $z^{(0)}$ is at least
its equilibrium coordinate of  $z(\alpha)$, then $z^{(t)}$ should converge to $z(\alpha)$ from above.  The second observation is that
every time we flip some $\alpha_i$, this will not increase any coordinate of the equilibrium vector $z(\alpha)$, thereby
preserving the condition that the current estimate $z^{(t)} \geq z(\alpha)$.  Hence, without worrying about the accuracy
of the current estimate $z$, we will simply flip coordinate $\alpha_i$ to $l_i$ when $z_i$ drops below $s_i$.  However, it is possible that we might need to flip $\alpha_i$ back to $u_i$, if $z_i$ increases in the next iteration and becomes larger than $s_i$ again.  We shall see in Section~\ref{sec:experiment} that this scenario is extremely rare.  Specifically, in line~\ref{ln:rare} of Algorithm~\ref{alg:opt},
the set $J$ is (almost) always empty.

\begin{algorithm}[!ht]
\caption{Optimistic Local Search} \label{alg:opt}
\begin{small}
\KwIn{Innate opinions $s \in [0,1]^V$; interaction matrix $P$;
for each agent~$i \in V$, upper $u_i$ and lower $l_i$ bounds for resistance.
}
\KwOut{Optimal resistance vector $\alpha \in \times_{i \in V} \{l_i, u_i\}$.
}

\noindent \emph{(Technical step.)} Randomly perturb each coordinate of $s$ slightly. \label{ln:perturb}

Initially, for each agent $i$, set $\alpha_i \gets u_i$ to its upper bound;
denote $\epsilon_\alpha := \min_{i \in V } \alpha_i$.

Pick $z = (1, 1, \ldots, 1)$, and set $t \gets 0$;
denote $\err(t) := \frac{(1-\epsilon_\alpha)^t}{\epsilon_\alpha}$.

\While{$\exists i \in V: |s_i - z_i| \leq \err(t)$}{\label{ln:while}

$z \gets A s + (I-A)Pz$, where $A = \Diag(\alpha)$. \label{ln:iter2}

$t \gets t+ 1$.

\emph{(Optimistic Candidates.)} Set $L \gets \{i \in V:  z_i \leq s_i \wedge \alpha_i = u_i\}$.
\label{ln:opt}

\emph{(Rare Mistakes.)} Set $J \gets \{i \in V:  z_i > s_i \wedge \alpha_i = l_i\}$. \label{ln:rare}

\If{$L \cup J \neq \emptyset$}{

	\For{each $i \in L$}{
		Set $\alpha_i \gets l_i$ to its lower bound (and update $\epsilon_\alpha$).
	}
	
	\For{each $i \in J$}{
		Set $\alpha_i \gets u_i$ to its upper bound (and update $\epsilon_\alpha$).
	}

	$t \gets 0$.
}

}


\Return{Resistance vector $\alpha$.}
\end{small}
\end{algorithm}

\section{Heuristic Algorithms for Budgeted Opinion Susceptibility Problem}
\label{sec:heuristic}

\subsection{Marginal Greedy}
We propose the Marginal Greedy in Algorithm \ref{alg:marginal} which has similar framework as the greedy heuristic in \cite{AbebeKPT18} but employs the optimistic update strategy to approximate the equilibrium opinion vector $z(\alpha)$.

\begin{algorithm}[!ht]
\caption{Marginal Greedy} \label{alg:marginal}
\begin{small}
\KwIn{Innate opinions $s \in [0,1]^V$; initial resistance vector $\alpha^{(0)}$;
budget $k$; interaction matrix $P$;
for each agent~$i \in V$, upper $u_i$ and lower $l_i$ bounds for resistance.
}
\KwOut{The optimal set $T$ of agents with changed resistance and the corresponding resistance vector $\alpha$ with $\forall i \in T: \alpha_i \in \{l_i, u_i\}$
and $\forall i \in V \setminus T: \alpha_i = \alpha_i^{(0)}$.
}

\noindent \emph{(Technical step.)} Randomly perturb each coordinate of $s$ slightly.

Initialize the sum of equilibrium options $f \gets |V|$ and the set of agents $T \gets \emptyset$.

\For{$j \gets 1$ to $k$}{

Denote $\epsilon_\alpha := \min_{i \in V } \alpha_i^{(j)}$
and $\err(t) := \frac{(1-\epsilon_\alpha)^t}{\epsilon_\alpha}$.

\For{each $v\in V \setminus T$}{


Set $\alpha^{(j)}\gets \alpha^{(j-1)}$ and $\alpha_v^{(j)} \gets u_v$ to its upper bound (and update $\epsilon_\alpha$).

Set $z \gets (1, 1, \ldots, 1)$ and $t \gets 0$.

\While{$\exists i \in V: |s_i - z_i| \leq \err(t)$}{

$z \gets A s + (I-A)Pz$, where $A = \Diag(\alpha^{(j)})$.

$t \gets t+ 1$.

Set $L \gets \{i \in T\cup \{v\}:  z_i \leq s_i \wedge \alpha_i^{(j)} = u_i\}$
and $J \gets \{i \in T\cup \{v\}:  z_i > s_i \wedge \alpha_i^{(j)} = l_i\}$.

\If{$L \cup J \neq \emptyset$}{

	\For{each $i \in L$}{
		Set $\alpha_i^{(j)} \gets l_i$ to its lower bound (and update $\epsilon_\alpha$).
	}
	
	\For{each $i \in J$}{
		Set $\alpha_i^{(j)} \gets u_i$ to its upper bound (and update $\epsilon_\alpha$).
	}

	$t \gets 0$.
}

}

\If{$f > \sum_{i\in V} z_i$}{
Set $f \gets \sum_{i\in V} z_i$.

Update the selected agent $v' \gets v$ and the corresponding resistance vector $\alpha' \gets \alpha^{(j)}$.
}

}

Update the set of selected agents $T \gets T\cup \{v'\}$ and the corresponding resistance vector $\alpha^{(j)} \gets \alpha'$.

}


\Return{The set of agents $T$ and resistance vector $\alpha^{(k)}$.}
\end{small}
\end{algorithm}

\subsection{Batch Gradient Greedy}

We also give a gradient-based heuristic, called Batch Gradient Greedy (BGG), in Algorithm \ref{alg:batch_grad}.
The \textbf{while} loop in line \ref{ln:update_r_z_while} employs
the optimistic update strategy to approximate the equilibrium opinion vector $z(\alpha)$ (as well as $r(\alpha)$) until it is far apart enough from $s$ to enter the following procedures.

Observe that in Line \ref{ln:batch_while}, we introduce the batch approach to accelerate the algorithm.
When dealing with a large scale network and a large budget, we can set the batch size proportional to the budget, e.g. 1\% of budget, to limit the times to run the outer \textbf{while} loop.

From line \ref{ln:measure_begin} to \ref{ln:measure_end},
we consider a measure $\delta_i$ for each agent $i$, that is the partial derivative times the change of resistance, to decide which agent to include in the batch in order to maximize the decrease of the equilibrium.
Since we are using the optimistic update strategy to approximate $z$ and $r$, we need to estimate its lower $\delta_i^{(l)}$ and upper $\delta_i^{(u)}$ bounds based on Lemma \ref{lemma:error} and \ref{lemma:error2}.

Then in lines \ref{ln:pick_subset_begin} to \ref{ln:pick_subset_end}, we try to pick a subset of agents such that their minimal measure lower bound is greater than the maximal measure upper bound of the rest agents in $V \setminus T$,
i.e. to make sure that we select a batch of agents with the greatest measure.
Otherwise, we discard the subset and do one more update of $r$ and $z$ until we can find such a subset. 

\begin{algorithm}[!ht]
\caption{Batch Gradient Greedy (BGG)} \label{alg:batch_grad}
\begin{small}
\KwIn{Innate opinions $s \in [0,1]^V$; initial resistance vector $\alpha^{(0)}$;
budget $k$; batch size $b$, interaction matrix $P$;
for each agent~$i \in V$, upper $u_i$ and lower $l_i$ bounds for resistance; precision $\rho$.
}
\KwOut{The optimal set $T$ of agents with changed resistance and the corresponding resistance vector $\alpha$ with $\forall i \in T: \alpha_i \in \{l_i, u_i\}$
and $\forall i \in V \setminus T: \alpha_i = \alpha_i^{(0)}$.
}

\noindent \emph{(Technical step.)} Randomly perturb each coordinate of $s$ slightly.

Initialize the resistance vector $\alpha \gets \alpha^{(0)}$ and the set of agents $T \gets \emptyset$.

Denote $\epsilon_\alpha := \min_{i \in V } \alpha_i$
and $\err(t) := \frac{(1-\epsilon_\alpha)^t}{\epsilon_\alpha}$.

\While{$|T| < k$}{
Set $z \gets (1, 1, \ldots, 1)$; $r \gets (1, 1, \ldots, 1)$; $t \gets 0$. \label{ln:update_r_z_begin}

\While{$\exists i \in V: |s_i - z_i| \leq \err(t)$ \label{ln:update_r_z_while}}{

$r \gets \one + P^\top (I-A)r$ and $z \gets A s + (I-A)Pz$, where $A = \Diag(\alpha)$.

$t \gets t+ 1$.

Set $L \gets \{i \in T:  z_i \leq s_i \wedge \alpha_i = u_i\}$ and
$J \gets \{i \in T:  z_i > s_i \wedge \alpha_i = l_i\}$.

\If{$L \cup J \neq \emptyset$}{

	\For{each $i \in L$}{
		Set $\alpha_i \gets l_i$ to its lower bound (and update $\epsilon_\alpha$).
	}
	
	\For{each $i \in J$}{
		Set $\alpha_i \gets u_i$ to its upper bound (and update $\epsilon_\alpha$).
	}

	$t \gets 0$. \label{ln:update_r_z_end}
}

}

Set $\delta^{(u)} \gets (0, 0, \ldots, 0)$; $\delta^{(l)} \gets (0, 0, \ldots, 0)$; $T' \gets \emptyset$; $b' = \min\{b, k - |T|\}$.

\While{$|T'| < b'$ \label{ln:batch_while}}{

\For{each $i\in V \setminus T$ \label{ln:measure_begin}}{

(Compute upper $d_i^{(u)}$ and lower $d_i^{(l)}$ bound for the partial derivative.)

$d_i^{(u)} \gets [r_i + |V| \cdot \err(t)] \cdot \frac{|s_i - z_i| + \err(t)}{1 - \alpha_i}$.
and
$d_i^{(l)} \gets [r_i - |V| \cdot \err(t)] \cdot \frac{|s_i - z_i| - \err(t)}{1 - \alpha_i}$.


\If{$s_i \geq z_i$}{
$\delta_i^{(u)} \gets d_i^{(u)} \cdot (\alpha_i - l_i)$
and
$\delta_i^{(l)} \gets d_i^{(l)} \cdot (\alpha_i - l_i)$.
}\Else{
$\delta_i^{(u)} \gets d_i^{(u)} \cdot (u_i - \alpha_i)$
and
$\delta_i^{(l)} \gets d_i^{(l)} \cdot (u_i - \alpha_i)$.
\label{ln:measure_end}
}

}

Pick $T'$ to be the set of $b'$ agents $i$ in $V \setminus T$ with the largest $\delta^{(l)}_i$. \label{ln:pick_subset_begin}

\If{$\max_{i\in V \setminus (T\cup T')}\delta_i^{(u)}
> \min_{i\in T'}\delta_i^{(l)}$}{
$r \gets \one + P^\top (I-A)r$ and $z \gets A s + (I-A)Pz$, where $A = \Diag(\alpha)$.

$t \gets t+ 1$ and
$T' \gets \emptyset$. \label{ln:pick_subset_end}
}

}
\For{$i \in T'$}{
\If{$s_{i} \geq z_{i}$}{
Set $\alpha_{i} \gets l_{i}$ to its lower bound (and update $\epsilon_\alpha$).
}
\Else{
Set $\alpha_{i} \gets u_{i}$ to its upper bound (and update $\epsilon_\alpha$).
}
}
Update the set of selected agents $T \gets T\cup T'$.
}

{
\textbf{The code as Lines \ref{ln:update_r_z_begin} to \ref{ln:update_r_z_end} without updating $r$.}}

\Return{The set of agents $T$ and resistance vector $\alpha$.}
\end{small}
\end{algorithm}

\subsection{Approximating the Derivative Vector}
In Algorithm~\ref{alg:batch_grad}, we compute the partial derivative vector $d(\alpha)=\times_{i \in V}\frac{\partial f(\alpha)}{\partial \alpha_i}$ according to the equations in the proof of Lemma \ref{le:same_sign}, where

\begin{align}
\begin{split}
d_i(\alpha) = \frac{\partial f(\alpha)}{\partial \alpha_i}
=\frac{s_i - z_i(\alpha)}{1 - \alpha_i} \cdot \one^\top[I - (I - A)P]^{-1}e_i
&=\frac{s_i - z_i(\alpha)}{1 - \alpha_i} \cdot
\sum_{j=0}^{\infty}\one^\top[(I - A)P]^{j}e_i\\
&=\frac{s_i - z_i(\alpha)}{1 - \alpha_i} \cdot
r_i(\alpha)
\nonumber
\end{split}
\end{align}
where we define $r(\alpha)^\top = \one^\top[I - (I - A)P]^{-1}$
or $r(\alpha) = [I - P^\top (I - A)]^{-1}\one$.

\begin{lemma}[Approximation Error]
\label{lemma:error2}
Suppose for some $\epsilon > 0$,
for all $i \in V$, $\alpha_i \geq \epsilon$.
For $t \geq 0$,
define $r^{(t)} = \sum_{j=0}^{t}[P^\top(I - A)]^{j} \one$;
we denote $r(\alpha) = r^{(\infty)}$.

Then, for every $t \geq 0$,
$\| r(\alpha) - r^{(t)}\|_1 \leq n \cdot \frac{(1 - \epsilon)^{t+1}}{\epsilon}$.

\ignore{
For the special case $z^{(0)} = A s$,
we have for each $i \in V$,
$0 \leq z_i(\alpha) - z^{(t)}_i \leq 
\frac{(1 - \epsilon)^{t+1}}{\epsilon}$.}
\end{lemma}

\begin{proof}
We use the Neumann series 
$[I - P^\top(I - A)]^{-1} = \sum_{j=0}^\infty [P^\top(I - A)]^j$.

We have
$\|r(\alpha) - r^{(t)}\|_1 = \one^\top \sum_{j=t+1}^\infty [P^\top(I - A)]^j \one$.

We next prove, by induction, that
for all $j \geq 0$,
 $\one^\top [P^\top(I - A)]^j \one \leq n (1 - \epsilon)^j$.

The base case $j=0$ is trivial.

For the inductive step, assume that for some $j \geq 0$, 
$\one^\top [P^\top(I - A)]^j \one \leq n (1 - \epsilon)^j$.
Since $P$ is a row stochastic matrix, it follows
that $\one^\top P^\top  = \one^\top$.
Hence, 
$\one^\top [P^\top(I - A)]^{j+1} \one
= \one^\top (I-A) [P^\top(I - A)]^{j}  \one
\leq (1 - \epsilon) \one^\top [P^\top(I - A)]^{j} \one$,
where the inequality holds because every entry of the row vector
$\one^\top [P^\top(I - A)]^{j}$ is non-negative.
The inductive step is completed by using the induction hypothesis.

Finally,
we have
\begin{align}
\begin{split}
\| r(\alpha) - r^{(t)}\|_1 = \one^\top \sum_{j=t+1}^\infty [P^\top (I-A)]^j \one
\leq n \sum_{j=t+1}^\infty (1 - \epsilon)^j = 
n \cdot \frac{(1 - \epsilon)^{t+1}}{\epsilon},
\nonumber
\end{split}
\end{align}
as required.
\end{proof}

\section{Experiments}
\label{sec:experiment}

\noindent \textbf{Experimental Setup.}
Our experiments run on a server with 2.1 GHz Intel Xeon Gold 6152 CPU and 64GB of main memory.
The server is limited to activate at most 24 threads by the administrator.
The real network topologies we use in our experiment are shown in Table~\ref{tab:dataset};
we interpret each network as an undirected graph.
The number $n$ of nodes in the dataset networks ranges from about 1 million to 65 million; in each network,
the number $m$ of edges is around $2n$ to $30n$.

\vspace{10pt}

\begin{table*}[!ht]
\begin{center}
  \caption{Datasets Information for Different Networks}
  \label{tab:dataset}
  \begin{tabular}{lcccc}
    \toprule
    Name & Number $n$ of Nodes & Number $m$ of Edges & Source\\
    \midrule
    Residence hall& 217 & 2,672  & \cite{konect} \\
    Twitter& 548 & 3,638  & \cite{abir2019learning} \\
    Hamsterster& 1,788 & 12,476  & \cite{konect} \\
    Musae-twitch (PT)& 1,912 & 31,299  & \cite{snapnets} \\
    Facebook (NIPS)& 2,888 & 2,981  & \cite{konect} \\
    Advogato & 5,042 &  39,227 & \cite{konect} \\
    Chess & 7,115 &  55,779 & \cite{konect} \\
    Pretty Good Privacy & 10,680 & 24,316  & \cite{konect} \\
    DBLP & 12,495 &  49,563 & \cite{konect} \\
    Google+& 23,613 & 39,182  & \cite{konect} \\
    Facebook (WOSN)& 63,392 & 816,831 & \cite{konect} \\
    Catster & 148,826 & 5,447,464 & \cite{konect} \\
    com-Youtube & 1,134,890 & 2,987,624  & \cite{snapnets} \\
    com-LiveJournal & 3,997,962 & 34,681,189 & \cite{snapnets}\\
    LiveJournal & 10,690,276 & 112,307,385 & \cite{konect}\\
    com-Friendster & 65,608,366 & 1,806,067,135  & \cite{snapnets}\\
  \bottomrule
\end{tabular}
\end{center}
\end{table*} 

\noindent \textbf{Input Generation.}
For each dataset, we utilize the network topology and generate the input parameters as follows.
The innate opinion $s_i$ of each agent~$i$ is independently generated uniformly at random from $[0,1]$.
For each edge $\{i,j\}$ in the network, we independently pick $w_{ij}$ uniformly at random from $[0,1]$;
otherwise, we set $w_{ij} = 0$.  For $(i,j) \in V \times V$,
we normalize $P_{ij} := \frac{w_{ij}}{\sum_{k \in V} w_{ik}}$.
From Lemma~\ref{lemma:error}, one can see
that approximating the equilibrium opinions is more difficult when the resistance is low.
However, since we still want to demonstrate that
the resistance for each agent can have varied lower and upper bounds,
we set the lower bound $l_i$ of each agent~$i$ independently such that
with probability 0.99, $l_i$ equals $0.001$, and with probability 0.01, it is picked uniformly at random
from $[0.001,0.1]$.
Similarly, each upper bound $u_i$ is independently selected such that
with probability 0.99, it equals to 0.999,
and with probability 0.01, it is 
chosen uniformly at random from $[0.9,0.999]$.

\subsection{Update Strategies Comparison}

We compare the following three update strategies described in Section~\ref{sec:efficient}: conservative, opportunistic and optimistic.  For the three smaller networks (com-Youtube, com-LiveJournal, LiveJournal),
we apply all three update strategies.
For the largest network (com-Friendster), we only report the performance
of the optimistic update strategy, as the other two update strategies are
not efficient enough for such a large dataset.

\noindent \textbf{Experimental Setup.}
For fair comparison among the update strategies, we always initialize $z=(1,1,...,1)$.
To compare their performances, we plot a curve for each strategy.  The curves
have a common x-axis, which corresponds to the number of times
that the vector $z$ has been updated so far, i.e.,
the number of times line~\ref{ln:iter} (in both Algorithms~\ref{alg:local} and~\ref{alg:opt})
has been executed.  Since line~\ref{ln:iter} is the most time-consuming part
of the algorithms, it will be a suitable common reference.  We use the term \emph{iteration}
to refer to each time $z$ is updated.
For each update strategy,
we explain what is plotted for the y-axis.

\begin{itemize}
    \item \textbf{Conservative Strategy.} We run Algorithm \ref{alg:local}
		such that in line~\ref{ln:non-determ}, $L$ is non-empty only if
		for all $i \in V$, $|s_i - z_i| > \err(t)$,
		in which case, we pick $L$ to be the collection of 
		all $i$'s such that $\alpha_i = u_i$ and $z_i \leq s_i - \err(t)$.
		
		For the y-axis, we plot the ratio of agents~$i$ for which currently $\alpha_i = l_i$,
		or we know definitely that $\alpha_i$ should be switched to $l_i$, i.e., currently $\alpha_i = u_i$ and
		$z_i \leq s_i - \err(t)$.
		
		In Algorithm~\ref{alg:local}, the iterations (referring to each time $z$ is updated) are grouped into
		\emph{phases}, where a non-empty $L$ in an iteration marks the end of a phase.  Observe that
		at the end of a phase, for all $i \in L$, $\alpha_i$ is set to $l_i$ and $t$ is reset to 0.  Hence, in the next iteration, no coordinate $\alpha_i$
		is certain to be switched.  Hence, the curve has a step-like shape, where each plateau occurs after the end of each phase.
				
		Observer that initially $\epsilon_{\alpha}=\min_i u_i \geq 0.9$.  Hence, it takes very few number of iterations to satisfy
		$\forall i \in V: |s_i - z_i| > \err(t)$; we call this the phase~0.
    At the end of the phase~$0$, we set some $\alpha_i=l_i$ and $\epsilon_{\alpha}$ decreases significantly.
		Hence, subsequent phases have many more iterations.

		Observe that we can stop the iterative process, when for all $i \in V$, $|s_i - z_i| > \err(t)$, but there is no~$i \in V$
		such that $z_i < s_i$ and $\alpha_i = u_i$.  This marks the end of the curve.

    In each phase, we pick $L$ of line~\ref{ln:non-determ} as the collection of all $i$'s such that $z_i \leq s_i - \err(t)$
    only when $\forall i \in V: |s_i - z_i| > \err(t)$
    (otherwise, we pick $L=\emptyset$).
    Then, we set $\alpha_i = l_i$ for each $i\in L$ and $t=0$.
    We call such a phase a \emph{conservative phase}.
    
    \item \textbf{Opportunistic Strategy.} We run Algorithm \ref{alg:local} similarly as before.
		Phase~0 is the same as the conservative strategy; we call a phase \emph{conservative},
		if it follows the conservative update strategy.
		
		Starting from phase~1, we can perform it in an opportunistic manner as follows.
		Recall that at the beginning of a phase, $t$ has just been reset to 0.
		At the $t$-th iteration of that phase, we use $L(t)$ to denote 
		the collection of $i$'s such that $z_i \leq s_i - \err(t)$.
		For every 1000 iterations, we compute an estimate $\mathsf{slope}(k) := \frac{|L(1000k)|-|L(1000(k-1))|}{1000}$
		of the slope; we keep track $k_m$ of the maximum slope computed so far.
		After some iteration, if the estimated slope drops below some factor (we use 0.1 in our experiments) of $k_m$,
		then we end this phase.  Intuitively, each additional iteration flips only
		a small number of coordinates~$\alpha_i$, and hence, one would like to end this phase.  We call
		such a phase \emph{opportunistic}.
		
		Since typically the total number of phases is around 8, we run phase~1 to 6 opportunistically, after which
		we run the remaining phases conservatively to make sure that 
		all coordinates $\alpha_i$ that need to be changed will be flipped.	
			
			\ignore{
			Within phase~1, we get the index at which the  maximum slope
		is attained: $k_m := \arg \max_k \mathsf{slope}(k)$.
		
		Then, we obtain the iteration at which the slope falls below some factor (say 0.1) of $\mathsf{slope}(k_m)$,
		i.e., $t_e :=  1000 \times \min \{ k>k_{m} : \mathsf{slope}(k) < 0.1 \cdot \mathsf{slope}(k_{m})\}$.
		Intuitively, this means that after iteration $t_e$ in phase~$0$, each additional iteration flips only
		a small number of coordinates~$\alpha_i$, and hence, one might wish to end this phase earlier.
		
		\hubert{Actually, even phase~1 can be run differently.  We keep track of maximum slope $s_m$ found so far in this phase.
		Then, we can stop phase~1 if the current slope drops below $0.1 \cdot s_m$.}

		Therefore, we run phase~2 differently.  Instead of completing phase~2 as in conservative strategy,
		we just run phase~2 for $t_e$ iterations, after which we flip all coordinates $\alpha_i$ that
		satisfies $\alpha_i = u_i$ and $z_i \leq s_i - \err(t_e)$.  We call phase~2 an \emph{estimated} phase.
		
		We will run phase 3 conservatively, in which 
		a similar parameter $t_e$ is computed for running estimated phase~4.
		
		Since the total number of phases is usually around 8, after phase~4 we run all remaining phases conservatively to ensure
		all coordinates that need to be changed will be flipped.}
		
		As in the conservative update strategy, for the y-axis, we plot the ratio of coordinates $\alpha_i$ that currently $\alpha_i = l_i$,
		or we know for sure should be switched to $l_i$, i.e.
        $\alpha_i = u_i$ and
		$z_i \leq s_i - \err(t)$.
		
    %
    %
    
    \item \textbf{Optimistic Strategy.} We implement Algorithm~\ref{alg:opt}, where in each iteration after $z$ is updated, a coordinate $\alpha_i$ is (re)set to $l_i$ if $z_i < s_i$, and (re)set to $u_i$ if $z_i > s_i$.  For the y-axis, we plot the ratio of coordinates that currently take their lower bounds.  The curve ends when enough iterations are performed after some coordinate of $\alpha$ is last updated, in order to ensure that the estimate $z$ is close enough to the equilibrium vector according to Lemma~\ref{lemma:error}.

\end{itemize}

\ignore{
We use the accumulative times $T$ to update $z$ to evaluate a strategy's efficiency, since updating $z$ is the most time-consuming operation.
Let $r(T)$ be the ratio of coordinates of $\alpha$ that should be set to their lower bounds after the $T$th update of $z$ in the algorithm.}

\noindent \textbf{Experiment Results.}
Each of Figures~\ref{fig:1M_ratio}, \ref{fig:4M_ratio} and \ref{fig:10M_ratio} shows the plots for the three strategies
in the corresponding network (com-Youtube, com-LiveJournal or LiveJournal).
Figure~\ref{fig:65M_ratio} shows the plot of the optimistic strategy in the com-Friendster network,
where the other two strategies are not efficient enough for such a large network.
%
%
As expected, the opportunistic strategy is slightly better than the conservative strategy.
From the positions of the plateaus,
we can see that the initial opportunistic phases end sooner than their conservative counterparts.
Hence, overall the opportunistic strategy performs slightly better than the conservative strategy;
in increasing sizes of the three tested networks, the numbers of iterations taken by the opportunistic
strategy are $79.2\%$, $77.9\%$ and $71.5\%$, respectively, of those taken by the conservative strategy.

On the other hand,
the optimistic strategy can achieve the optimal resistance vector
with much fewer number of iterations than
the other two strategies.
In increasing sizes of the three smaller networks,
the numbers of iterations taken by
the optimistic strategy are only $12.8\%$, $13.4\%$ and $12.4\%$, respectively, of those taken by the conservative strategy.
Moreover,
the optimistic strategy makes very few mistakes; in increasing sizes of the four networks,
the numbers of times coordinates are flipped from lower bounds back to upper bounds are 1, 0, 13 and 168,
which are negligible for networks with millions of nodes.

\ignore{
\hubert{edit below paragraph later when the results are available.}

We can see that there are some terraces in the curves of the conservative and opportunistic strategies.
These terraces appears right after $t\leftarrow 0$, thus they can serve as the sign of the beginning and end of a phase.
Compared with the conservative strategy, the opportunistic strategy can save some times (the terrace is shorter) in the estimated phase while having smaller $r(T)$ at first.
But we can see that later on, the opportunistic strategy enters those last phases sooner than the conservative strategy and has higher $r(T)$ under the same $T$.
However, the opportunistic strategy can only save a small amount of update times.
It takes only $?\%$, $?\%$ and $?\%$ of the update times of the conservative strategy for the com-Youtube, com-LiveJournal and LiveJournal networks respectively.
}

\begin{figure}[!ht]
\begin{center}
\includegraphics[width=3.2in]{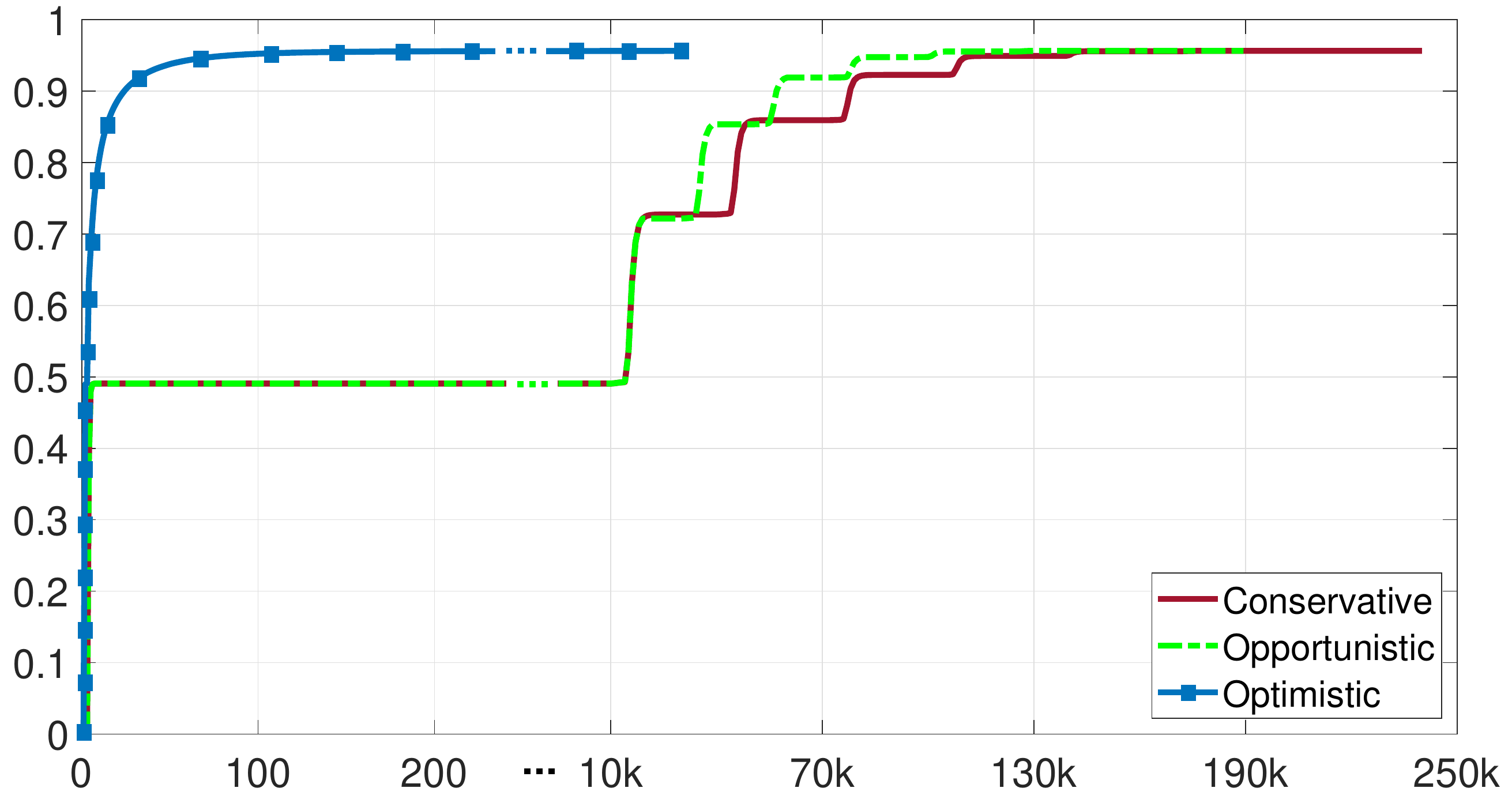}
\caption{Update Strategies Comparison on com-Youtube.}
\label{fig:1M_ratio}
\end{center}
\end{figure}

\begin{figure}[!ht]
\begin{center}
\includegraphics[width=3.2in]{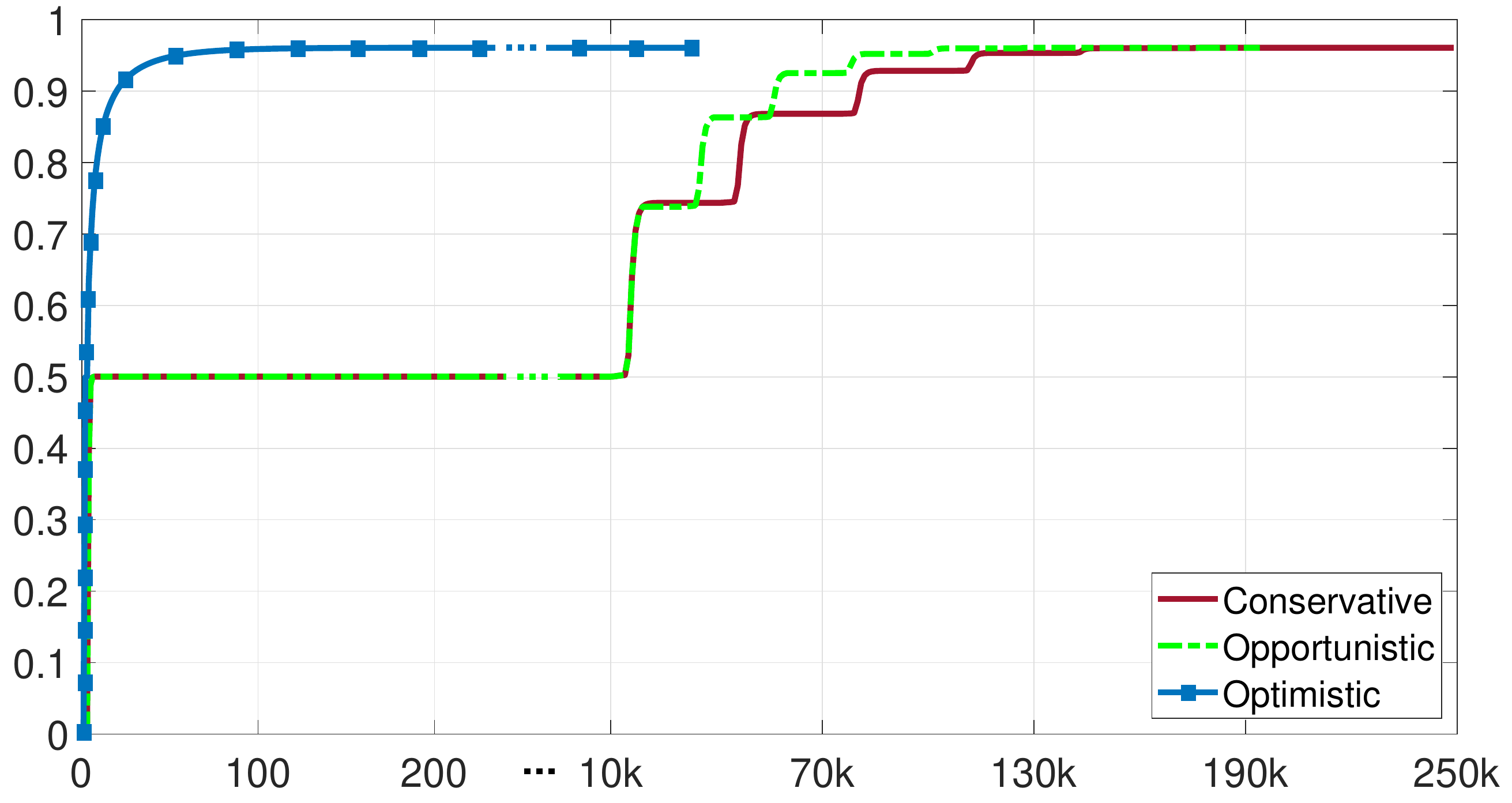}
\caption{Update Strategies Comparison on com-LiveJournal.}
\label{fig:4M_ratio}
\end{center}
\end{figure}

\begin{figure}[!ht]
\begin{center}
\includegraphics[width=3.2in]{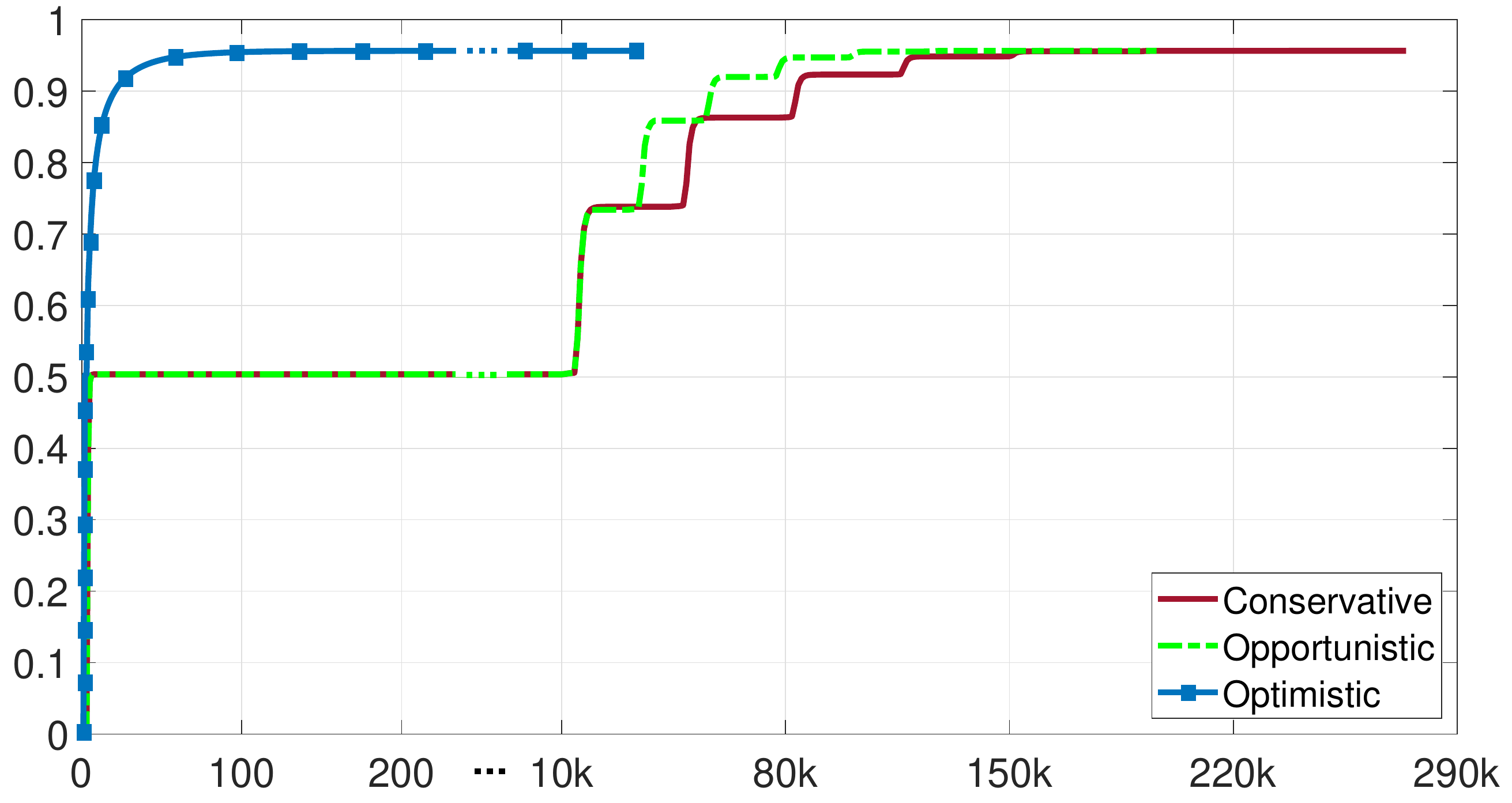}
\caption{Update Strategies Comparison on LiveJournal.}
\label{fig:10M_ratio}
\end{center}
\end{figure}

\begin{figure}[!ht]
\begin{center}
\includegraphics[width=3.2in]{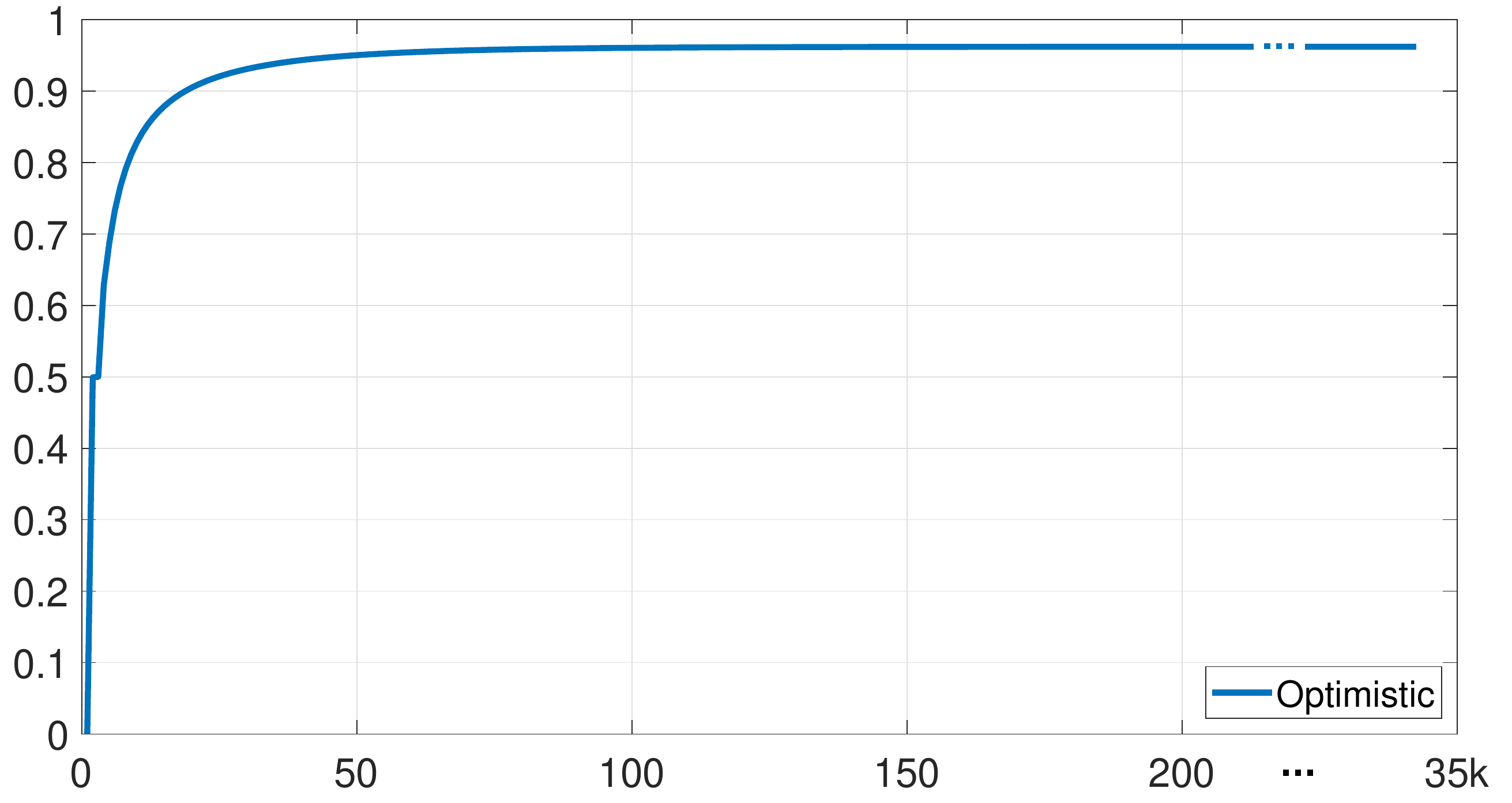}
\caption{Optimistic Strategy on com-Friendster.}
\label{fig:65M_ratio}
\end{center}
\end{figure}

\subsection{Running Time with Multiple Threads}

We compare the actual running time using different number of threads
for the optimistic strategy on only the three smaller networks, since
the largest network takes too long using only one thread.  Using all 24 available threads,
running the optimistic strategy on the com-Friendster network already takes around 50 hours.

The three bar graphs in Figure~\ref{fig:thread} show the
 running time (measured in minutes) for running the optimistic strategy with different number of threads
on the com-Youtube, com-LiveJournal and LiveJournal networks.
Since updating $z$ (line~\ref{ln:iter} of Algorithm~\ref{alg:opt}) is the most time-consuming part of the 
algorithm, the fact that it is readily parallelizable supports the empirical results that
using multiple threads can greatly reduce the running time, where the effect is more prominent for larger networks.

\begin{figure}[h]
    \centering
    \subfigure[]{
        \includegraphics[width=1.03in]{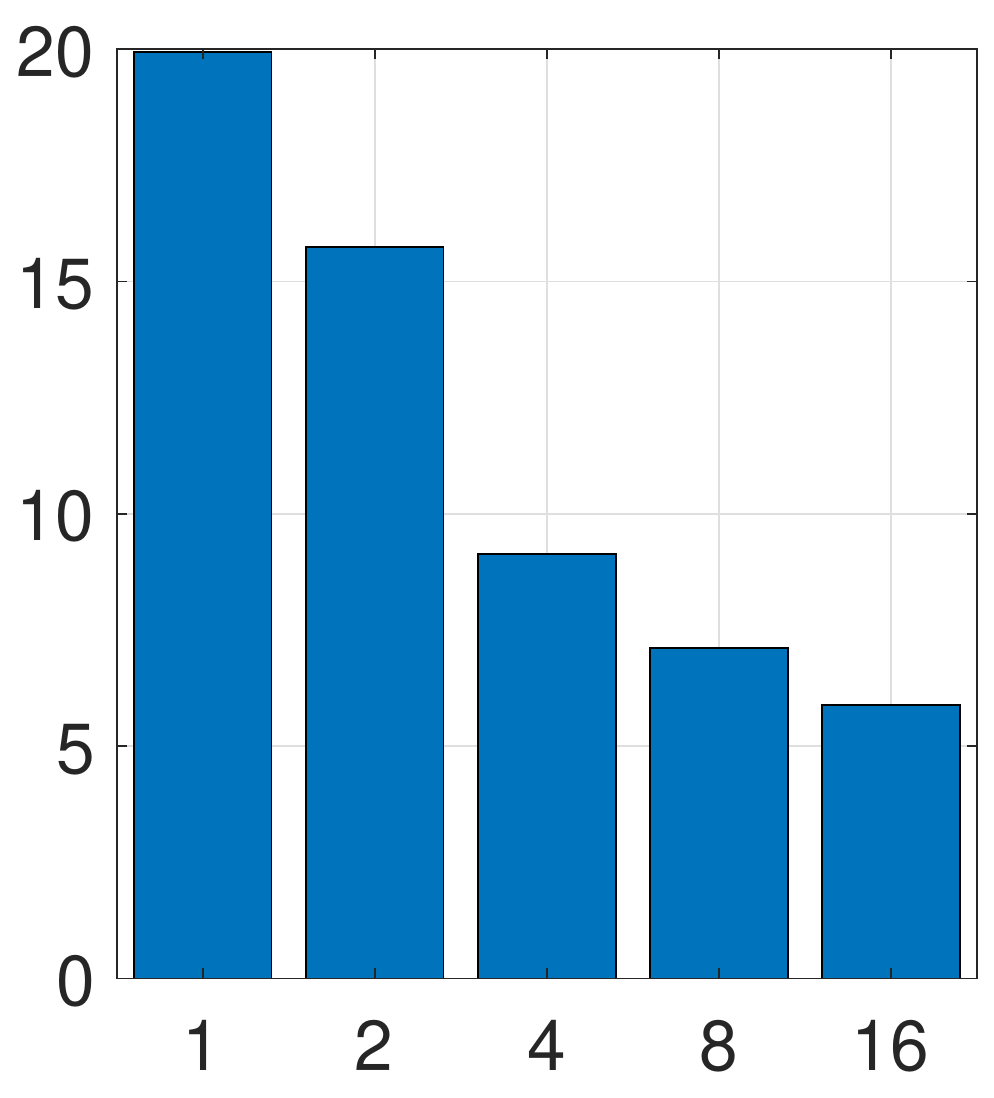}
		}
    \subfigure[]{
        \includegraphics[width=1.03in]{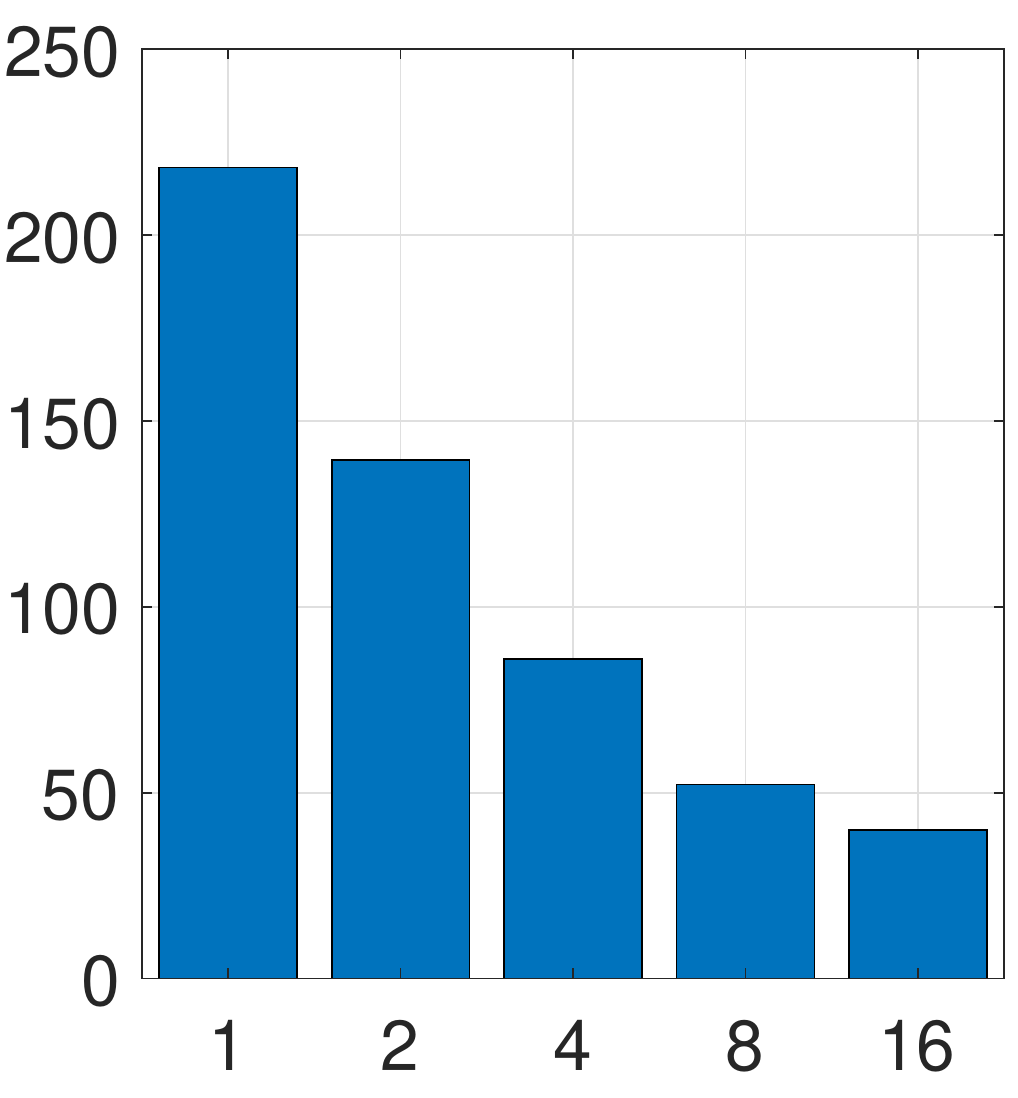}
    }
    \subfigure[]{
        \includegraphics[width=1.03in]{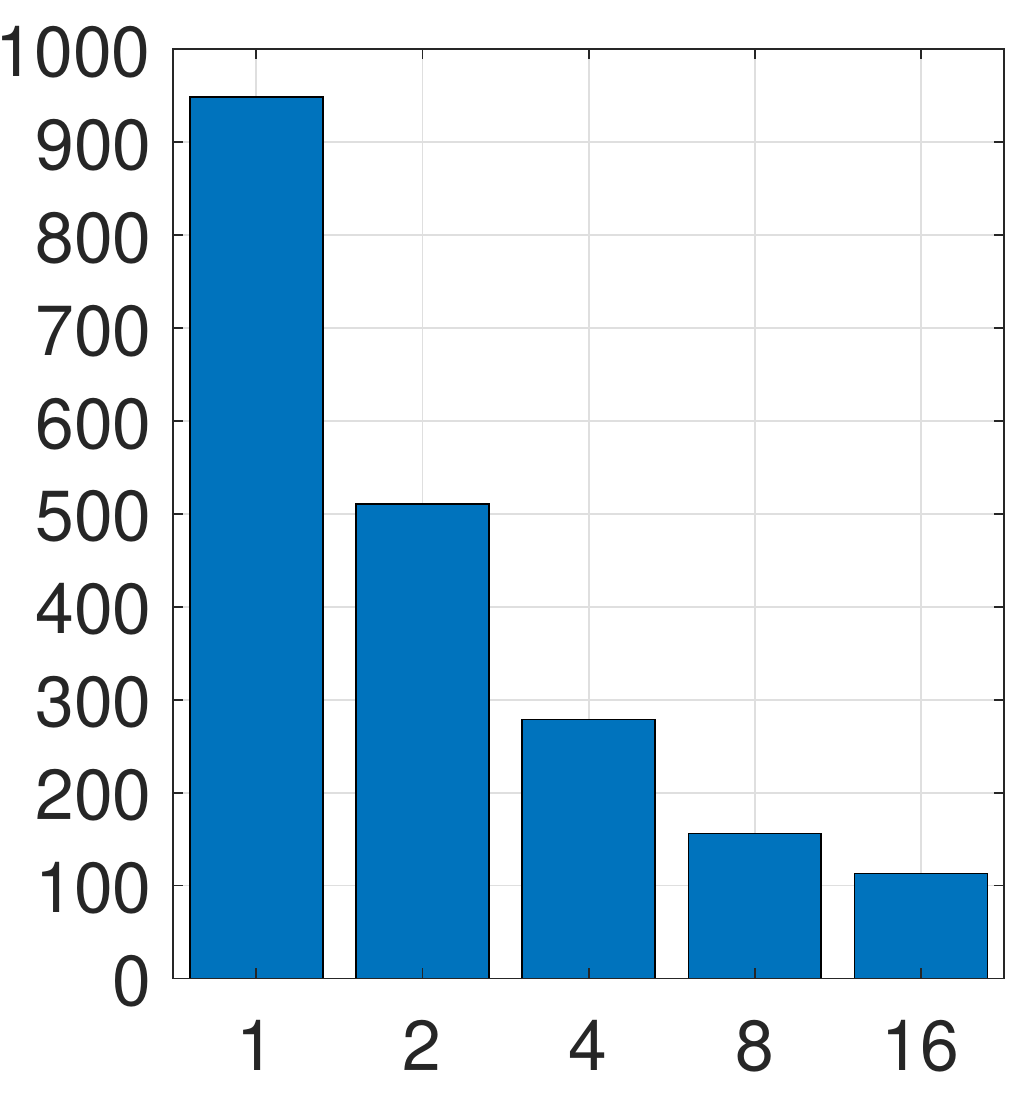}
    }
    \caption{Running Time with Different Number of Threads on
		com-Youtube, com-LiveJournal and LiveJournal}
    \label{fig:thread}
\end{figure}

\ignore{
\noindent \textbf{Multithreading implementation.}
Is this section, we focus on the most efficient optimistic strategy.
In Algorithm~\ref{alg:opt}, updating $z$ in line~\ref{ln:iter} is the most time-consuming operation and will be executed many times during the whole process.
Thus, we employ multithreading to update $z$ in order to speed up our program.
We consider the time required by the \textbf{while} loop as the program's running time.
Figure and compare the running time with different threads for the com-Youtube, com-LiveJournal and LiveJournal networks respectively.
}

\section{Experiments for Budgeted Variant}
\label{sec:experiment_budgeted}
\noindent \textbf{Experimental Setup.}
We conduct the experiments on a server with 3.4 GHz Intel(R) Core(TM) i5-3570 CPU and 16GB of main memory.
The server can activate at most 4 threads.
The real network topologies we use are also shown in Table~\ref{tab:dataset}.

\noindent \textbf{Input Generation.}
In each instance of the following experiments, we generate a setup of $s, P, u, l$ and $\alpha^{(0)}$ randomly in a similar way to that in Section \ref{sec:experiment}.
Particularly, the initial resistance $\alpha^{(0)}_i$ of each agent $i$
is independently selected uniformly at random from $[l_i, u_i]$ after $l_i$ and $u_i$ are generated.

\subsection{Agent Selection Strategies Comparison}
\label{sec:agent_selection_comp}

\begin{figure}[h]
    \centering
    \subfigure[]{
        \includegraphics[width=2.65in]{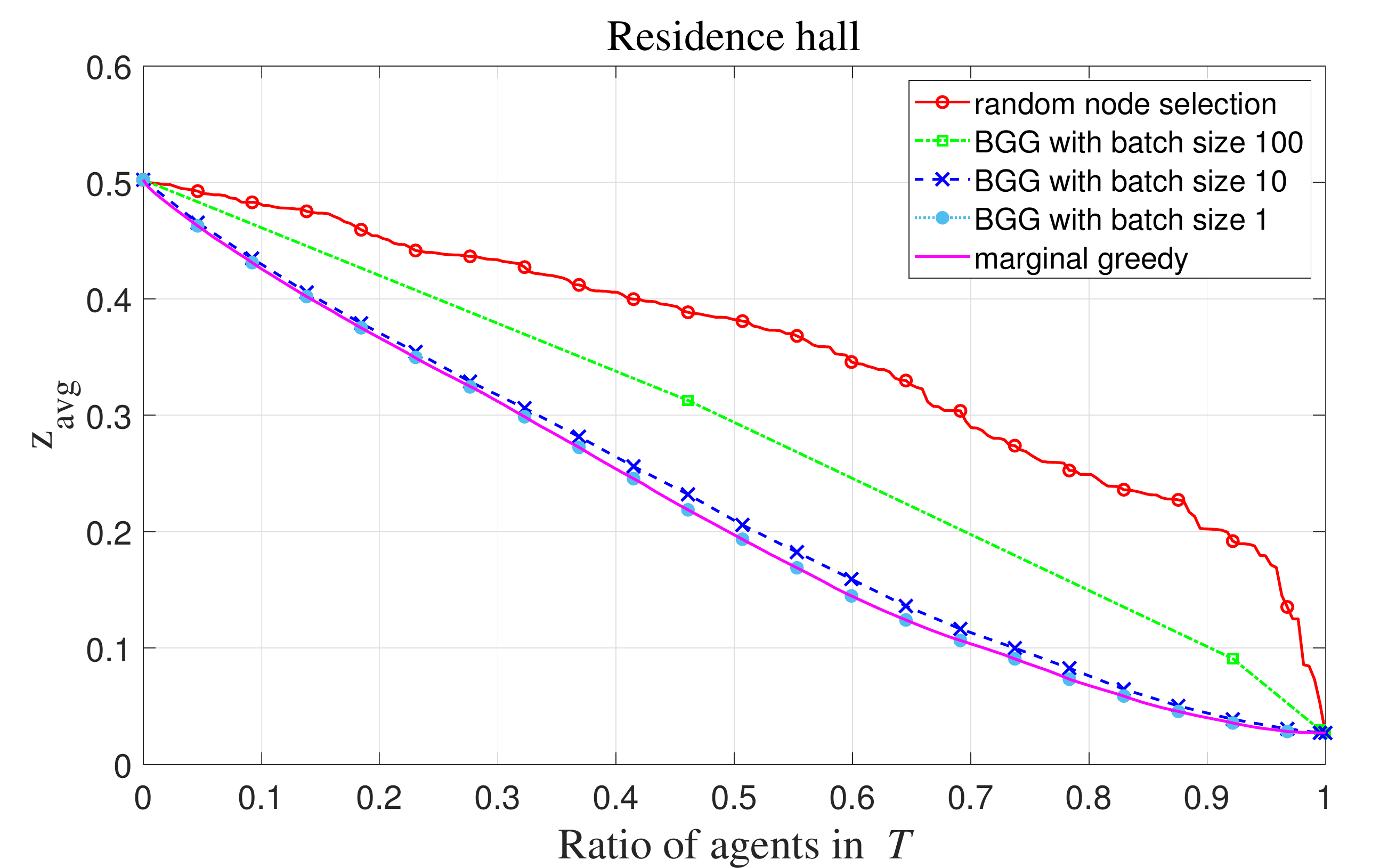}
		}
    \subfigure[]{
        \includegraphics[width=2.65in]{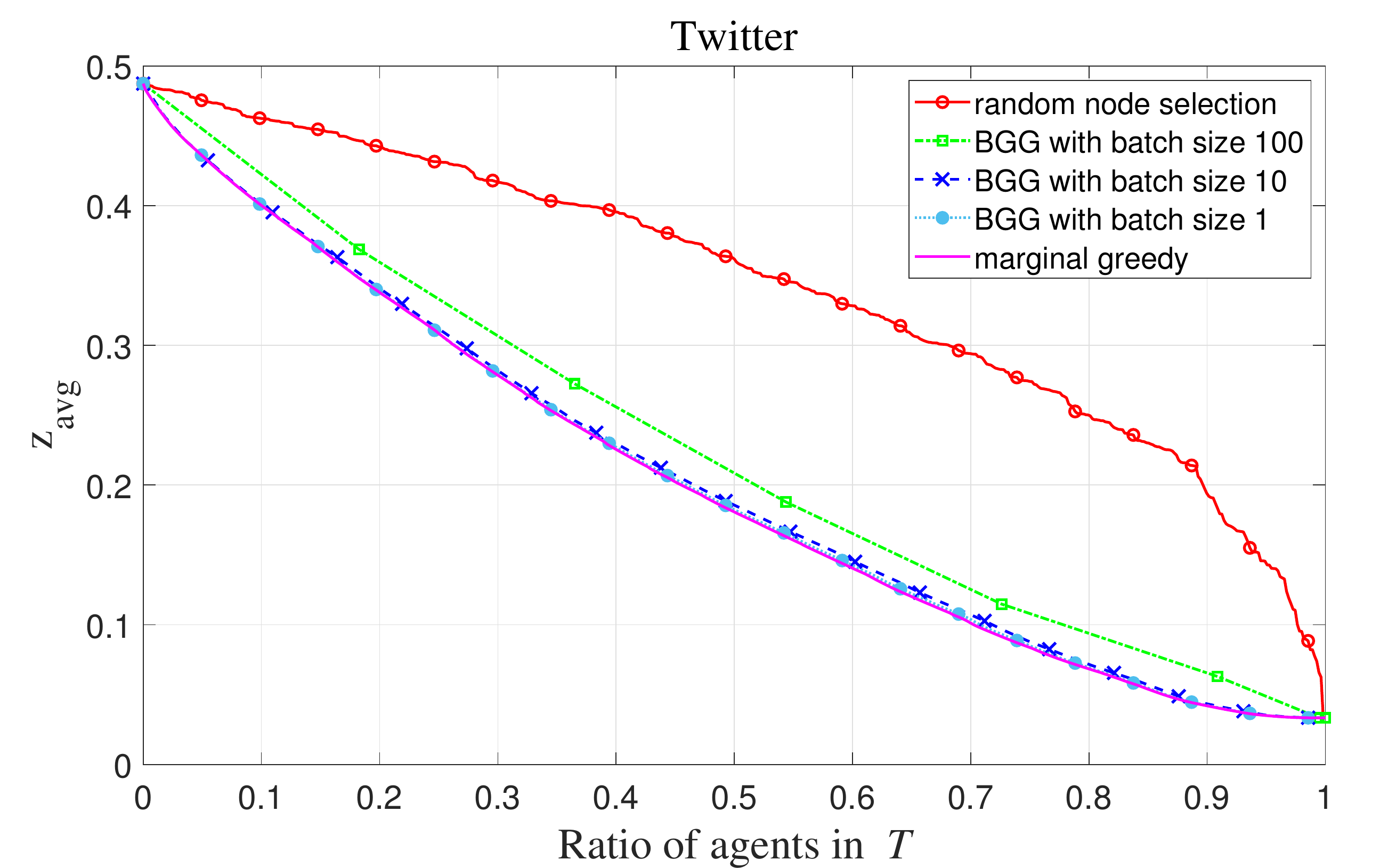}
    }
    \subfigure[]{
        \includegraphics[width=2.65in]{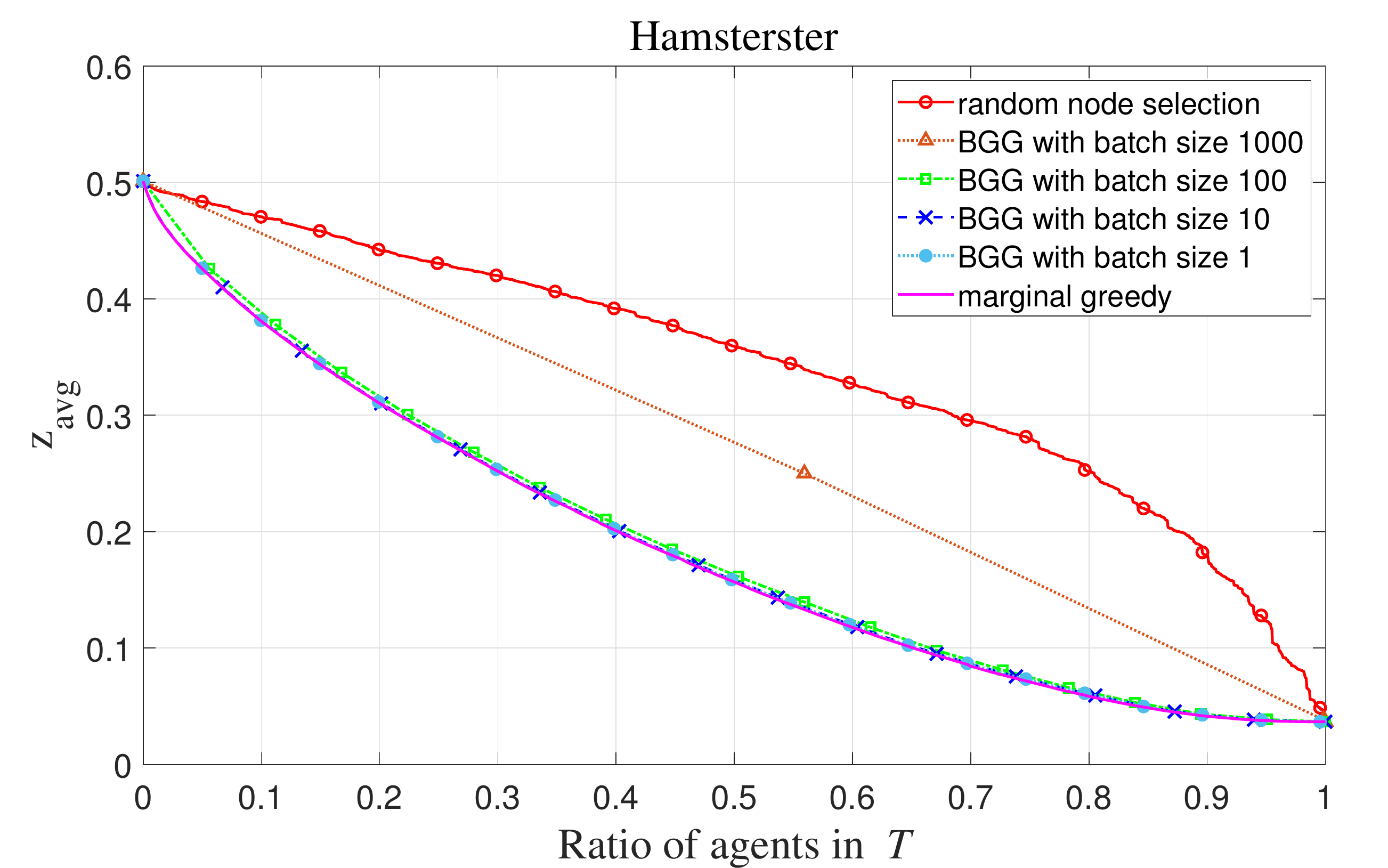}
    }
    \subfigure[]{
        \includegraphics[width=2.65in]{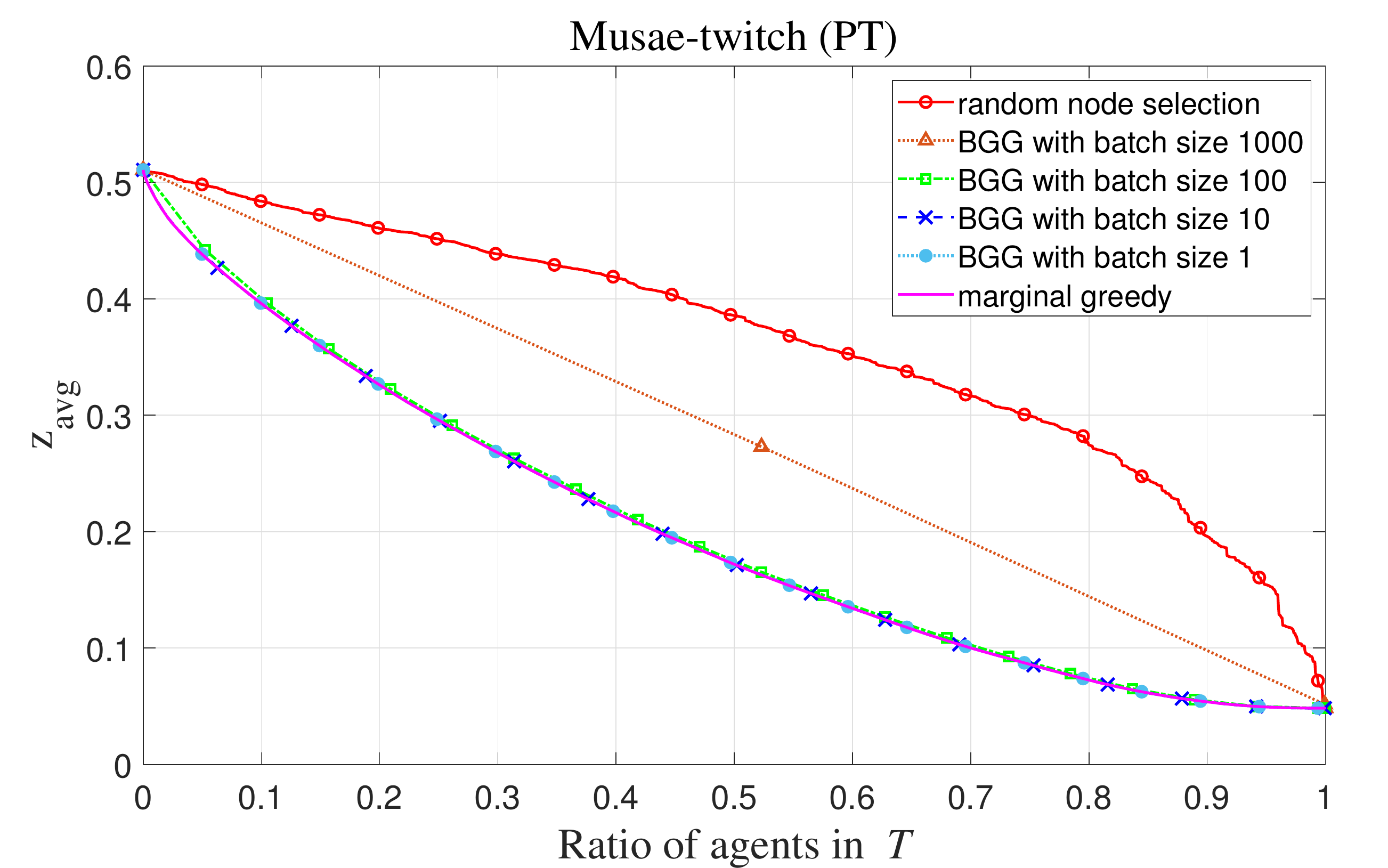}
    }
    \subfigure[]{
        \includegraphics[width=2.65in]{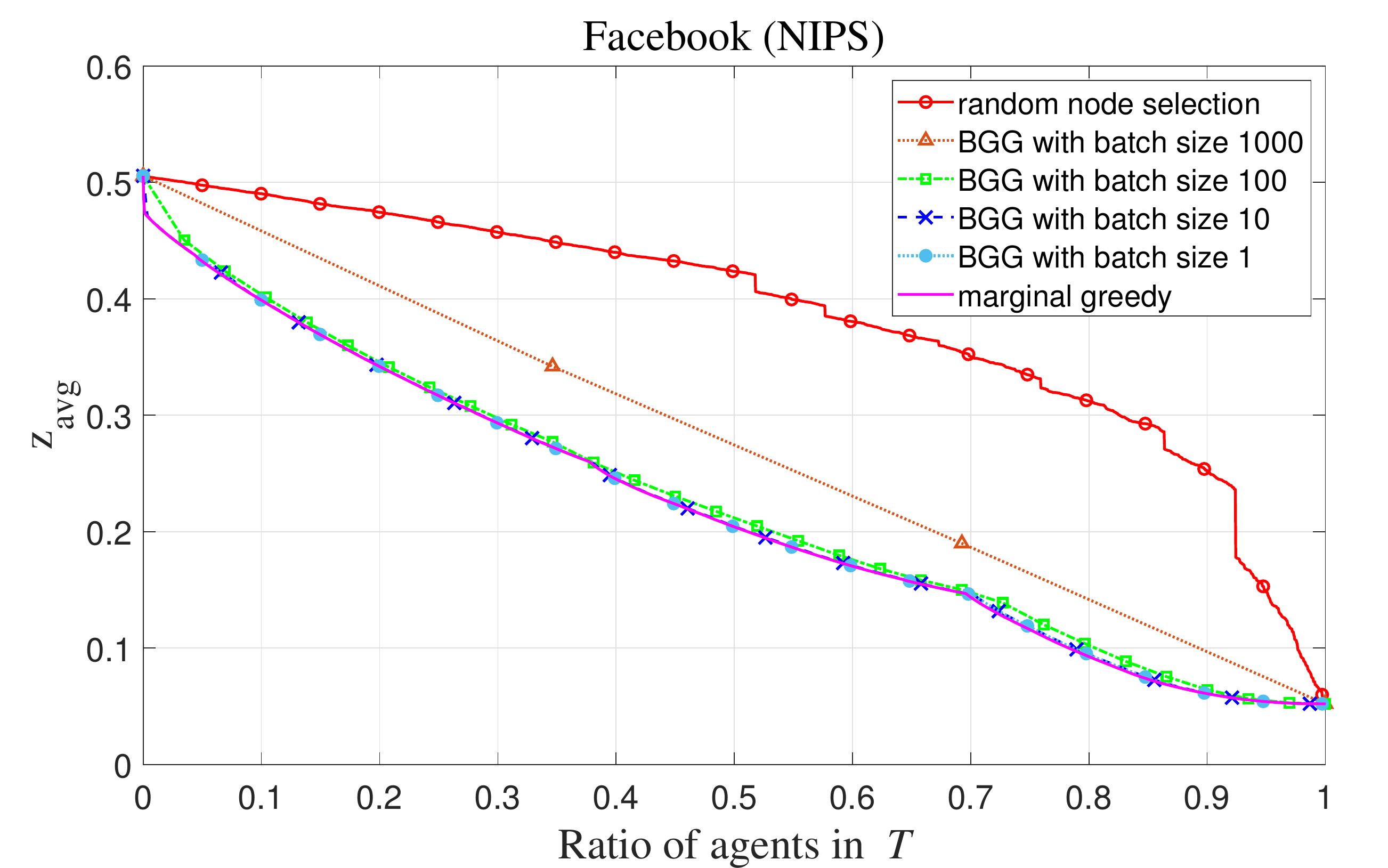}
    }
    \caption{Average Equilibrium Opinion on Small Networks.}
    \label{fig:small}
\end{figure}

As shown in Algorithm \ref{alg:marginal} and \ref{alg:batch_grad},
we have two heuristic strategies, Marginal Greedy and Batch Gradient Greedy (BGG), to select new agent into $T$. 
We also run experiments on a trivial random node selection strategy,
where it selects a new agent uniformly at random from $V\setminus T$ into $T$ until the budget is satisfied.

To compare their performance, we give the average equilibrium opinion $z_{avg}$ employing these three agent selection strategies on small networks in Figure \ref{fig:small}.
For fair comparison, we use the same setup of $s, P, u, l$ and $\alpha^{(0)}$ on the same network.
In each graph, one curve represents one strategy.
The curves share the same x-axis, which corresponds to the ratio of agents selected in $T$.

Observe that Marginal Greedy and Batch Gradient Greedy with batch size 1 have almost the same $z_{avg}$ with the change of ratio of agents in $T$, which implies that they share similar performance.
While the random node selection strategy performs the worst among them.
We also run experiments using Batch Gradient Greedy with different constant batch sizes.
When the batch size is small enough relative to the number of agents in the network, Batch Gradient Greedy have similar performance to Marginal Greedy. 
We will show more results on choosing different batch sizes in the next section.

\subsection{Batch Size Comparison}

\begin{figure}[h]
    \centering
    \subfigure[]{
        \includegraphics[width=2.65in]{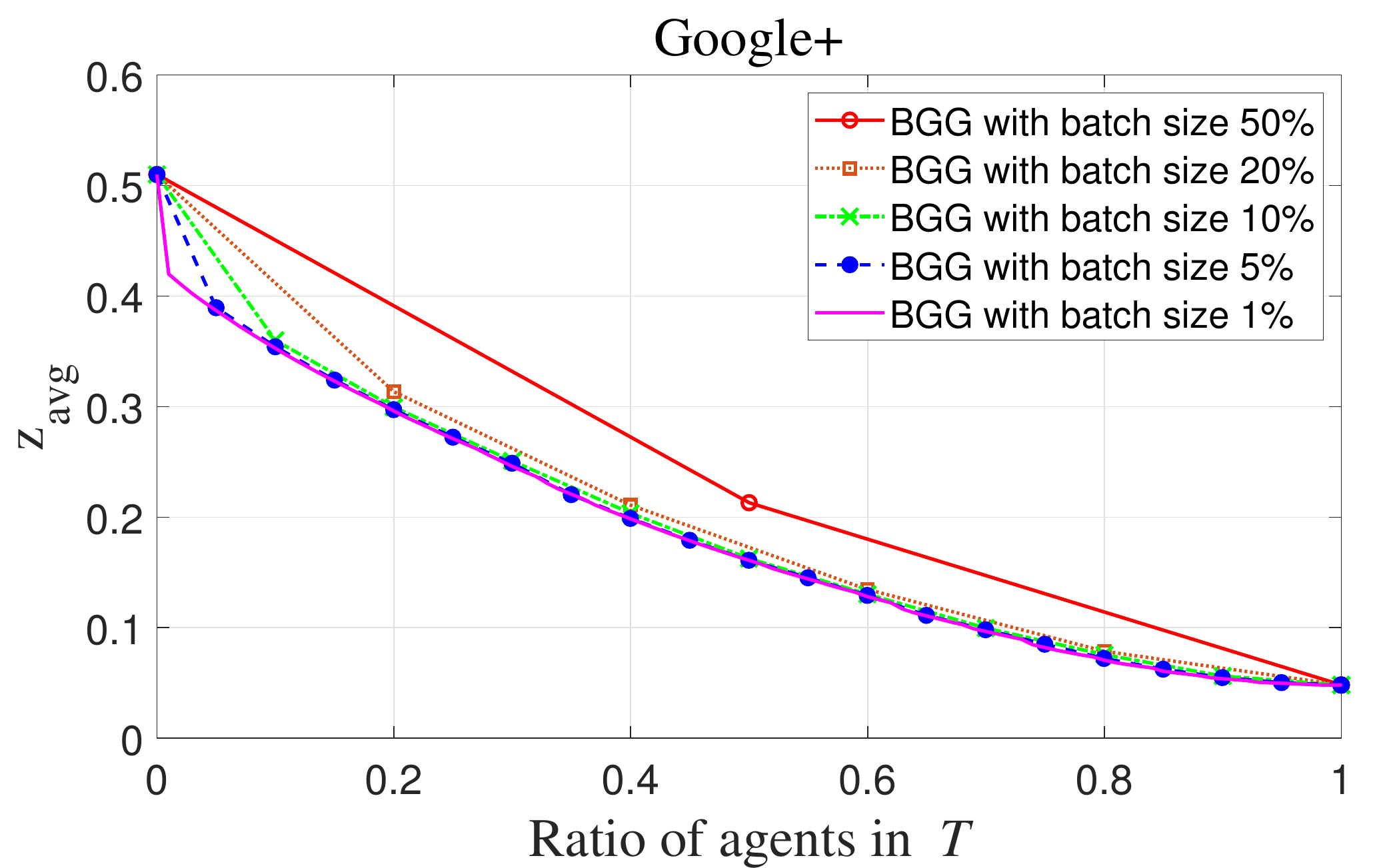}
		}
    \subfigure[]{
        \includegraphics[width=2.65in]{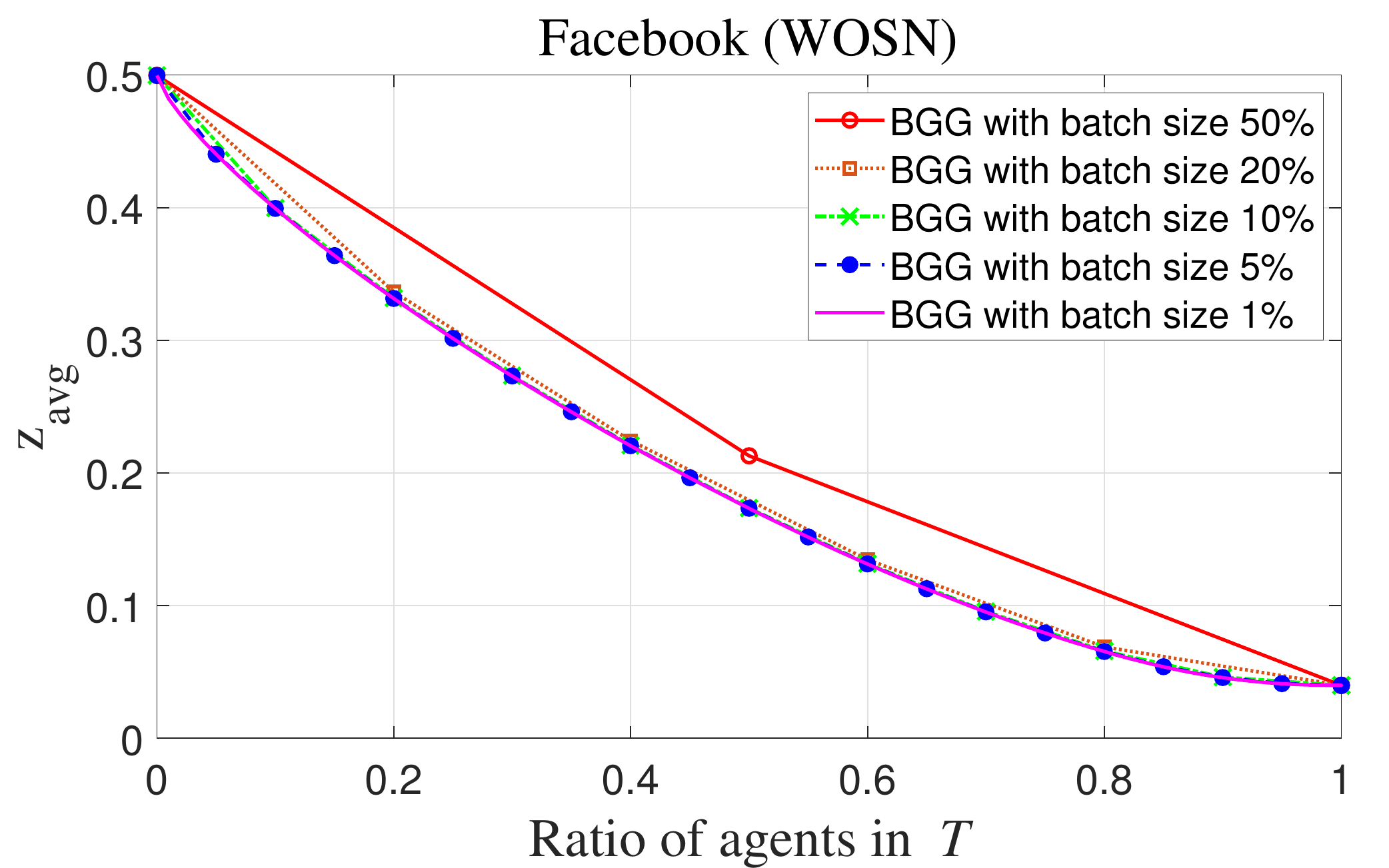}
    }
    \subfigure[]{
        \includegraphics[width=2.65in]{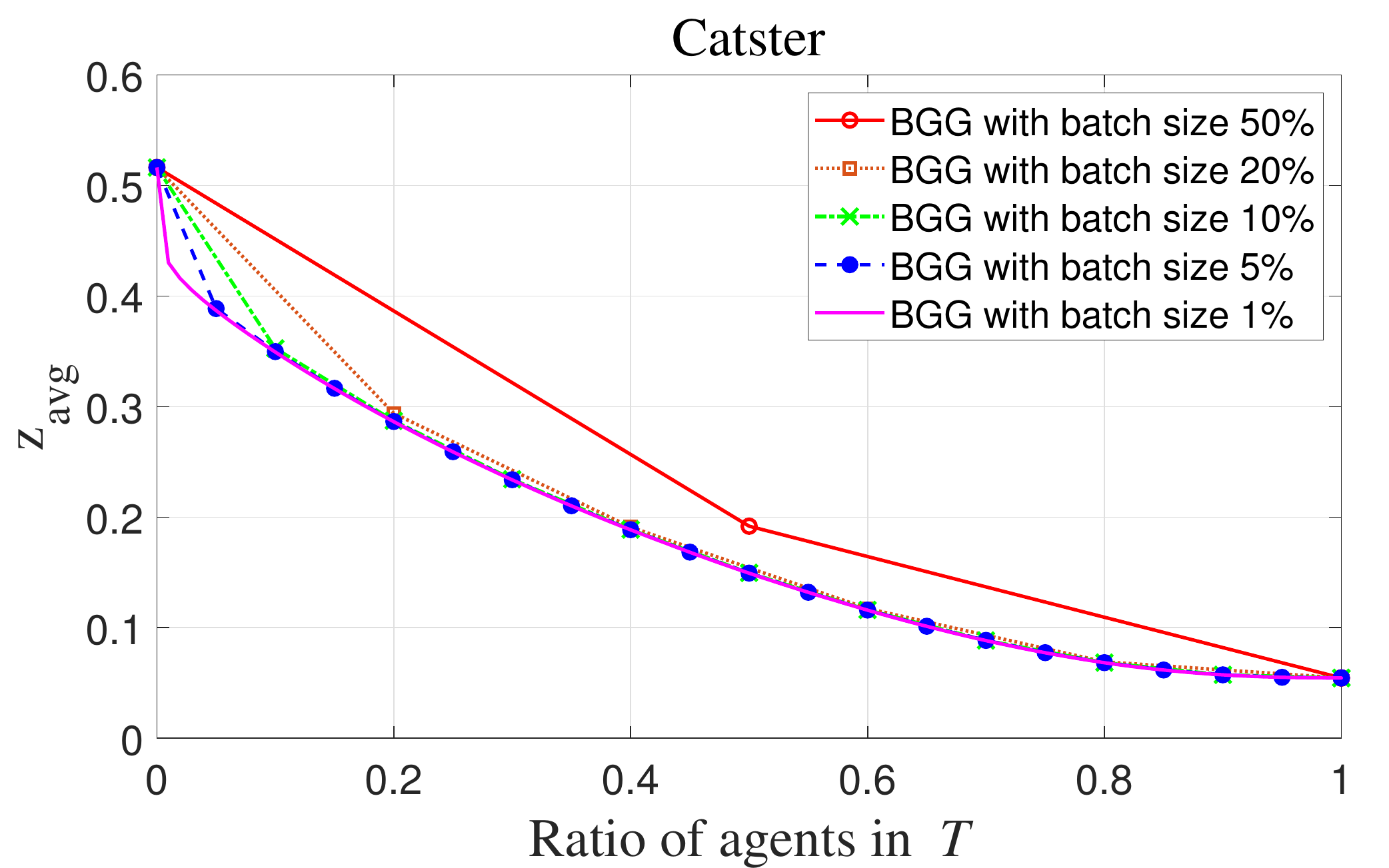}
    }
    \subfigure[]{
        \includegraphics[width=2.65in]{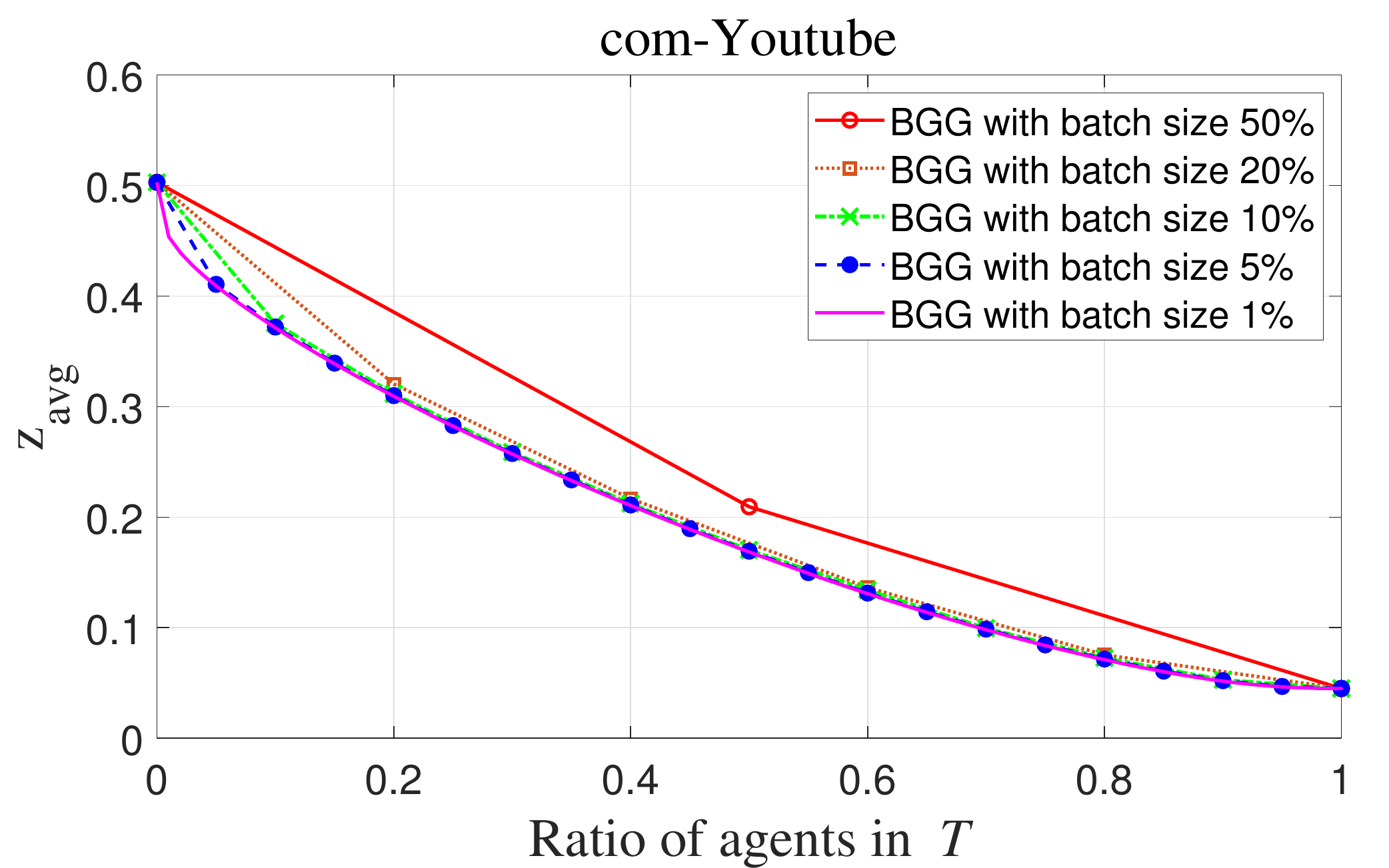}
    }
    \subfigure[]{
        \includegraphics[width=2.65in]{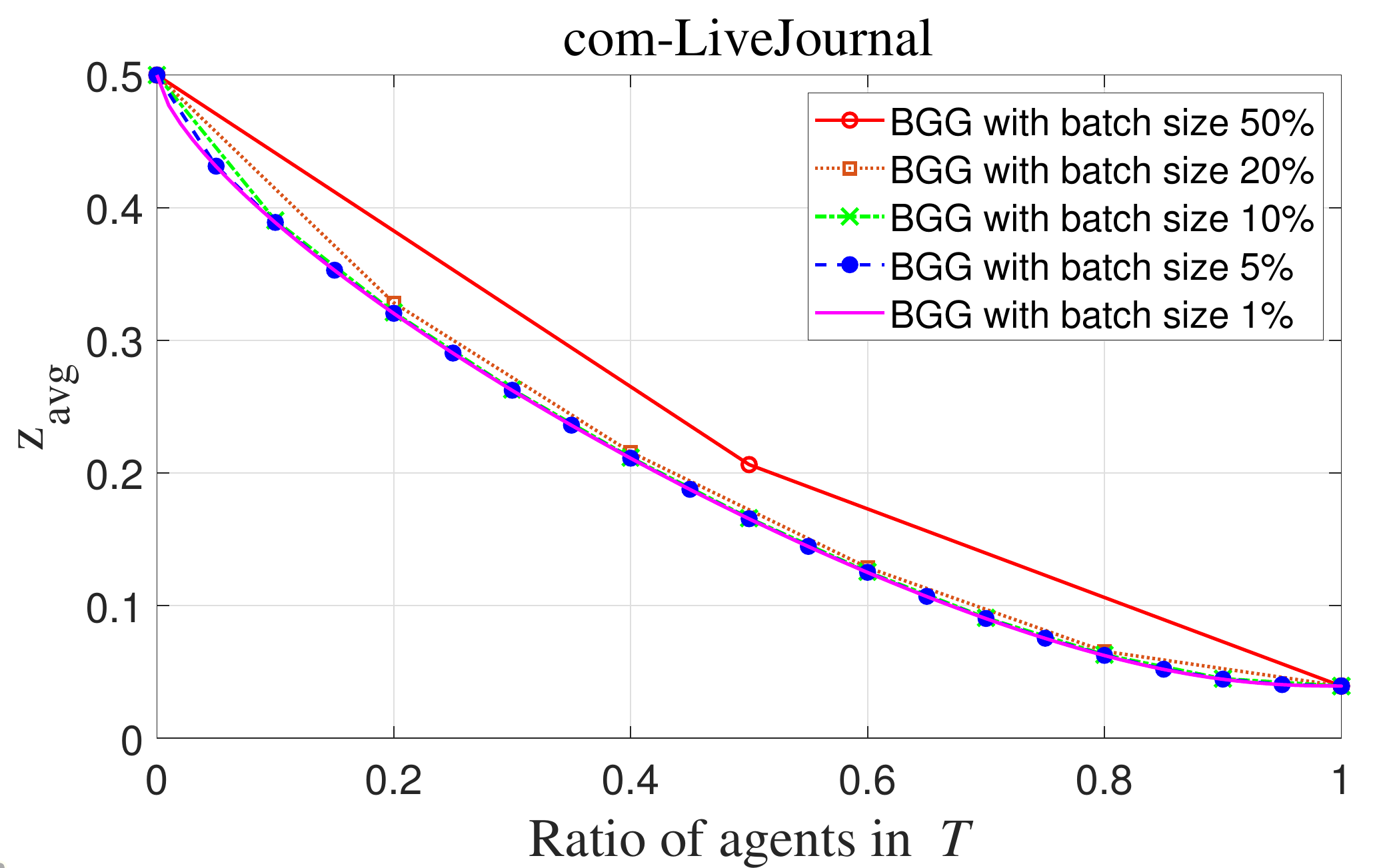}
    }
    \caption{Average Equilibrium Opinion on Large Networks.}
    \label{fig:large}
\end{figure}

When the budget is proportional to the number $n$ of agents,
if we use constant batch size, the number of times to pick a batch would be $O(n)$ which can be too many on large scale network.
One way to solve this problem is to choose the batch size proportional to the number of agents (or the budget), e.g. $1\%n$.
Then the number of times to pick a batch would become $O(1)$.

Figure \ref{fig:large} gives the average equilibrium opinion when using Batch Gradient Greedy (BGG) with different batch sizes on large graphs.
We can see that the batch sizes from 1\% to 10\%  get similarly good performance.
The results of batch size 20\% is slightly worse but acceptable.
While the batch size 50\% performs the worst among them and is not a good choice.
Thus, it is recommended to select 10\% of the number of agents (or the budget) or less as the batch size to balance the speed and performance.

\begin{figure}[h]
    \centering
    \subfigure[]{
        \includegraphics[width=2.65in]{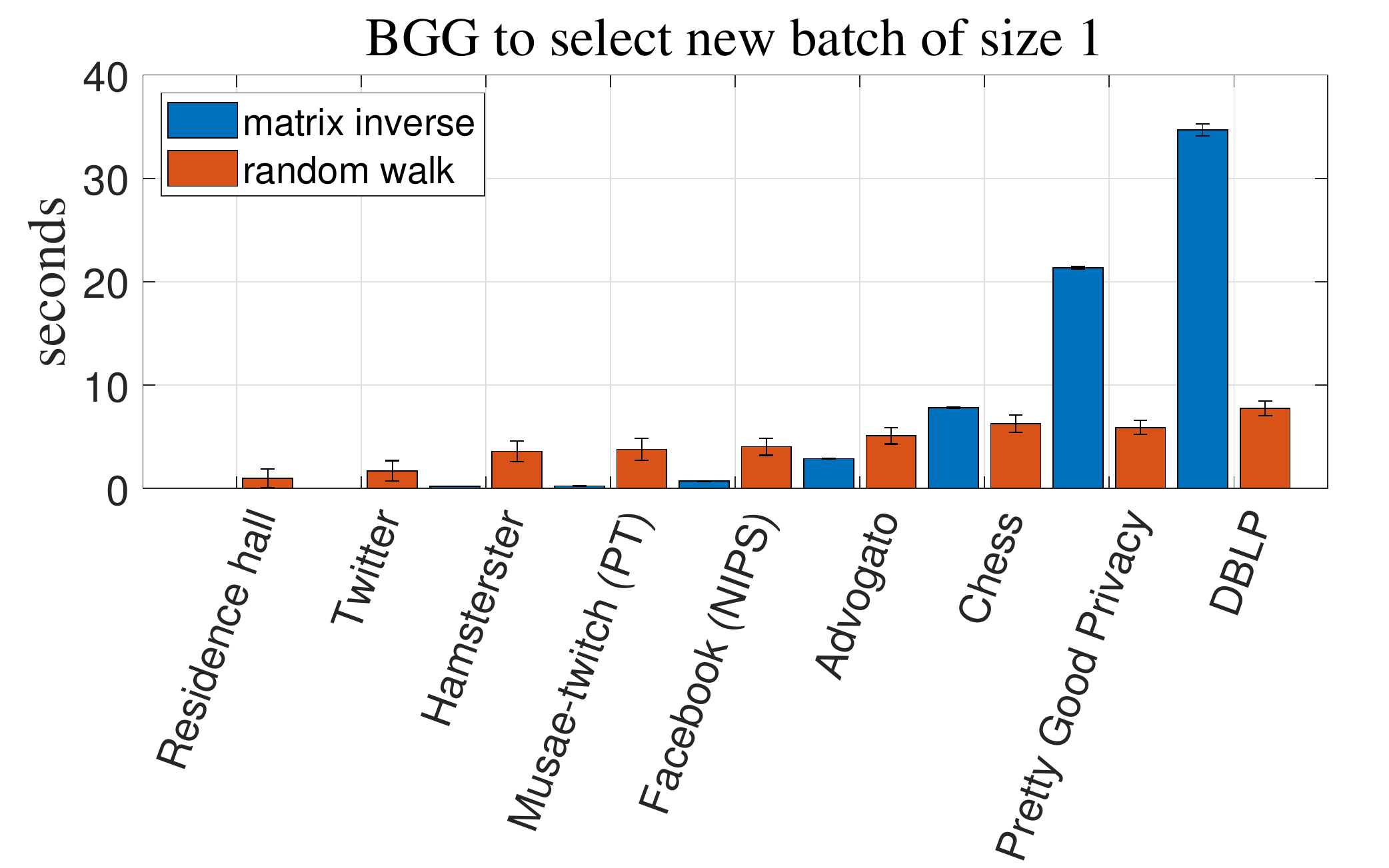}
		}
    \subfigure[]{
        \includegraphics[width=2.65in]{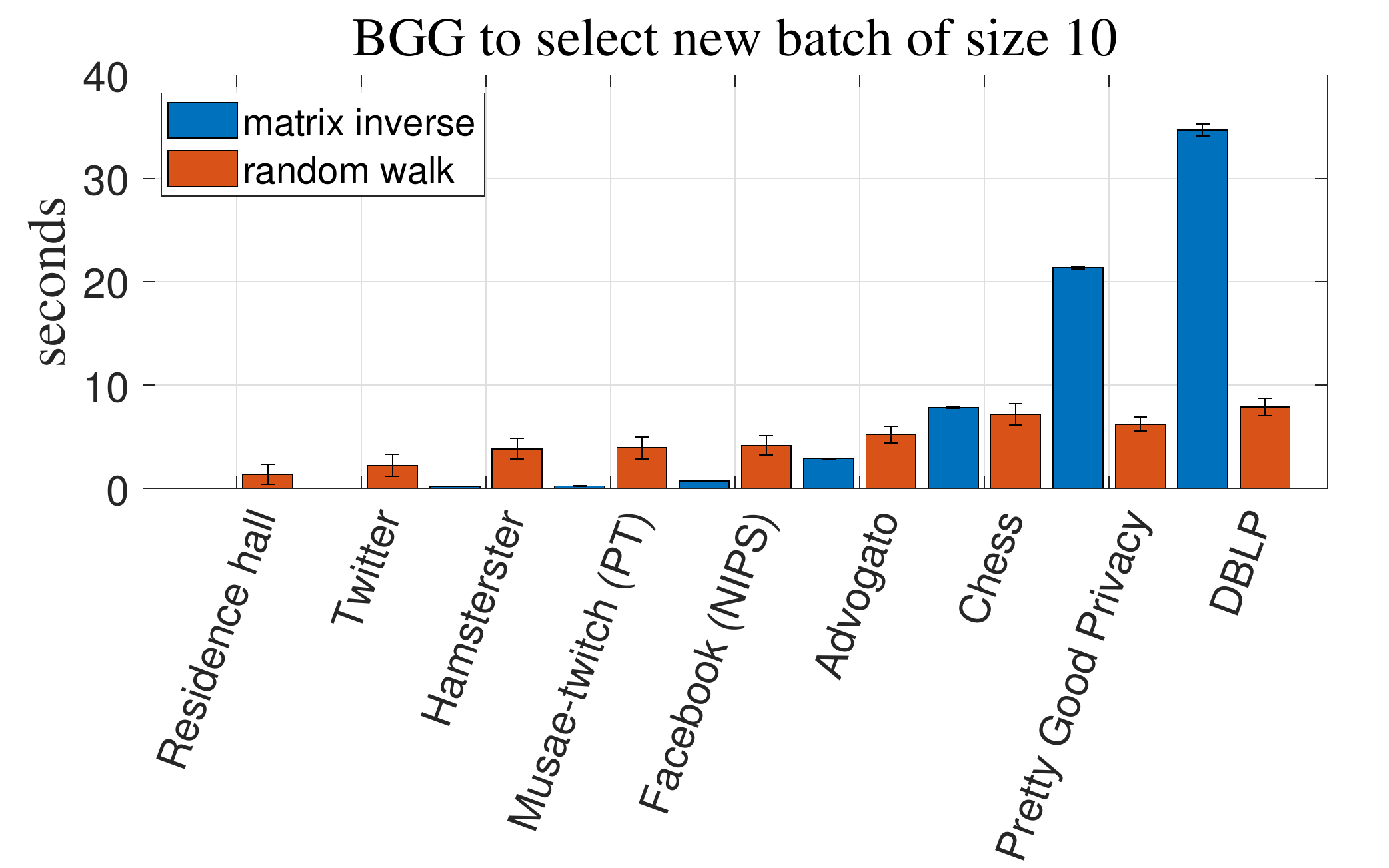}
    }
    \subfigure[]{
        \includegraphics[width=2.65in]{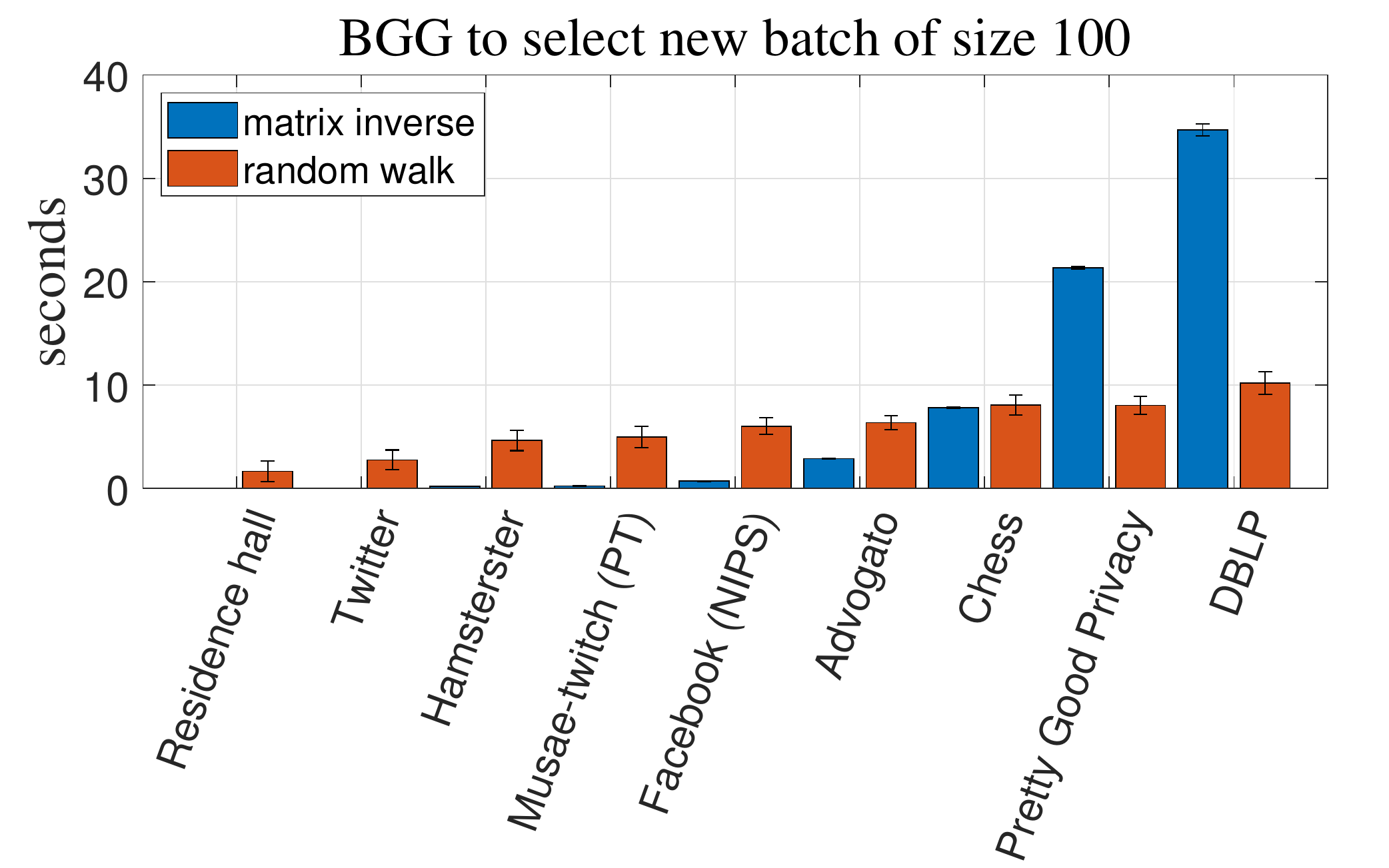}
    }
    \subfigure[]{
        \includegraphics[width=2.65in]{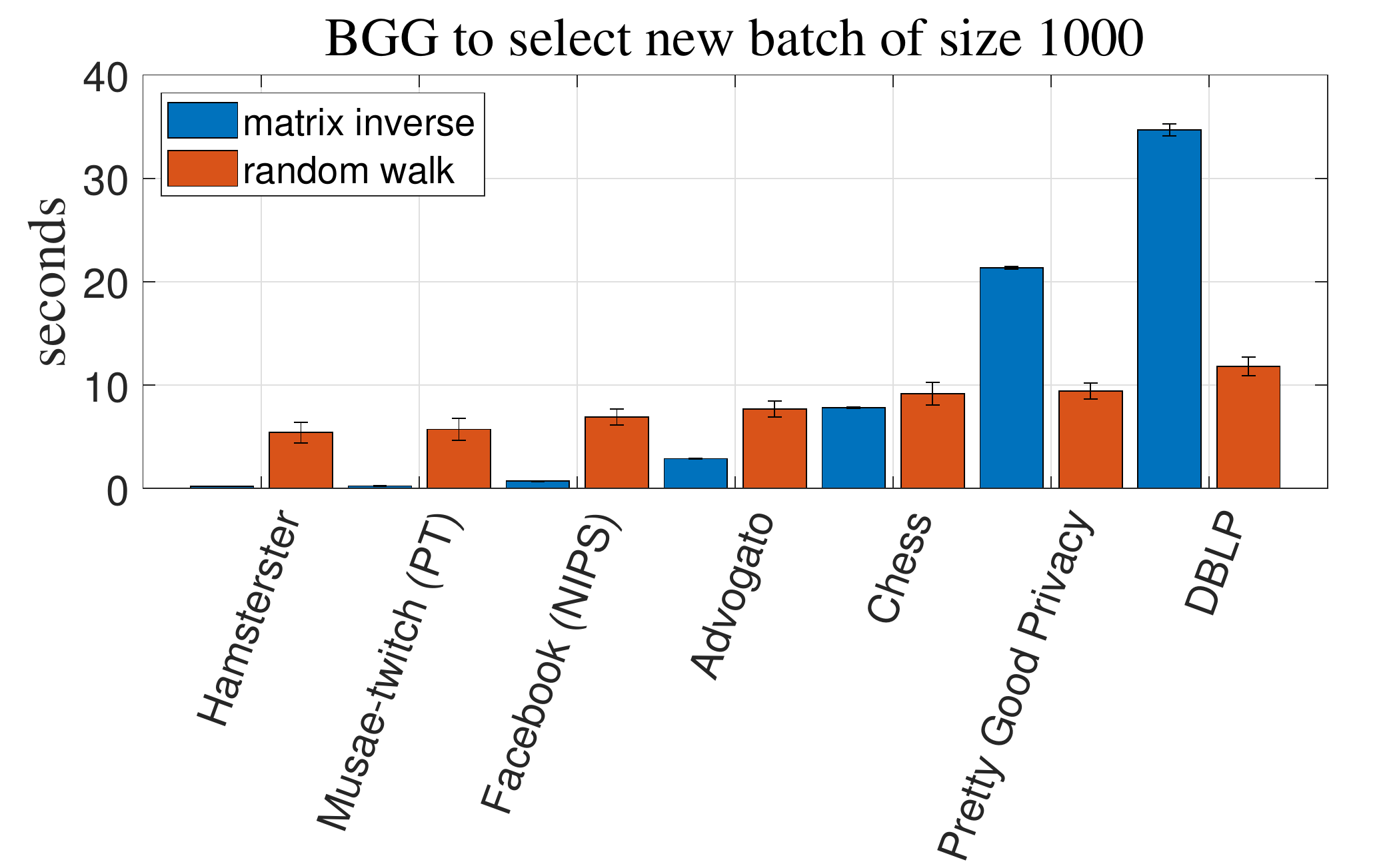}
    }
    \subfigure[]{
        \includegraphics[width=2.65in]{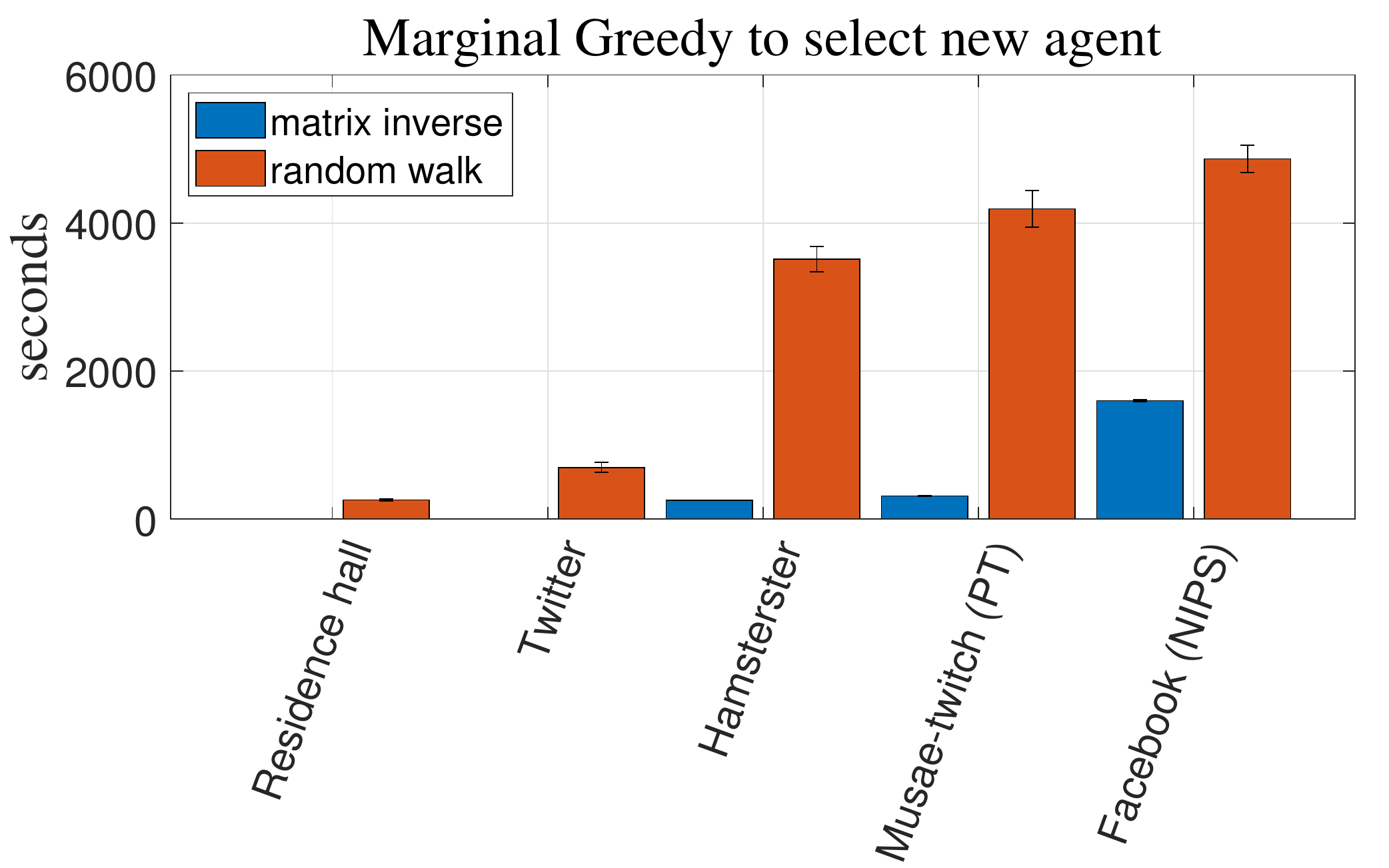}
    }
    \caption{Average and Standard Deviation of Running Time to Select New Agent or Batch.}
    \label{fig:time}
\end{figure}

\subsection{Running Time Comparison}

We compare the actual running time of Marginal Greedy and Batch Gradient Greedy using different methods to compute (or approximate) the equilibrium on only small networks,
since the results on larger networks are too time-consuming to collect.

Figure \ref{fig:time} (a) to (d) show the running time of Batch Gradient Greedy (BGG).
We compare the efficiency of computing the equilibrium opinion vector $z(\alpha)$ using the matrix inverse $z(\alpha) = [I - (I - A)P]^{-1} A s$ and approximating $z(\alpha)$ using the random walk recurrence $z^{(0)} \in [0,1]^V$
and $z^{(t+1)} := A s + (I - A) P z^{(t)}$, where $A = \Diag(\alpha)$.
For each network, we run 30 different random setups of $s, P, u, l$ and $\alpha^{(0)}$.
Then we record the running time to select the first batch of different batch sizes in each setup and give the average and standard deviation in the bar graphs.
Observe that the matrix inverse is faster for networks with less than a few thousands agents.
But roughly starting from Chess, the matrix inverse requires more time to select a new batch than the random walk recurrence, which implies that the complexity of matrix inverse is much higher.

While in Figure \ref{fig:time} (e), we show the running time results using Marginal Greedy.
We run only one random setup of $s, P, u, l$ and $\alpha^{(0)}$ on the same network and collect the running time to select the first 30 agents.
Then we report the average and standard deviation of running time to select a new agent.
Together with Figure \ref{fig:time} (a) to (d), we see that Marginal Greedy is significantly more time-consuming than Batch Gradient Greedy since Marginal Greedy has to compute the equilibrium opinion vector for adding each candidate before selecting the best one.
Considering they have similar performance as shown in Section \ref{sec:agent_selection_comp}, Batch Gradient Greedy with proper batch size is much more efficient than Marginal Greedy.

\subsection{Resistance Generation From Power Law Distribution}
\begin{figure}[h]
    \centering
    \subfigure[]{
        \includegraphics[width=2.65in]{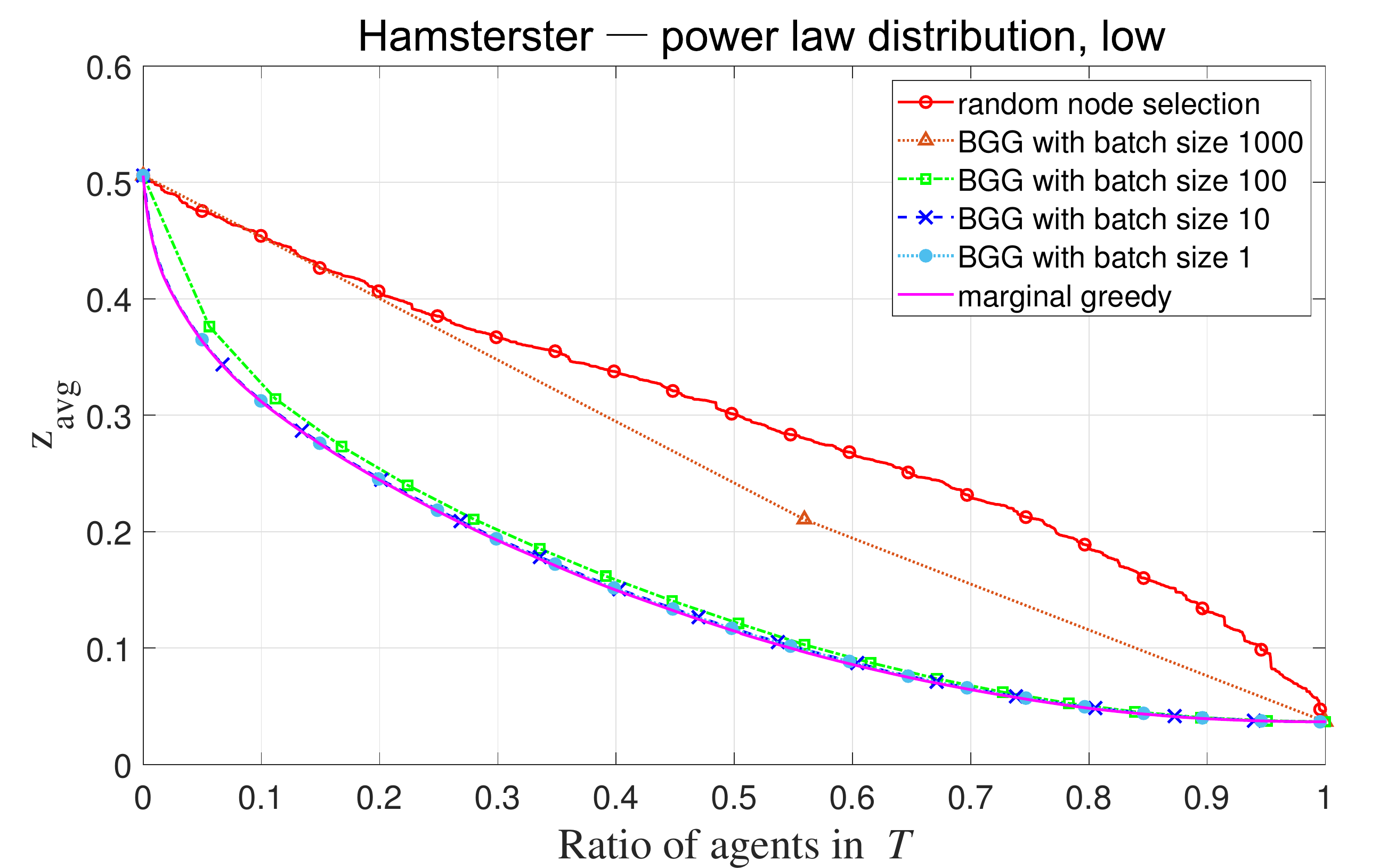}
		}
    \subfigure[]{
        \includegraphics[width=2.65in]{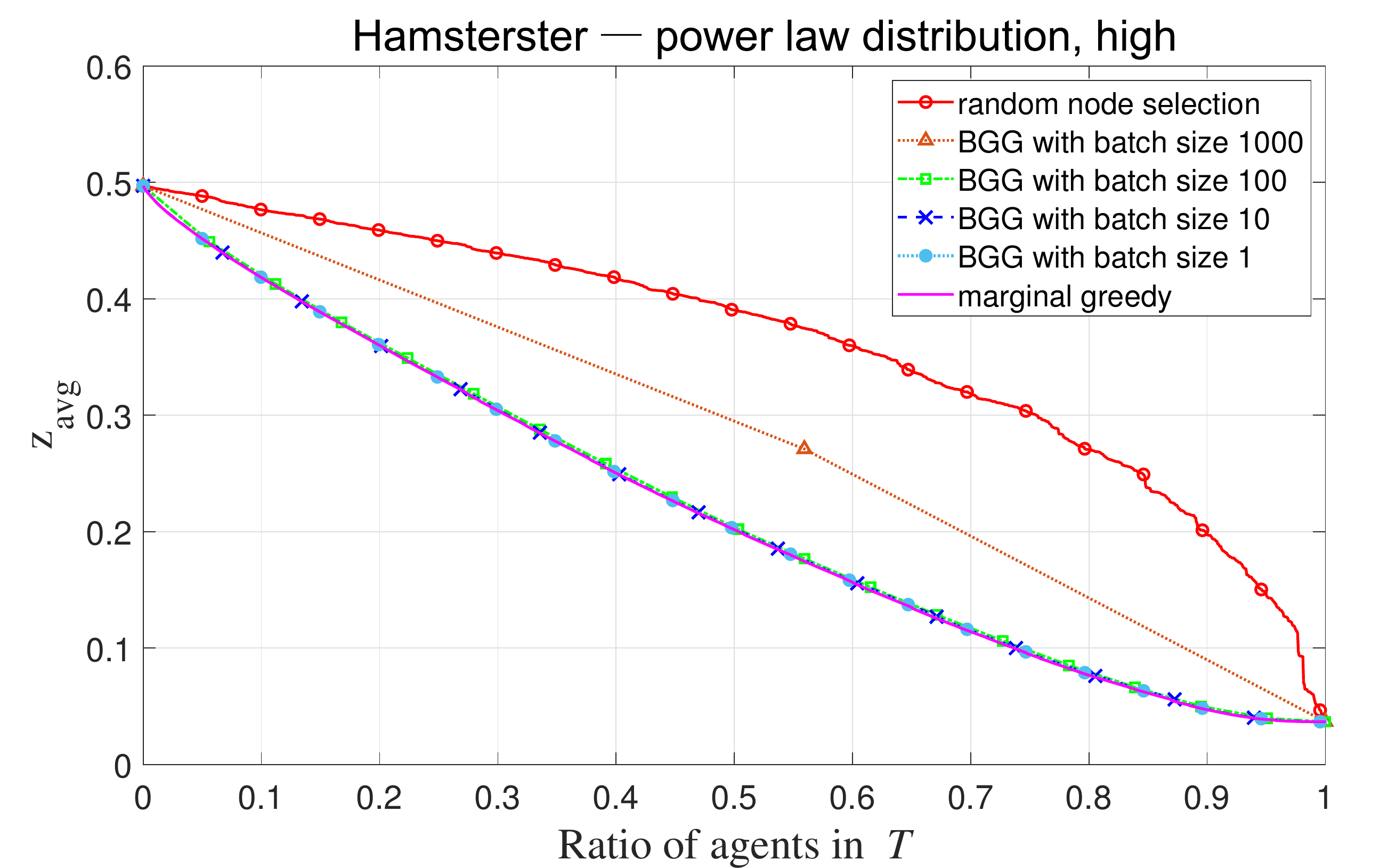}
    }
    \caption{Average Equilibrium Opinion on Hamsterster with $\alpha^{(0)}$ Generated from Power Law Distribution.}
    \label{fig:pl_small}
\end{figure}

We run experiments with the initial resistance vector $\alpha^{(0)}$ generated from the power law distribution instead of the uniform distribution to see how the heuristic algorithms respond to different distributions.
Particularly, each coordinate $\alpha^{(0)}_i$ is independently generated from $[l_i, u_i]$ with probability density function $f(x) = A x^{-2}$, where $A$ is the normalization constant.
Observe that the initial resistance vector generated in this way would have most of its coordinates being low values.
Figure \ref{fig:pl_small} (a) shows the results using the above setup (denoted as power law distribution, low) on Hamsterster where $s, P, u, l$ are unchanged.
We also run experiments using a resistance vector with most of its coordinates having high values.
For each coordinate $\alpha^{(0)}_i$, we first independently generate a $\alpha_i$ from $[l_i, u_i]$ with probability density function $f(x) = A x^{-2}$
and then compute $\alpha^{(0)}_i = u_i - \alpha_i + l_i$.
The corresponding results (denoted as power law distribution, high) on Hamsterster are given in Figure \ref{fig:pl_small} (b).
We can see that Marginal Greedy and Batch Gradient Greedy with batch size 1 still have similar performance under different resistance distributions.

\begin{figure}[h]
    \centering
    \includegraphics[width=2.65in]{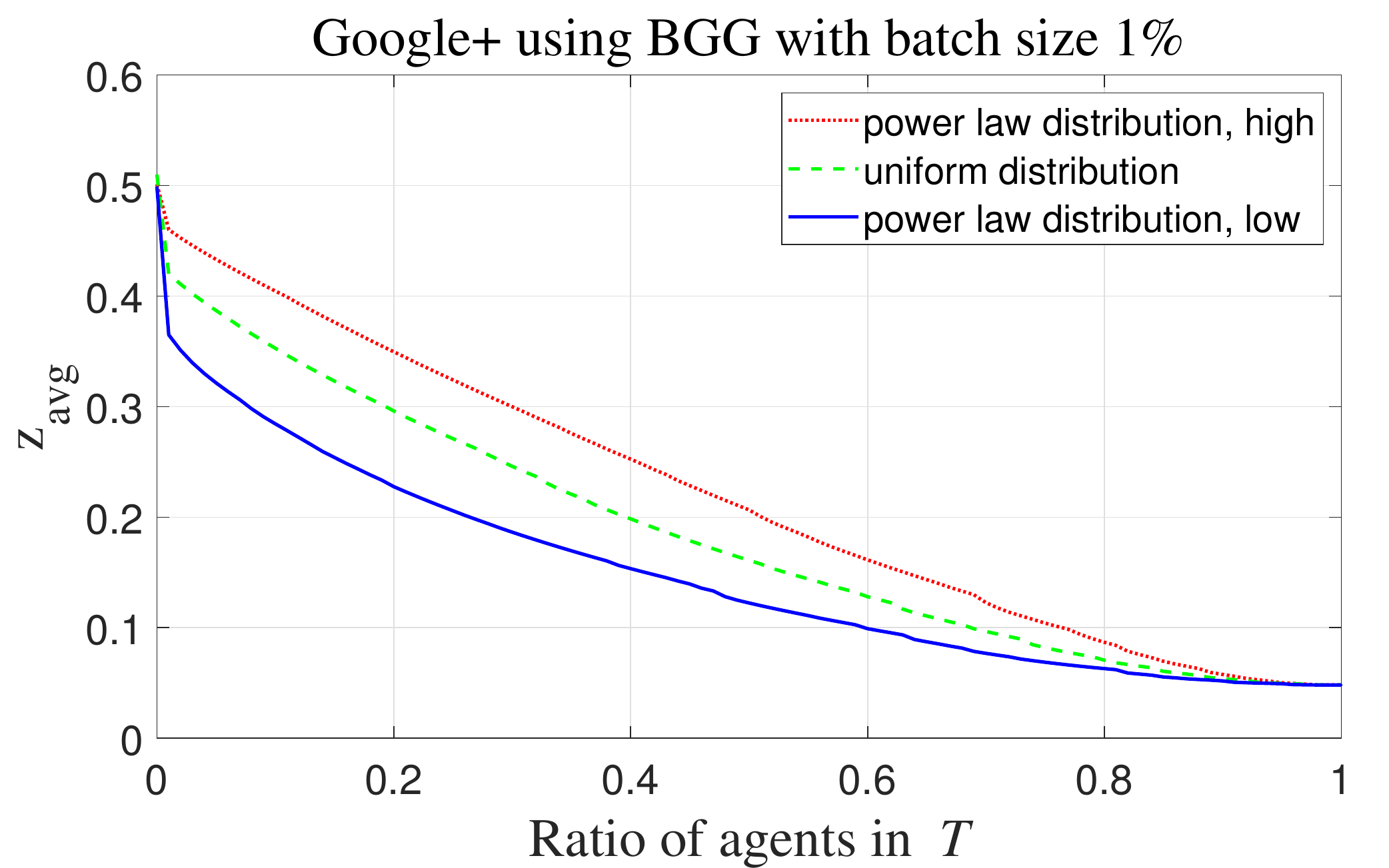}
    \caption{Average Equilibrium Opinion on Google+ with $\alpha^{(0)}$ Generated from Different Distributions.}
    \label{fig:pl_large}
\end{figure}

Figure \ref{fig:pl_large} gives the average equilibrium opinion on Google+ with initial resistance vector $\alpha^{(0)}$ generated from different distributions.
We only run experiments with batch size 1\% since the performance is good enough.
We can conclude that under the same $s, P, u, l$ and budget, the obtained average equilibrium opinion would be higher if the agents of the network tend to have higher resistance.
\section{Technical Proofs}
\label{sec:proofs}

\begin{lemma}(Lemma~\ref{le:PM_M} restated)
    Given ${K}\subsetneq V$ and $\alpha \in (0,1)^V$, let $A := \Diag(\alpha)$
		and recall that $P$ is the irreducible interaction matrix.
		Then, the inverse $M=[I-(I-A_{-K})P]^{-1}$ exists,
    and every entry of $M$ is positive.
		 Moreover, for each $k \in V$, 
			define $a_k = 0$ if $k \in {K}$, otherwise $a_k = \alpha_k$.
		Then, we have:

		\begin{enumerate}
        \item $(PM)_{kk}=\frac{M_{kk}-1}{1-a_k} > 0$;
        \item $(PM)_{kj}=\frac{M_{kj}}{1-a_k} > 0$, for each $j\not=k$.
    \end{enumerate}
\ignore{
    \begin{enumerate}
        \item $M_{kk}=(1-a)(PM)_{kk}+1 \geq 1$,
        \item $M_{kj}=(1-a)(PM)_{kj}$ for each $j\not=k$.
    \end{enumerate}
}
\end{lemma}
\begin{proof}
Observe that $P$ corresponds to an irreducible random walk.
Hence,
$(I-A_{-K})P$ represents a diluted random walk, where at the beginning of each step,
the measure at nodes $i \notin K$ will suffer a factor of $1 - \alpha_i\in(0,1)$.
The irreducibility of the random walk $P$ means that every state is reachable from any state.
Hence, starting from any measure vector, eventually the measure at every node will tend to 0.
This means that $(I-A_{-K})P$ has eigenvalues with magnitude strictly less than 1.
Therefore, we can consider the following Neumann series of a matrix:
$$M=[I-(I-A_{-K})P]^{-1}=I + \sum_{k=1}^{\infty}[(I-A_{-K})P]^k,$$
which implies that the inverse $M$ exists, and every entry of $M$ is positive;
in particular, for every $k \in V$, $M_{kk} > 1$.

By the definition of $M$, we have $[I-(I-A_{-{K}})P]M=I$.
We fix some $k \in V$.
By considering the $(k,k)$-the entry, i.e.,
the dot product between
the $k$-th row of $[I-(I-A_{-{K}})P]$
    and the $k$-th column of $M$,
		we have $$
    M_{kk}-\sum_{i \in V}(1-a_k)P_{ki}M_{ik}=1
    $$
where $a_k=0$ if $k\in {K}$, otherwise $a_k=\alpha_k<1$.
Hence, we have $
    (PM)_{kk}=\sum_{i \in V} P_{ki} M_{ik}=\frac{M_{kk}-1}{1-a_k}.
    $

Similarly, for  $j\not=k$, by considering the dot product between
    the $k$-th row of $[I-(I-A_{-{K}})P]$
    and the $j$-th column of $M$,
    we have
    $$
    M_{kj}-\sum_{i \in V} (1-a_k)P_{ki}M_{ij}=0.
    $$
		
		Hence, we have for $j\not=k$,
    $$
    (PM)_{kj}=\sum_{i \in V}P_{ki}M_{ij}=\frac{M_{kj}}{1-a_k},
    $$
as required.
\end{proof}

\begin{lemma}(Lemma~\ref{le:mar_mon} restated)
Referring to Lemma~\ref{le:same_sign}
.
For any $\alpha \in (0,1)^V$ and any $i \in V$,
denote $M=[I-(I-A_{-\{i\}})P]^{-1}$.
Then, $\frac{\partial f(\alpha)}{\partial \alpha_i}$ has the same sign in $\{-, 0, +\}$
as $s_i -\sum_{j\not=i}M_{ij}\alpha_j s_j$,
which is independent of $\alpha_i$.
\end{lemma}

\begin{proof}
    Using the Sherman-Morrison formula, we consider
    \begin{align}
    \begin{split}
    X:=&[I-(I-A)P]^{-1}
    =[I-(I-A_{-\{i\}}-\alpha_i\mathbf{e}_i\mathbf{e}_i^T)P]^{-1}
    \\
    =&[I-(I-A_{-\{i\}})P+\alpha_i\mathbf{e}_i\mathbf{e}_i^T P]^{-1}
    \\
    =&
    M-\frac{\alpha_i}{1+\alpha_i\mathbf{e}_i^T PM\mathbf{e}_i}
    M\mathbf{e}_i\mathbf{e}_i^T PM.
    \end{split}
    \nonumber
    \end{align}
		
    Observe that
    $\mathbf{e}_i^T PM\mathbf{e}_i=(PM)_{ii}$ and
    $(M\mathbf{e}_i\mathbf{e}_i^T PM)_{ij}=M_{ii}(PM)_{ij}$
    for each $j\in V$.
    Then, by Lemma \ref{le:PM_M} with ${K}=\{i\}$, we have 
    \begin{align}
    \begin{split}
    X_{ii}=&
    M_{ii}-\frac{\alpha_iM_{ii}(PM)_{ii}}{1+\alpha_i(PM)_{ii}}
    =M_{ii}-\frac{\alpha_iM_{ii}(M_{ii}-1)}{1+\alpha_i(M_{ii}-1)}
    \\
    =&
    \frac{(1-\alpha_i+\alpha_iM_{ii}-\alpha_iM_{ii}+\alpha_i)M_{ii}}{1-\alpha_i+\alpha_iM_{ii}}
    =\frac{M_{ii}}{1-\alpha_i+\alpha_iM_{ii}};
    \end{split}
    \nonumber
    \end{align}
    and for $j\not=i$, 
    \begin{align}
    \begin{split}
    X_{ij}=&
    M_{ij}-\frac{\alpha_iM_{ii}(PM)_{ij}}{1+\alpha_i(PM)_{ii}}
    =M_{ij}-\frac{\alpha_iM_{ii}M_{ij}}{1+\alpha_i(M_{ii}-1)}
    \\
    =&\frac{(1-\alpha_i+\alpha_iM_{ii}-\alpha_iM_{ii})M_{ij}}{1-\alpha_i+\alpha_iM_{ii}}
    =\frac{(1-\alpha_i)M_{ij}}{1-\alpha_i+\alpha_iM_{ii}}.
    \end{split}
    \nonumber
    \end{align}
    By Lemma \ref{le:same_sign}, we know $\frac{\partial f(\alpha)}{\partial \alpha_i}$ and $s_i - z_i(\alpha)$ have the same sign in $\{-, 0, +\}$.
    Recall that $z(\alpha) = [I-(I-A)P]^{-1} A s = XAs$.
    Applying the above results, we have
    \begin{align}
    \begin{split}
	s_i - z_i(\alpha)
	=&s_i-\sum_{j\in V}X_{ij}\alpha_j s_j
    \\
    =&s_i-\frac{M_{ii}\alpha_i s_i}{1-\alpha_i+\alpha_iM_{ii}}
    -\sum_{j\not= i}
    \frac{(1-\alpha_i)M_{ij}\alpha_j s_j}{1-\alpha_i+\alpha_iM_{ii}}
    \\
    =&
    \frac{s_i(1-\alpha_i+\alpha_iM_{ii})-M_{ii}\alpha_is_i
    -(1-\alpha_i)\sum_{j\not= i}M_{ij}\alpha_j s_j}
    {1-\alpha_i+\alpha_iM_{ii}}
    \\
    =&
    \frac{(1-\alpha_i)(s_i-\sum_{j\not= i}M_{ij}\alpha_js_j)}
    {1-\alpha_i+\alpha_iM_{ii}}.
    \end{split}
    \nonumber
    \end{align}
		
    Since $\alpha_i\in (0,1)$, we conclude that
    $\frac{1-\alpha_i}{1-\alpha_i+\alpha_iM_{ii}}>0$.
    Thus $\frac{\partial f(\alpha)}{\partial \alpha_i}$, $s_i - z_i(\alpha)$ and $s_i-\sum_{j\not= i}M_{ij}\alpha_js_j$ have the same sign.
\end{proof}

\begin{lemma}(Lemma~\ref{le:M2R} restated)
For any $i,k\in V$ such that $i\not= k$, let $M=[I-(I-A_{-\{i\}})P]^{-1}$
and $R=[I-(I-A_{-\{i,k\}})P]^{-1}$. Then for any $j \in V$,
we have
\begin{enumerate}
\item $M_{jk}=\frac{R_{jk}}{1+\alpha_kR_{kk}-\alpha_k}$,
\item $M_{jh}=R_{jh}-\frac{\alpha_{k}R_{jk}R_{kh}}{1+\alpha_kR_{kk}-\alpha_k}$, for each $h\not=k$.
\end{enumerate}

In particular, the quantity in Lemma~\ref{le:mar_mon} can be rewritten as follows:
$$
s_i-\sum_{j\not=i}M_{ij}\alpha_js_j
=s_i-\sum_{j\not=i,k}R_{ij}\alpha_js_j
-\frac{\alpha_kR_{ik}}{1+\alpha_kR_{kk}-\alpha_k}
(s_k-\sum_{j\not=i,k}R_{kj}\alpha_js_j).
$$
\end{lemma}

\begin{proof}
Using the Sherman-Morrison formula, we have
\begin{align}
\begin{split}
M&=[I-(I-A_{-\{i,k\}})P+\alpha_k\mathbf{e}_k\mathbf{e}_k^T P]^{-1}
\\&
=R-\frac{\alpha_k}{1+\alpha_k\mathbf{e}_k^T PR\mathbf{e}_k}
R\mathbf{e}_k\mathbf{e}_k^T PR
\end{split}
\nonumber
\end{align}
We can compute that
$\mathbf{e}_k^T PR\mathbf{e}_k=(PR)_{kk}$
and
$(R\mathbf{e}_k\mathbf{e}_k^T PR)_{jh}=R_{jk}(PR)_{kh}$ for $j,h\in V$.
Then we have
$$
M_{jh}=R_{jh}-\frac{\alpha_kR_{jk}(PR)_{kh}}{1+\alpha_k(PR)_{kk}}.
$$
By Lemma \ref{le:PM_M}, we obtain
$$
M_{jh}=R_{jh}-\frac{\alpha_{k}R_{jk}R_{kh}}{1+\alpha_kR_{kk}-\alpha_k}
\text{  for $j,h\in V$ and $h\not=k$,}
$$
and
$$M_{jk}
=R_{jk}-\frac{\alpha_{k}R_{jk}(R_{kk}-1)}{1+\alpha_kR_{kk}-\alpha_k}
=\frac{R_{jk}}{1+\alpha_kR_{kk}-\alpha_k}
\text{   for $j\in V$.}
$$
as required.
\end{proof}

\begin{lemma}(Lemma~\ref{le:col_max} restated)
Suppose $\alpha \in (0,1)^V$,
recall that $A_{-\{i,k\}} := \Diag(\alpha_{-\{i,k\}})$,
and $P$ corresponds to an irreducible interaction matrix.
For any $i, k\in V$ such that $i\not=k$, let $R=[I-(I-A_{-\{i,k\}})P]^{-1}$,
then $R_{ii}=\max_{j\in V}R_{ji}$.
Moreover, $R_{ii}=R_{ki}$ if and only if
$P_{kk} + P_{ki}=1$.
\end{lemma}

\begin{proof}
We have $[I-(I-A_{-\{i,k\}})P]R=I$.
By considering the dot product between
each row of $[I-(I-A_{-\{i,k\}})P]$ and column $i$ of $R$,
we have
$$
R_{ii}-\sum_{h \in V}P_{ih}R_{hi}=1,
R_{ki}-\sum_{h \in V}P_{kh}R_{hi}=0,
$$
$$
\text{and }R_{ji}-\sum_{h \in V}(1-\alpha_j)P_{jh}R_{hi}=0,
\text{   for } j\not=i,k.
$$
After rearranging,
we have
$$
R_{ii}=1+\sum_{h \in V}P_{ih}R_{hi},
R_{ki}=\sum_{h \in V}P_{kh}R_{hi},
$$
$$
\text{and }R_{ji}=(1-\alpha_j)\sum_{h \in V}P_{jh}R_{hi},
\text{   for } j\not=i,k.
$$
Now it suffices to show that for $j \neq i,k$, the above $R_{ji}$
cannot be the maximum among them, and $R_{ki}\leq R_{ii}$.

First, we show that $R_{ji}$
cannot be the maximum.
Since $\sum_{h \in V}P_{jh}=1$ and $\alpha_j\in (0,1)$,
we have
$$
R_{ji}=(1-\alpha_j)\sum_{h \in V} P_{jh}R_{hi}
\leq(1-\alpha_j)\max_{h \in V} R_{hi}
<\max_{h \in V}R_{hi}.
$$
Thus, $R_{ji}$ cannot be the maximum.

Next, we show that $R_{ki}\leq R_{ii}$ by contradiction.
Suppose $R_{ki}>R_{ii}$,
then $R_{ki}$ is the unique maximum in the $i$-th column of $R$.
Since $\sum_{h \in V} P_{kh} = 1$ and
$R_{ki}=\sum_{h \in V}P_{kh}R_{hi}$,
it must be the case that $P_{kk} = 1$.
This means $P$ corresponds to a random walk with absorbing state~$k$,
which contradicts that $P$ is irreducible.
Therefore, we have $R_{ki}\leq R_{ii}$, and hence $R_{ii}=\max_{h\in V}R_{hi}$.

Observe that we already know $R_{ji} < R_{ii}$ for $j \neq i,k$,
and $R_{ki}=\sum_{h \in V}P_{kh}R_{hi}$.
Hence, $R_{ki} = R_{ii}$ implies that $P_{kk} + P_{ki} = 1$.

Conversely, $P_{kk} + P_{ki} = 1$ implies that
$R_{ki} = P_{kk} R_{ki} + P_{ki} R_{ii}$.
As argued above, we must have $P_{kk} \neq 1$, which implies
$R_{ki} = R_{ii}$.
\end{proof}
\section{Conclusion and Future Work}
\label{sec:conclusion}

In this work we have introduced a novel formulation of social influence,
that focuses on interventions at the level of susceptibility using a well-established opinion dynamics model.  We  give a solid theoretical analysis
of the unbudgeted variant of the opinion susceptibility problem,
and designed scalable local search algorithms that
can solve the problem optimally on graphs with millions of nodes. We also prove that the budgeted variant is NP-hard, and provide scalable heuristics that we evaluate experimentally.  We believe that our techniques for the unbudgeted variant will lead to insights for  the analysis of the budgeted variant of the problem.  We leave the task of providing theoretical guarantees for greedy algorithms on the budgeted variant as future work.

{
\bibliographystyle{abbrv}
\bibliography{opinion,dihyper}
}

\end{document}